\documentclass[11pt]{amsart}

\usepackage{amsmath,amsxtra,amssymb,amsthm,amsfonts}
\usepackage{mathrsfs}
\usepackage[T1]{fontenc}
\usepackage[utf8]{inputenc}
\usepackage{mathtools}   % loads »amsmath«
\usepackage{tikz}
\usetikzlibrary{arrows, automata}
\usepackage{color,hyperref}
\definecolor{dblue}{rgb}{0, 0, 0.72}
\numberwithin{equation}{section}
\newtheorem{lemma}{Lemma}[section]
\newtheorem{prop}[lemma]{Proposition}
\newtheorem{theorem}[lemma]{Theorem}

\newtheorem{cor}[lemma]{Corollary}

\newtheorem{conj}[lemma]{Conjecture}

\newtheorem{prob}[lemma]{Problem}
\newtheorem{rem}[lemma]{Remark}
\newtheorem{remark}[lemma]{Remark}

\newtheorem{example}[lemma]{Example}
\newtheorem{definition}[lemma]{Definition}
\newtheorem{corollary}[lemma]{Corollary}

\newcommand{\re}{\begin{rem}\rm}
  \newcommand{\mar}{\end{rem}}

\newcommand{\kla}{\left ( }
\newcommand{\mer}{\right ) }

\newcommand{\ee }{\mathrm{I}\!\!1}

\renewcommand{\for}{\begin{eqnarray*}}

\newcommand{\mel}{\end{eqnarray*}}

\newcommand{\kl}{\pl \le \pl}
\newcommand{\gl}{\pl \ge \pl}

\newcommand{\lel}{\pl = \pl}

\newcommand{\BB}{{\mathbb B}}
\renewcommand{\L}{\mathcal{L}}

\newcommand{\nz}{{\mathbb N}}
\newcommand{\nen}{n \in \nz}

\newcommand{\rz}{{\mathbb R}}
\newcommand{\zz}{{\mathbb Z}}

\newcommand{\Mz}{{\mathbb M}}

\newcommand{\cz}{{\mathbb C}}

\newcommand{\ten}{\otimes}

\DeclareMathOperator{\dom}{dom}

\DeclareMathOperator{\Hs}{H}
\DeclareMathOperator{\Vs}{V}
\DeclareMathOperator{\Ks}{K}

\DeclareMathOperator{\Ric}{Ric}

\DeclareMathOperator{\fix}{fix}

\DeclareMathOperator{\Tr}{Tr}
\DeclareMathOperator{\tr}{tr}

\DeclareMathOperator{\Res}{Res_{\small{W}}}

\DeclareMathOperator{\CLSI}{CLSI}
\DeclareMathOperator{\MLSI}{MLSI}
\DeclareMathOperator{\MpSI}{MpSI}
\DeclareMathOperator{\CpSI}{C_pSI}

\DeclareMathOperator{\Lin}{\text{\fontsize{2.5}{4}\selectfont {Lin}}}

\newcommand{\pl}{\hspace{.1cm}}

\newcommand{\hhz}{\vspace{0.3cm}}

\newcommand{\qd}{\end{proof}\vspace{0.5ex}}

\newcommand{\Om}{\Omega}
\newcommand{\om}{\omega}

\newcommand{\al}{\alpha}

\newcommand{\si}{\sigma}

\newcommand{\la}{\lambda}
\newcommand{\eps}{\varepsilon}

\newcommand{\E}{{\mathcal E}}
\newcommand{\A}{{\mathcal A}}

\newcommand{\MM}{{\mathcal M}}

\newcommand{\LSI}{\operatorname{LSI}}

\newcommand{\N}{{\mathcal N}}

\newcommand{\T}{\mathcal T}

\newcommand{\U}{{\mathcal U}}
\newcommand{\B}{{\mathcal B}}

\newcommand{\pf}{\begin{proof}}

\newcommand{\be}{\left|{\atop}}

\newcommand{\xspace}{\hbox{\kern-2.5pt}}
\newcommand{\xyspace}{\hbox{\kern-1.1pt}}

\newcommand\ltriple{ltriple}

\newcommand{\ssubset} {\!\!\subset\! \!}

\allowdisplaybreaks

\definecolor{LightGray}{rgb}{0.94,0.94,0.94}
\definecolor{VeryLightBlue}{rgb}{0.9,0.9,1}
\definecolor{LightBlue}{rgb}{0.8,0.8,1}
\definecolor{DarkBlue}{rgb}{0,0,0.6}
\definecolor{LightGreen}{rgb}{0.88,1,0.88}
\definecolor{MidGreen}{rgb}{0.6,1,0.6}
\definecolor{DarkGreen}{rgb}{0,0.6,0}
\definecolor{DarkGrreen}{rgb}{0,0.8,0}

\definecolor{VeryLightYellow}{rgb}{1,1,0.9}
\definecolor{LightYellow}{rgb}{1,1,0.6}
\definecolor{MidYellow}{rgb}{1,1,0.5}
\definecolor{DarkYellow}{rgb}{0.8,1,0.3}
\definecolor{VeryLightRed}{rgb}{1,0.9,0.9}
\definecolor{LightRed}{rgb}{1,0.8,0.8}
\definecolor{DarkRed}{rgb}{0.8,0.2,0}
\definecolor{DarkRedb}{rgb}{0.6,0.2,0}
\definecolor{DarkLila}{rgb}{0.8,0,1}
\definecolor{Beige}{rgb}{0.96,0.96,0.86}
\definecolor{Gold}{rgb}{1.,0.84,0.}
\definecolor{Goldb}{rgb}{0.7,0.3,0.5}
\definecolor{MyYellow}{rgb}{1.,0.84,0.8}

\newcommand{\lan}{\langle}
\newcommand{\ran}{\rangle}

\newcommand{\lgg}{\mathfrak{g}}

\def\mod{\,\, {\rm mod}\,\,}

\def\dom{\operatorname{dom}}

\def\Mm{\mathcal{M}}

\def\11{\mathbb{I}}

\usepackage{graphicx}

\DeclareRobustCommand\openone{\leavevmode\hbox{\small1\normalsize\kern-.33em1}}

\renewcommand{\be}{\begin{equation}}
	\renewcommand{\ee}{\end{equation}}
\newcommand{\bea}{\begin{eqnarray}}
	\newcommand{\eea}{\end{eqnarray}}
\newcommand{\beas}{\begin{eqnarray*}}
	\newcommand{\eeas}{\end{eqnarray*}}

\usepackage{amsmath}
\usepackage{ams math}
\usepackage{amssymb}
\usepackage{amsthm}
\usepackage{dsfont}
\usepackage{upgreek}
\usepackage{enumitem}
\usepackage{array}
\usepackage{ams fonts}
\usepackage{graphics}
\usepackage[utf8]{inputenc}
\newtheorem*{theorem*}{Theorem}
\newtheorem*{remark*}{Remark}
\newtheorem*{lemma*}{Lemma}
\newtheorem*{cor*}{Corollary}
\newtheorem*{note*}{Note}
\newtheorem*{prop*}{Proposition}
\newtheorem*{example*}{Example}

\renewcommand{\Ric}{\mbox{Ric}}
\newcommand{\Rc}{\mbox{Rc}}
\newcommand{\Rcu}{\ensuremath{\Rc_{\text{\fontsize{2.5}{4}\selectfont {U}}}}}
\newcommand{\Du}{\ensuremath{\Delta_{\text{\fontsize{2.5}{4}\selectfont {U}}}}}
\newcommand{\Dv}{\ensuremath{\Delta_{\text{\fontsize{2.5}{4}\selectfont {V}}}}}

\newcommand{\Lu}{\ensuremath{L_{\text{\fontsize{2.5}{4}\selectfont {U}}}}}
\newcommand{\Zu}{\ensuremath{Z_{\text{\fontsize{2.5}{4}\selectfont {U}}}}}
\newcommand{\Zv}{\ensuremath{Z_{\text{\fontsize{2.5}{4}\selectfont {V}}}}}

\renewcommand{\T}{\mbox{T}}

\newcommand{\trans}{\mbox{\scalebox{.5}{trans}}}

\renewcommand{\dom}{\mbox{dom}}

\newcommand{\triple}{\ensuremath{(\mathcal{N}\ssubset \mathcal{M},\tau,\delta)}}

\renewcommand{\ltriple}{\ensuremath{(\tilde{\mathcal{N}}\ssubset \tilde{\mathcal{M}},\tau,\delta)}}

\usepackage{tikz-cd}
\usepackage[margin=1.0 in]{geometry}

\begin{document}
%\singlespace
\title[]{ Graph H\"ormander Systems
}

\thanks{NL acknowledges support from NSF grant DMS-1700168. NL is supported by NSF Graduate Research Fellowship Program DMS-1144245. HL and  MJ are partially supported by NSF grants  DMS 1800872 and Raise-TAG 1839177. }

\author[H. Li]{Haojian  Li}
\address{Department of Mathematics\\
University of Illinois, Urbana, IL 61801, USA}
\email[Haojian Li]{hli102@illinois.edu}

\author[M. Junge]{Marius Junge}
\address{Department of Mathematics\\
University of Illinois, Urbana, IL 61801, USA} \email[Marius
Junge]{mjunge@illinois.edu}

\author[N. LaRacuente]{Nicholas LaRacuente}
\address{Department of Physics\\
University of Illinois, Urbana, IL 61801, USA} \email[Nick LaRacuente]{laracue2@illinois.edu}
%\author[C. Rouz\'{e}]{Cambyse Rouz\'{e}}
%\address{Department of Mathematics\\
%	Technische Universit\"{a}t M\"{u}nchen} \email[Cambyse %Rouz\'{e}]{rouzecambyse@gmail.com}

\maketitle
\begin{abstract} This paper extends the Bakry-\'{E}mery  theorem relating the Ricci curvature and log-Sobolev inequalities to the matrix-valued setting. Using tools from noncommutative geometry, it is shown that for a right invariant second order differential operator on a compact Lie group, a lower bound for a matrix-valued modified log-Sobolev inequality is equivalent to a uniform lower bound for all finite dimensional representations. Using combinatorial tools, we obtain computable lower bounds for matrix-valued log-Sobolev inequalities of graph-H\"ormander systems using combinatorial methods.

\end{abstract}

%Section 1: introduction. not yet\\
%Section 2: almost done except the last subsection of $\CpSI$.\\
%Section 3: commuting square. rewrite the main theorem and displacement convexity proof.\\
%Section 4: check details. In particular, the broken cycle proof. ( continuity argument. The continuity of the Fisher information is well known? Can we just invoke something from Petz?)\\
%Section 5: working on it now. (the introduction part and the representation part have been done.)
%\end{abstract}
%\input{intro}
%\include{intro_corr}
%\input{part1}
%\input{rep2b}
%\input{graph}
%\input{graphH}
%\include{examples}
%%%\include{applications2}
%\include{applications}
%\include{app}
%\include{reference}

%\include{trrb}

%\include{state}

%\input{trr}

%\onehalfspace
%\input{driver2.bbl}
%\bibliographystyle{alpha}
%\bibliography{library}

\section{Introduction}

Estimates for the spectral gap of a Laplace type operator are relevant in several areas of mathematics such as analysis, geometry, combinatorics, and even theoretical  computer science. The aim of this paper is to study matrix-valued versions of log-Sobolev inequalities related to representations of compact Lie groups, and the connection to noncommutative geometry.
\hhz

Log-Sobolev inequalities have been an area of active research for several decades, see \cite{overview} for an overview. Let us highlight the local estimates on Sobolev inequalities by Rothchield and Stein \cite{SR}, the new paradigm relating the Ricci curvature and the log-Sobolev inequality \cite{BE}, the so-called Bakry-\'{E}mery theory, and the discovery by  Meyer and Gross of log-Sobolev inequalities in infinite dimension, see \cite{Ya4, Gross, Me1, Me2}.  We refer to Ledoux's article \cite{Led} for the fundamental connection between log-Sobolev inequalities and concentration inequalities. More recently, Otto and Villani, and Sturm connected geometric insights with the theory  of optimal transport and displacement convexity.
\hhz

%In the theory of open quantum systems it is natural to assume that a diffusion process interacts with an environment systems. It appears unrealistic to assume that this process is globally ergodic, and hence we do allow the existence of fixpoint algebras in our investigation.  Transferring log-Sobolev inequalities to the quantum (and non-ergodic) setting is an area of recent, ongoing research, see \cite{CM,CM2}, \cite{Ba}, \cite{Datta}, \cite{Angela,CamByze}.

Transferring these estimates to the quantum setting is an area of recent ongoing research, see \cite{CM,CM2, cm20, Ba, Datta_2020, bardet2020}. In the theory of open quantum systems, it is natural to assume that a diffusion process interacts with an environment system.  Mathematically, this implies that a certain diffusion process is no longer ergodic, acting trivially on one part of the space. Bakry-\'{E}mery theory in the non-ergodic setting is less developed. This remains particularly true for the connection between Fisher information and log-Sobolev inequalities, see however \cite{LJR}. We should point out that in the matrix-valued setting, a crucial argument based on uniform convexity of $L_p$ spaces, the famous Rothaus Lemma, is no longer valid and hence the usual approach via hypercontractivity has to be replaced by new methods. Developing these methods is the main technical contribution in this paper.
\hhz

Log-Sobolev inequalities have also been studied for discrete objects in particular graphs, see the seminal work of Saloff-Coste and Diaconis \cite{DSF}, and also later work by Yau and his collaborators \cite{Ya4, Ya5}. Nowadays, it is clear that different variations of log-Sobolev inequalities are relevant.  Let us recall that an ergodic system $T_t=e^{-t\Delta}$ on a probability space $(\Om,\mu)$  satisfies a \textit{ log-Sobolev inequality ($\la$-LSI)} if
 \[ 2\la\left( \int f^2\ln f^2 d\mu - \ln(\int f^{2}d\mu)\int f^{2} d\mu \right) \kl \int \Delta(f)f d\mu
 \lel \E_{\Delta}(f,f) \pl .\]
The right hand side defines the energy form. A proper formulation of log-Sobolev inequalities for non-ergodic system appears only in hidden form in the existing literature. Indeed, let $T_t=e^{-tL}$ be a semigroup of self-adjoint positive unital maps. Let  $N_{\fix}=\{f| T_t(f)=f\}$ (some or all $t>0$) be the subalgebra of invariant functions which admits a conditional expectation $E_{\fix}$. Then $L$ is said to satisfy a \emph{log-Sobolev inequality {\rm (}$\la$-$\LSI${\rm )}} if
 \[ \la \pl D(f^2||E_{\fix}(f^2))\kl\E_{L}(f,f) \pl .\]
Here and in the following we will use $D(f||g)=\tau(f\ln f)-\tau(f\ln g)$ for the relative entropy, and $\tau(f)=\int fd\mu$ for the canonical  trace induced by the probability measure $\mu$. We recall that  $L$ satisfies  \emph{a modified log-Sobolev inequality {\rm (}$\la$-$\MLSI${\rm )}}  if
 \[ \la \pl D(f||E_{\fix}(f)\kl \E_{L}(f,\ln f) \]
 for any positive function $f$.  The right hand side $I_L(f)=\E_{L}(f,\ln f)$ is the Fisher information. Let us denote by $\LSI(L)$ and $\MLSI(L)$  the largest possible such constants, and $\la_2$ the spectral gap. Then
 \begin{align}\label{introgap}2\LSI(L)\kl \MLSI(L) \kl 2 \la_2(L)
 \end{align}
holds  in general. For smooth manifolds we have $2\LSI(L)= \MLSI(L)$, which however fails for discrete graphs. For Riemanian manifolds there are two approaches to prove log-Sobolov estimates. One can first use local Sobolev inequalities from the Euclidean setting, and then use the aforementioned Rothaus Lemma. Alternatively, one can use Bakry-\'{E}mery theory for an equivalent measure with satisfies a lower Ricci-curvature bound, i.e. it is enough to satisfy Ricci-curvature at $\infty$, see \cite{Led}. The second approach may fail for manifolds not admitting convex functions (\cite{wu79}).
\hhz

The real power of log-Sobolev inequalities stems from their stability under perturbation, crucial in proving results for spin systems \cite{Ya4,Ya5}. The tensor stability is also crucial in Talagrand's inequality, as a special case of an $\LSI$-estimates, which has applications in computer science, see \cite{bous,Ld2}.
\hhz

Our original motivation for this projects comes from the study of Lindbladians in quantum information theory, or more generally dynamical systems on matrix algebras $\Mz_m$. Generators of a self-adjoint quantum dynamical system $T_t=e^{-tL}$ are completely classified and of the form
 \[ L(f) \lel \sum_k a_k^2f+fa_k^2-2a_kfa_k\lel \sum_k [a_k,[a_k,f]] \pl. \]
Since $\delta_k(f)=i[a_k,f]$ is a derivation, such Lindblad generators may be considered as analogues of a second order differential operators. We may define the modified log-Sobolev ($\MLSI(L)=\sup\{\la\}$) using the normalized matrix trace $\tau(f)=\frac{\tr(f)}{m}$ such that
 \[ \la \pl  D(f||E_{\fix}(f))\kl I_{L}(f)\lel \tau(L(f)\ln f) \]
 for any positive $f$. The complete modified log-Sobolev constant $\CLSI(L)=\sup\{\la\}$, the main notion of this paper, is the
best constant such that
 \[  \la  D(f||E_{\fix}(f)) \kl \tau((id\ten L)(f)\ln f) \]
for any positive $f\in \Mz_m\ten \mathcal{M}$, where $\Mm$ is a finite von Neumann algebra.  In contrast to commutative systems, tensorization with an auxiliary system is not `for free'. For the complete version of the modified log-Sobolev inequality, however, we have tensor stability (see \cite{LJR})
 \[ \CLSI(L_1\ten 1+1\ten L_2)\gl \min\{\CLSI(L_1),\CLSI(L_2))\} \pl .\]
The best way to understand the $\CLSI$ constant is via the entropy decay rate, i.e. the best constant such that
 \[ D(e^{-tL}(f)||E_{\fix}(f))\kl e^{-t\CLSI(L)} D(f|| E_{\fix}(f)) \pl \]
holds for all positive matrix valued $f$. We will show that even for commutative systems the complete, i.e. matrix version, provides useful insight. It should be noted, however, that the failure of the Rothaus Lemma is probably responsible for a the lack of large classes of examples in the quantum information literature. Before \cite{LJR} and this paper, all the known results could be deduced from a result by Bardet's result for $L=I-E$, or Gaussian systems, see \cite{BaRo}. Note, however, that tensor stability will be crucial in analyzing many body system through the interaction of local systems.
\hhz

The main focus of this paper is  Lindbladians which are ``transferred'' from a group representation $u:G\to U(\Hs)$, $\Hs$ a (finite dimensional) Hilbert space of a finite dimensional Lie group $G$ with Lie algebra $\mathfrak{g}$. A vector field $\mathcal{X}=\{X_1,...,X_m\}\subset \mathfrak{g}\cong \T_1M$ is called a H\"ormander system if the iterated commutators from $\mathcal{X}$  generate the Lie algebra. Such a H\"ormander system generates an ergodic semigroup of right-invariant maps $T_t=e^{-t\Delta_X}$ with the sub-Laplacian generator
 \[ \Delta_{\mathcal{X}} \lel \sum_j -X_j^2 \lel \sum_j X_j^*X_j \pl ,\pl X_j(f)(g)=\frac{d}{dt}f(\exp(tX_j)g)|_{t=0} \pl .\]
Here the adjoint $X_j^*$ refers to the Haar measure. The Laplace-Beltrami operator on $G$ is obtained by taking an orthonormal basis of the Lie algebra, and hence is automatically a H\"ormander system. Locally the  geometry given by the induced Carnot-Caratheodory metric is extremely well-understood thanks to the famous Box-Ball theorem, \cite{gromov}. It is, however, not easy to obtain good dimension-free estimates for the spectral gap from the Ball-Box theorem, because of the implicit dependence of the dimension of the underlying space and the number of iterations required to generate the Lie algebra, see \cite{chow, OL, SR}. Thanks to the transference theorem from \cite{LJR} such estimates also imply estimates for self-adjoint Lindbladians. Indeed, let $\hat{u}:\mathfrak{g}\to \mathbb{B}(\Hs)$ be the induced Lie algebra representation %i.e. $\tilde{u}X=ia_X$ 
such that
 \[ u(\exp(tX)) \lel e^{it\hat{u}(X)} \]
describes the one-parameter group of unitaries. Let $a_1,...,a_m\in \mathbb{B}(\Hs)$ be images of $\mathcal{X}$ under $\hat{u}$.  The induced Linbladian
 \[ L_{\mathcal{X}}^{\Hs}(f) \lel L_{\mathcal{X}}^{u,\Hs}(f)  \lel \sum_j [a_j,[a_j,f]] \]
can be controlled by the original semigroup thanks to the following diagram
\[ \begin{array}{ccc} L_{\infty}(G,\mathbb{B}(\Hs)) &\stackrel{e^{-t\Delta_{\mathcal{X}}}\ten id}{\to}& L_{\infty}(G,\mathbb{B}(\Hs)) \\
 \uparrow_{\pi}& & \uparrow_{\pi} \\
 \mathbb{B}(\Hs) &\stackrel{e^{-tL_{\mathcal{X}}^{\Hs}}}{\to}& \mathbb{B}(\Hs)  \end{array}
 \]
given by the trace preserving $^*$- homomorphism $\pi(T)(g)=u(g)^*Tu(g)$.  Therefore 
$\CLSI(L_{\mathcal{X}}^{\Hs})\gl\CLSI(\Delta_{\mathcal{X}})$. 
%$\la_2(L_X^{\Hs})\gl \la_2(\Delta_X)$. 
Our main result is what might be called an anti-transference result.

\begin{theorem}\label{maint} Let $G$ be a finite dimensional compact Lie group, $\mathcal{X}\subset \mathfrak{g}$, a vector field. Then
 \[ \CLSI^+(\Delta_\mathcal{X}) \lel \inf_{u:G\to U(\Hs)} \CLSI^+(L_{\mathcal{X}}^{\Hs}) \pl .\]
\end{theorem}
\noindent We  conjecture that indeed
 \[ \CLSI(\Delta_{\mathcal{X}}) \pl \stackrel{?}{=}\pl  \inf_{\Hs}\CLSI(L_{\mathcal{X}}^{\Hs})  \pl. \]
In this paper  we have to  work with a technical variant $\CLSI^+(L)=\inf_{p>1}\CpSI(L)$ to justify the use of Connes' trace formula in our proof of Theorem \ref{maint}.  Our $p$-R\'enyi version of a complete Sobolev inequality is completely new for quantum dynamical semigroups, although  anticipated in \cite{Bobtet}, and defined as the best constant $\CpSI=\sup\{\la\}$  such that
 \[  \la \left(\|f\|_{p}^{p}- \|E_{\fix}(f)\|_{p}^{p}\right)  \kl p \mathcal{E}_{L}(f,f^{p-1})\]
 for all positive matrix-valued $f$. In the scalar ergodic case this inequality has been studied by \cite{Bobtet},  and it was proved that the inequality is equivalent to the decay estimate
\[ \|T_t(f)\|_{p}^{p}-\|E_{\fix}(f)\|_{1}^{p}\kl e^{-t\CpSI(L)}\left(\|f\|_{p}^{p}-\|E_{\fix}(f)\|_{1}^{p}\right) .\]
 We refer to \cite{hao} for a systematic study of the relative entropy $d_p$ associated with $\CpSI$.  In the existing literature examples of $\CLSI$ for quantum and classical systems are very rare, because the usual hypercontractivity argument fails. However, we are able to extend the famous  result of Bakry-\'{E}mery in this new setting.

\begin{theorem}[Complete Bakry-\'{E}mery theorem]\label{main-M} Let $(M, g, \mu)$ be a smooth Riemannian manifold with a probability measure $\mu$ defined by $d\mu=\frac{1}{\Zu} e^{-U} dvol$ with $\Zu=\int_{M} e^{-U}dvol$ and $U\in C^{\infty}(M)$
such that the Bakry-\'{E}mery Laplacian
 \[ \int \Du(f_1)f_2 d\mu \lel \int (\nabla f_1,\nabla f_2)d\mu \]
satisfies $\Ric(\Du)\gl \kappa >0$. Then
 \[ \CLSI(\Du)\gl 2\kappa \pl .\]
\end{theorem}

%With a suitable notion of Ricci-curvature this results applies to sub-Laplacian, and hence also to systems studied by Baudoin, and Gong, Thalmeier.
Our key ingredient, motivated from \cite{CM}, is to use Lieb's concavity theorem, refined by Hiai and Petz \cite{HP}. This allows us to define  `Ricci curvature bounded below' for noncommutative dynamical systems. This result competes with the recent results in \cite{cm20, Wi}, where a notion of transportation Ricci curvature has been introduced. For finite dimensional QMS our notion of geometric Ricci curvature implies complete transportation Ricci curvature. As it turns out to be a source of a large class of examples.

%Thus there is no need to analyze the vertical part of the foliation to closely. This topic will be further perused in \cite{JLiHao}. In this paper we find a formulation of `lower bound for the Ricci curvature from below' for arbitrary von Neumann algebras. The novelty here is to consider semigroups on spaces of differential forms.
\hhz

%At any rate, the theorem is a strong indication that determining the $\CLSI$ constant for a generator on a matrix algebra requires indeed some insight on second order differential operator on a compact Lie group. $X,Y$ example.

Indeed, the remainder of this paper is to find concrete estimates for the $\CLSI$ constant for graphs and related Lindbladian or differential operators. This is motivated form quantum information theory, but certainly interesting in view finite Markov process in the sense of \cite{SFC}.
\begin{theorem}Let ${\sf{G}}=(\mathscr{V},\mathscr{E})$ be a connected undirected graph with a uniform distribution on the vertex set.  Let
 \[ A_\mathscr{E}(f)(x) \lel 2 \sum_{(x,y)\in \mathscr{E}} (f(x)-f(y)) \]
be the graph Laplacian and ${\sf{T}}_{s}$ be the minimum spanning tree with the number of edges $l({\sf{T}}_{s})$ and the maximum degree $d({\sf{T}}_{s})$ . Then
 \[ \CLSI(A_\mathscr{E}) \gl \CLSI^+(A_\mathscr{E}) \gl \frac{2}{45 d({\sf{T}}_{s})l({\sf{T}}_{s})^2} \pl .\]
\end{theorem}

This result is optimal for the cyclic graph $\zz_n$ with nearest neighboring interactions. The lower bound is efficiently computable.  We refer to \cite{Ya3} (and references therein) for other estimates of the LSI constant  that are not directly comparable to our result. Our estimate is not expected to be  the best possible because it is modelled after a long a one-dimensional structure. More edges should improve the estimates of the $\CLSI$ constant, and this is true for graphs  with tensor product structure. Our results  support the following conjecture.

\begin{conj} For every self-adjoint Lindbladian $L$, we have $\CLSI(L)>0$.
\end{conj}

We can verify this conjecture for, what we  call the Graph-H\"ormander systems. Indeed, let ${\sf{G}}=(\mathscr{V},\mathscr{E})$ be a connected undirected graph with a uniform distribution on $\mathscr{V}=\{1,\dots, n\}$. For every edge $e=(r,s)$ with $r<s$ we may define the tangent vector $X_e=|r\ran\lan s|-|s\ran\lan r|$ and
$$L_{\mathscr{E}}(x)=\sum_{e=(r,s)\in \mathscr{E}, r<s} [X_e,[X_e,x]]$$ the corresponding Lindblad transferred from the sub-Laplacian $\Delta_{\mathscr{E}} \lel- \sum_e X_e^2$ over $C^{\infty}(SO_{n})$.

\begin{theorem}  Let ${\sf{G}}=(\mathscr{V},\mathscr{E})$ be a connected undirected graph with a uniform distribution over $\mathscr{V}$.  Then $\Delta_\mathscr{E}$ is ergodic and
 \[ \CLSI^{+}(A_{\mathscr{E}}) \gl \CLSI^+(L_\mathscr{E}) \gl \frac{\CLSI^+(A_\mathscr{E})}{1+ 5\pi^{2}\CLSI^+(A_\mathscr{E})} .\]
\end{theorem}
The paper is organized as follows.  In section 2, we  introduce important tools, such as derivations and double operator integrals, which we use in later chapters. We give the definitions  of $\CLSI$ and $\CLSI^{+}$ and their properties, in particular, tensor stability and stability under perturbation. In section 3, we define the geometric Ricci curvature of a \textit{derivation triple} and establish the abstract Bakry-\'Emery theorem. We compare our \textit{geometric} Ricci curvature with the \textit{transportation} Ricci curvature previously defined by Carlen and Maas. In section 4, we recapture the Bakry-\'Emery criterion for matrix-valued functions defined over a smooth manifold.  We also include some geometric examples to illustrate derivation triples. In section 5, we briefly review the transference principle and develop the \textit{anti-transference principle} with the help of representation theory and noncommutative geometry. In section 6, we give computable estimates of $\CLSI$ constants of connected graphs via the preorder traversal algorithm and existence of spanning trees.  In section 7, we define the \textit{graph H\"ormander systems} and present the relation between $\CLSI$ constants of a connected graph and the induced Lindblad operator. %In section 8, we explore more examples and applications.
\hhz

\section{Notation and background}
\subsection{Tracial von Neumann algebras and modules}
Let $(\mathcal{N},\tau)$ be a finite von Neumann algebra equipped with a normal faithful tracial state $\tau$, and  $\mathcal{N}_{+}$ be the set of positive elements in $\mathcal{N}$. We denote the noncommutative $L_{p}$-space by  $L_{p}(\mathcal{N},\tau)$, or $L_{p}(\mathcal{N})$ if the trace $\tau$ is clear from the context. The Hilbert-Schmidt inner product over $\mathcal{N}$ is defined by $\langle x,y \rangle_{\tau}=\tau(x^{*}y)$ (also denoted by $\langle x,y\rangle$).  Let $\mathcal{N}_{1}$ and $\mathcal{N}_{2}$ be two von Neumann algebras. The Hilbert  $\mathcal{N}_{1}$-$\mathcal{N}_{2}$ bimodule $_{\mathcal{N}_{1}}\Hs_{\mathcal{N}_{2}}$ is a Hilbert space $\Hs$ equipped with representations $\pi_{1}: \mathcal{N}_{1}\rightarrow \BB(\Hs)$ and $\pi_{2}^{op}: \mathcal{N}_{2}^{op}\rightarrow \BB(\Hs)$ satisfying $[\pi_{1}(\mathcal{N}_{1}),\pi_{2}^{op}(\mathcal{N}_{2}^{op})]=0$.  Let us recall that the opposite algebra $\mathcal{N}^{op}$ is obtained by exchanging the left and right multiplications in $\mathcal{N}$, i.e., $(ab)^{op}$ in $\mathcal{N}^{op}$ is given by $ba$ in $\mathcal{N}$ for $a,b\in \mathcal{N}$.
We use the notation $xhy$ to denote the left $\mathcal{N}_{1}$ action and  the right $\mathcal{N}_{2}$ actions for $x\in\mathcal{N}_{1}$, $y\in\mathcal{N}_{2}$ and $h\in \Hs$. The Hilbert bimodule $_{\mathcal{N}}\Hs_{\mathcal{N}}$ is said to be \textit{self-adjoint} if $\mathcal{N}_{1}=\mathcal{N}_{2}=\mathcal{N}$ and there exists an antilinear involution $J: \Hs\rightarrow \Hs$ such that $J(xhy)=y^{*}J(h)x^{*}$ for any $x,y\in \mathcal{N}$ and $h\in \Hs$.
%%We also use $*$ to denote $J$, which shall be clear from the context.
 In many situations we work with the slightly stronger notion of a $W^*$-right module $X$ which admits an $\mathcal{N}$-valued inner product $(x,y)$ such that $( x,ya)=( x,y) a$ for any $a\in \mathcal{N}$ and $x,y\in X$. Let us denote by $\L_{\mathcal{N}}(X)$ the von Neumann algebra of left adjointable operators on $X$, see \cite{Paschke}, \cite{JSher},\cite{Lance}. If in addition there is  a weak$^*$ continuous $^*$-representation $\pi:\mathcal{N}\to \L_{\mathcal{N}}(X)$, $X$ becomes an $\mathcal{N}$-$\mathcal{N}$-$W^*$-bimodule. If furthermore there is an antilinear isometry $J:\L_{\mathcal{N}}(X)\to \L_\N(X)$ such that $J(\pi(a)x) \lel xa^*$ for any  $x\in X$ and $a\in \mathcal{N}$, we recover all the data from above. Our typical example is given by a trace preserving inclusion $\mathcal{N}\subset \mathcal{M}$ equipped with a conditional expectation $E_{\mathcal{N}}:\mathcal{M}\to\mathcal{N}$. Then $( x,y)=E_{\mathcal{N}}(x^*y)$ makes $\mathcal{M}$ an $\mathcal{N}$-valued right module which extends to a complete $W^*$-module $X=\overline{\mathcal{M}E_{\mathcal{N}}}\subset \mathbb{B}(L_2(\mathcal{M}))$. The left representation is, of course, given by $\pi(a)\xi=a\xi$ which extends to the closure. The underlying Hilbert space is given by $\Hs=L_2(\mathcal{M},\tau)$.
%In full generality, in the presence of a trace, the Hilbert space $H=L_2(\mathcal{M},\tau)$ is given by $\langle x,y \rangle_{\tau}=\tau\left((x,y)\right)$, where $\langle x,y \rangle_{\tau}$ is the Hilbert-Schmidt inner product over $L_{2}(\mathcal{M},\tau)$.
We see that here $J(x)=x^*$ is an isometry on the Hilbert space, but only densely defined on $\mathcal{M}\subset X$.

\subsection{Derivations}
Let $_{\mathcal{N}}\Hs_{\mathcal{N}}$ be a self-adjoint Hilbert $\mathcal{N}$-$\mathcal{N}$ bimodule with the antilinear form $J$. A closable derivation of a von Neumann algebra $\mathcal{N}$ is a densely defined closable linear operator $\delta: L_{2}(\mathcal{N},\tau)\rightarrow \Hs$ such that
\begin{itemize}[leftmargin=6.0mm]
\item[(1)] $\dom(\delta)$ is a weakly dense $^*$-subalgebra in $\mathcal{N}$;
\item[(2)] the identity element $1\in \dom(\delta)$;
\item[(3)] $\delta(xy)=x\delta(y)+\delta(x)y$, for any $x,y\in \dom(\delta)$.
\end{itemize}
\noindent %Since $\delta$ is closable, for any $\rho\in\dom(\delta)$ there exists a sequence $\{\rho_n\}\subset \dom(\delta)$ such that $\lim_{n}\rho_{n}=\rho$ and $\lim_{n}\delta(\rho_{n})=\sigma$.
Let $\bar{\delta}$ denote the closure of $\delta$.
% \noindent Let $\bar{\delta}$ denote the closure of $\delta$. Indeed, for $\rho\in\dom(\delta)$, there exists a sequence $\{\rho_n\}\subset \dom(\delta)$ such that $\lim_{n}\rho_{n}=\rho$ and $\lim_{n}\delta(\rho_{n})=\sigma$, then we define $\bar{\delta}(\rho)=\sigma$.
A derivation $\delta$ is said to be \textit{$*$-preserving} if $J(\delta(x))=\delta(x^*)$. Every closable $*$-preserving derivation $\delta$ determines a positive operator $\delta^{*}\bar{\delta}$ on $L_{2}(\mathcal{N},\tau)$.
%Since the adjoint of a densely defined operator is closed, it is sufficient to use $\delta^*$ here.
It was shown  in \cite{Sau} that $T_{t}=e^{-t\delta^{*}\bar{\delta}}: \mathcal{N} \rightarrow \mathcal{N}$ is a strongly continuous semigroup of completely positive, unital and self-adjoint maps. Thus $T_{t}$ is also trace preserving since $\tau(x^{*}T_{t}(y))=\tau(T_{t}(x^{*})y)=\tau(T_{t}(x)^{*}y)$ for any $x$ and $y\in \mathcal{N}$.
%  If $T_{t}: \mathcal{N}\rightarrow \mathcal{N}$ is a strongly continuous semigroup of completely positive, unital and self-adjoint maps, then there is a self-adjoint Hilbert $\mathcal{N}$-$\mathcal{N}$ bimodule $_{\mathcal{N}}H_{\mathcal{N}}$ and a densely defined closable $*$-preserving derivation $\delta: L_{2}(\mathcal{N})\rightarrow H$ such that $T_{t}=e^{-t\delta^{*}\bar{\delta}}$.
 See\cite{SA}, \cite{AS}, \cite{Jesse}, \cite{Kap}, and \cite{BR} for more details.
%%%% here the domain of $\delta$ may not be a $*$-algebra. We do need a $*$-algebraic structure to proceed.

Now let $T_t=e^{-tA}:\mathcal{N}\rightarrow \mathcal{N}$ be a strongly continuous semigroup of completely positive unital self-adjoint maps on $L_2(\mathcal{N},\tau)$. The generator $A$ is a positive operator on $L_{2}(\mathcal{N},\tau)$ given by $$A(x)=\lim_{t\rightarrow 0^+} \frac{1}{t}(T_t(x)-x),\forall x\in \dom(A).$$
It was pointed out  in \cite{Sau} that  $\dom(\delta)=\{x\in \mathcal{N}| \|A^{1/2}x\|_2<\infty\}$ is indeed a $*$-algebra and  invariant under the semigroup.
%%%%%% The carr\'e du champ operator of the generator $A$ is defined as  $$\Gamma_{A}(x,y)=-\frac{1}{2}(A(x^*y)-A(x^*)y-x^*A(y)).$$
The weak gradient form of $A$ is defined by
$$\Gamma_A(x,y)(z)=\frac{1}{2}(\tau(A(x)^*yz)+\tau(x^*A(y)z)-\tau(x^*yA(z))).$$
If the weak gradient form $\Gamma_A(x,y)\in L_{1}(\mathcal{N})$ for all $x,y\in \dom(A^{1/2})$, we say the generator $A$ (or $T_{t}$) satisfies \textit{$\Gamma$-regularity}. It was  shown in  \cite{JRS18}  that  we may associate the generator $A$ satisfying the $\Gamma$-regularity with a closable $*$-preserving derivation $\delta_{A}$.
\begin{theorem} \label{derivation} If $A$ satisfies $\Gamma$-regularity, then there exists a finite von Neumann algebra $(\mathcal{M},\tau)$ containing $\mathcal{N}$ and a $*$-preserving derivation $\delta_{A}:\dom(A^{1/2})\rightarrow L_2(\mathcal{M})$ such that
\begin{equation}\label{inducedder}\tau(\Gamma_{A}(x,y)z)=\tau(\delta_{A}(x)^*\delta_{A}(y)z).\end{equation}
Equivalently $\Gamma_{A}(x,y)=E_{\mathcal{N}}(\delta_{A}(x)^*\delta_{A}(y))$, where $E_{\mathcal{N}}: \mathcal{M}\rightarrow \mathcal{N}$ is the conditional expectation.
\end{theorem}
% Theorem \ref{derivation} allows us to associate the semigroup $T_{t}$ or the generator $A$ with a closable symmetric derivation $\delta_{A}$.
%It validates our abstract theory in terms of derivation triple in Section 3.
Throughout the paper,  we always work with a closable $*$-preseving derivation $\delta$ and a strongly continuous semigroup $T_{t}=e^{-tA}$ of completely positive unital self-adjoint maps on $L_2(\mathcal{N},\tau)$  satisfying $\Gamma$-regularity.
% we always assume that $T_{t}=e^{-tA}$ is a strongly continuous semigroup of completely positive unital self-adjoint maps on $L_2(\mathcal{N},\tau)$  satisfying $\Gamma$-regularity.  We only work with closable and $*$-preserving derivations $\delta$, which generates a semigroup $e^{-t\delta^{*}\bar{\delta}}$ of completely positive unital self-adjoint maps.
%%% \begin{remark} $A_{\delta}$ is the $\delta$-induced generator, and $\delta_{A}$ is an $A$-induced derivation. $\delta_{A}$ is not necessarily unique, but $\delta_{A}$ is unique up to the conditional expectation $E_{N}$ since $\Gamma_{A}(x,y)=E_{N}(\delta_{A}(x)^*\delta_{A}(y))$. \end{remark}

\subsection{Double operator integrals}
Let $\phi: \Bbb{R}\times \Bbb{R}\to \Bbb{R}$ be bounded and $\rho, \sigma\in\mathcal{N}$ be self-adjoint.
%For a bounded  complex valued function  $\phi:  \Bbb{R}\times \Bbb{R}\rightarrow \Bbb{C}$ and self-adjoint $\rho,\sigma\in \mathcal{N}$,
The double operator integral  is defined by
$$Q_{\phi}^{\rho,\sigma}(T):=\int_{\Bbb{R}}\int_{\Bbb{R}}\phi(s,t) dE_{\rho}(s)TdE_{\sigma}(t),$$
where $E_{\rho}((s,t])=1_{(s,t]}(\rho)$ is the spectral projection of $\rho$. We write
 $Q_{\phi}^{\rho}$ if $\rho=\sigma$. The notion of double operator integrals was first introduced by Daleckii and Krein (see \cite{krein1}, \cite{krein2}) used for the analytical theory of perturbations. Further construction  of double operator integrals was created in a series of papers, (see \cite{bs1}, \cite{bs2}, \cite{bs3}) by Birman and Solomyak.
Let  $f:\Bbb{R}_{+}\rightarrow \Bbb{R}$ be  a continuously differentiable function and the difference quotient  be $f^{[1]}(x,y)=\frac{f(x)-f(y)}{x-y}$, then
$$Q_{f^{[1]}}^{\rho,\sigma}(T)=\int_{\Bbb{R}_{+}}\int_{\Bbb{R}_{+}}\frac{f(s)-f(t)}{s-t} dE_{\rho}(s)T dE_{\sigma}(t).$$
See \cite{PS10} for the convergence of the above formula. We abbreviate $Q^{\rho,\sigma}_{\ln^{[1]}}$, $Q^{\rho}_{\ln^{[1]}}$ as $Q^{\rho,\sigma}$, $Q^{\rho}$, respectively.  It was also shown in \cite{PS10} that
\begin{align*}
\lim_{t\to 0}\frac{f(\rho+t\sigma)-f(\rho)}{t}=Q^{\rho}_{f^{[1]}}(\sigma).
\end{align*}
Thus $\tau\left(Q^{\rho}_{f^{[1]}}(\sigma)\right)=\tau(f'(\rho)\sigma)$.
For $\rho\in \mathcal{N}_{+}$, recall that the functional calculus of derivations is given by
$$\delta(f(\rho))=\int_{\Bbb{R}_{+}}\int_{\Bbb{R}_{+}}\frac{f(s)-f(t)}{s-t}dE_{\rho}(s)\delta(\rho)E_{\rho}(t).$$
Hence $\delta(f(\rho))=Q_{f^{[1]}}^\rho(\delta(\rho))$.

\begin{example}\label{CM}
Let $f(x)=\ln(x)$, recall that $\frac{\ln(x)-\ln(y)}{x-y}=\int_{\Bbb{R}_{+}}\frac{1}{(x+r)(y+r)}dr$, then
$$Q^{\rho,\sigma}(y)=\int_{\Bbb{R}_{+}} (\rho+r)^{-1}y(\sigma+r)^{-1}dr.$$
In particular {\rm(\cite{CM})}  $\delta(\ln(\rho))=\int_{\Bbb{R}_{+}}(\rho+r)^{-1}\delta(\rho)(\rho+r)^{-1}dr.$
\end{example}
\noindent Let $f: \Bbb{R}_{+}\rightarrow \Bbb{R}_{+}$ be operator monotone and $f_{[0]}(x,y)=f(\frac{x}{y})y$,  and we consider
$$Q_{f_{[0]}}^{\rho,\sigma}(T)=\int_{\Bbb{R}_{+}}\int_{\Bbb{R}_{+}} f_{[0]}(s,t) dE_{\rho}(s)T dE_{\sigma} (t).$$
\begin{example}\label{example-hiai} Let $f(x)=\frac{x-1}{\ln(x)}$, then $f$ is operator monotone. Indeed $f(x)=\int_{0}^{1} x^{r}dr$ is a convex combination of operator monotone functions $x^{r}$.  By the integral identity $\frac{x-y}{\ln(x)-\ln(y)}=\int_{\Bbb{R}_{+}} x^{r}y^{1-r}dr$, we obtain that  $$Q_{f_{[0]}}^{\rho,\sigma}(T)=\int_{0}^{1}\rho^{r}T\sigma^{1-r}dr.$$
An important observation is that
$Q^{\rho}(T)=Q^{\rho}_{f_{[0]}^{-1}}(T).$
\end{example}

\subsection{Lieb's concavity Theorem}
Lieb (\cite{Lie73})  proved that the map $(A,B) \mapsto \tau(K^{*}A^{1-t}KB^{t})$ with $t\in[0,1]$ is jointly concave in the positive definite matrix pair $(A,B)$, usually referred to as \textit{ Lieb's concavity Theorem}.  %Kosaki recaptured and also generalized this result in the framework of interpolation theory.
 Petz (\cite{petz85}) discovered the following generalized  Lieb's concavity theorem by using the Jensen inequality of operator concave functions.
\begin{theorem}\label{HPconvex} Let $\beta: \mathcal{N}\rightarrow \mathcal{N}$ be a completely positive trace preserving map and $f:(0,\infty)\rightarrow (0,\infty)$ be an operator monotone function. Then for any $\rho,\sigma\in\mathcal{N}_{+}$, we have
$$\beta^{*} Q_{f_{[0]}^{-1}}^{\beta(\rho),\beta(\sigma)} \beta \leq Q_{f_{[0]}^{-1}}^{\rho,\sigma}.$$
Furthermore, $(\rho,\sigma,x)\mapsto \langle  x , Q_{f_{[0]}^{-1}}^{\rho,\sigma} (x)\rangle $
 is jointly convex for $\rho,\sigma\in \mathcal{N}_{+}$ and $x\in \mathcal{N}$.
 \end{theorem}
\begin{lemma}\label{dc} Let $f:(0,\infty)\to (0,\infty)$  and $\beta: \mathcal{N} \to \mathcal{N}$ be a completely positive trace preserving map.  The conditions
\begin{align}\beta^{*} Q_{f_{[0]}^{-1}}^{\beta(\rho),\beta(\sigma)} \beta& \leq Q_{f_{[0]}^{-1}}^{\rho,\sigma} \label{C1}
\end{align}
and
\begin{align}
\beta Q_{f_{[0]}}^{\rho,\sigma} \beta^{*} & \leq Q_{f_{[0]}}^{\beta(\rho),\beta(\sigma)} \label{C2}
\end{align}
are equivalent for any $\rho,\sigma \in \mathcal{N}_{+}$.
\end{lemma}
\noindent See the proof of Lemma 1 and Theorem 5 in \cite{HP}, and they assumed that $\rho,\sigma, \beta(\rho),\beta(\sigma)$ are invertible additionally. It is enough to assume the positivity by perturbation argument $\rho+\eps I$ for $\eps\to 0^{+}.$
\noindent Theorem \ref{HPconvex} remains true for a larger family of functions. (For details, see \cite{hao}.)
\begin{theorem} Let $\beta: \mathcal{N}\rightarrow \mathcal{N}$ be a  completely positive trace preserving mapping and $f(x)=x^{p}$, where $p\in(0,1)$. Assume that $\rho, \sigma\in \mathcal{N}_{+}$. Then $$\beta^{*} Q_{f^{[1]}}^{\beta(\rho),\beta(\sigma)} \beta \leq Q_{f^{[1]}}^{\rho,\sigma}.$$
\end{theorem}
%\begin{cor}  Let $f(x)=x^{p}$ with $p\in(0,1)$. For $\rho, \sigma\in\mathcal{N}_{+}$ and $c(t)=(1-t)\rho+t\sigma$, we have $$2\int_{0}^{1} \tau \left( (\rho-\sigma) Q^{c(t)}_{f^{[1]}}(\rho-\sigma) \right)\leq \tau \left( (\rho-\sigma) Q^{\rho}_{f^{[1]}}(\rho-\sigma) \right)+ \tau \left( (\rho-\sigma) Q^{\sigma}_{f^{[1]}}(\rho-\sigma)\right)$$\end{cor}
\subsection{Formulation of $\CLSI$}
\subsubsection{Quantum relative entropy}
Recall that the \textit{quantum relative entropy} of $\rho,\sigma\in \mathcal{N}_{+}$ is
\[
  D^{\tau}(\rho\|\sigma) \lel
  \begin{cases}
                                   \tau(\rho\ln(\rho))-\tau(\rho\ln(\sigma)),& \text{if } supp(\rho)\supset supp(\sigma); \\
                                   +\infty, & \text{otherwise,}
  \end{cases}
\]
where $supp(\rho)$ is the support projection of $\rho$. We denote the relative entropy by $D(\rho\|\sigma)$ if the trace $\tau$ is clear from the context. Equivalently  $D(\rho\|\sigma)=\lim_{\epsilon \rightarrow 0^+} D(\rho\|\sigma+\epsilon 1)$.
%$D(\rho\|\sigma)$ is jointly convex in $\rho$ and $\sigma$.
See e.g. \cite{Wilde} and \cite{Nielsen} for more entropy properties. Relative entropy is monotone decreasing under the application of quantum channels (also known as \textit{data processing inequality})
$$D(\beta(\rho)\|\beta(\sigma))\leq D(\rho\|\sigma),$$
 where  $\beta: \mathcal{N}\rightarrow \mathcal{N}$ is a completely positive trace preserving linear map. Let  $\mathcal{K}$ be a von Neumann subalgebra of $\mathcal{N}$ and $E_{\mathcal{K}}: \mathcal{N}\rightarrow\mathcal{K}$ be the conditional expectation onto $\mathcal{K}$. The relative entropy  $D_{\mathcal{K}}$ with respect to $\mathcal{K}$ is given by
\begin{align}\label{fix-d} D_{\mathcal{K}}(\rho)=D(\rho\|E_{\mathcal{K}}(\rho))=\inf_{\tau(\rho)=\tau(\sigma),\sigma\in \mathcal{K}}D(\rho\|\sigma).
\end{align}

\begin{lemma} \label{iter} Let $E_1$,....,$E_m$ be pairwise commuting  conditional expectations on $\N$. Then
 \[ D\left(\rho\|(\small{\prod}_{j=1}^{m} E_j)(\rho)\right)\kl \sum_{j=1}^m D(\rho\|E_j(\rho)) \pl.\]
\end{lemma}
\begin{proof} Let us define the subalgebras $\mathcal{M}_{m-k}=(\prod_{j=1}^{k} E_j)(\N)$, so that $\MM_0\subset \MM_1\subset \cdots \MM_{m}=\N$ is a filtration. Then we  deduce from the data processing inequality
 \begin{align*}
D(\rho\|E_{\MM_0}(\rho)) &= D(\rho\|E_{\MM_1}(\rho))+D(E_{\MM_1}(\rho)\|E_{\MM_0}(\rho))\\
&=D(\rho\|E_{\MM_1}(\rho))+ D\left((\small{\prod}_{j=1}^{m-1} E_{j} )(\rho)  \| (\small{\prod}_{j=1}^{m} E_{j} )(\rho) \right)\\
  &\le D(\rho\|E_{\MM_1}(\rho))+ D(\rho\|E_{m}(\rho)).
  \end{align*}
Repeating the argument for the first term $m-1$ times yields the assertion.\qd
\noindent Lindblad extended the relative entropy to positive functionals
  $$D_{\Lin}(\rho\|\sigma)\lel \tau(\rho\ln(\rho)-\rho\ln(\sigma)-\rho+\sigma), \forall \rho,\sigma\in\mathcal{N}_{+}.$$
  It follows from the definition that $D_{\Lin}(\rho\|\sigma)\geq 0,$
  with equality if and only if $\rho=\sigma.$ (\cite{Lindblad})
Thus $D(\rho\|\sigma)$ is non-negative when $\tau(\rho)\geq \tau(\sigma)$.  By the nonnegativity of Lindblad relative entropy, we rewrite the relative entropy with respect to $\mathcal{K}$ as
\begin{align}D_{\mathcal{K}}(\rho)=\inf_{\sigma\in\mathcal{K}}D_{\Lin}(\rho\|\sigma).\label{inf-lin}\end{align}
%\textcolor{red}{$D(\rho\|E_{\mathcal{K}}(\rho))=D_{\Lin}(\rho\|E_{\mathcal{K}}(\rho))=D_{\Lin}(\rho\|\sigma)+D_{\Lin}(\sigma\| E_{\mathcal{K}}(\rho))\geq D_{\Lin}(\rho\|\sigma)$ with equality if and only if $\sigma=E_{\mathcal{K}}(\rho)$.-Delete later.}
Recall that for any finite von Neumann algebra $\mathcal{N}$,
there exists a $\sigma$-finite measure space $(X,\mu)$ such that $\mathcal{Z}(\mathcal{N})\cong L_{\infty}(X,\mu)$ and $\mathcal{N}=\int_{X} \mathcal{N}_{x}d\mu(x)$, where $\mathcal{Z}(\mathcal{N})$ is the center of $\mathcal{N}$ and $\mathcal{N}_{x}$ is a factor for any $x\in X$.  Now we rewrite the relative entropy by using the direct integral
  $$D_{\Lin}(\rho\|\sigma)=\int_{X} D_{\Lin}(\rho_{x}\|\sigma_{x})d\mu(x).$$
\begin{lemma} \label{ent-per} Let $\tau_{1}$ and $\tau_{2}$ be two normal faithful traces over a finite von Neumann algebra $\mathcal{N}$ such that $\frac{d\tau_{1}}{d\tau_{2}}\leq c$ for $c>0$. For any $\rho,\sigma \in \mathcal{N}_{+}$,
$$D^{\tau_{1}}_{\Lin}(\rho\|\sigma)\leq cD^{\tau_{2}}_{\Lin}(\rho\|\sigma).$$
In particular, we have $D^{\tau_{1}}_{\mathcal{K}}(\rho) \leq c D^{\tau_{2}}_{\mathcal{K}}(\rho).$
\end{lemma}
\begin{proof}
Note that two traces only differ by two measures $\mu_{1}$ and $\mu_{2}$ over the center $L_{\infty}(X,\mu_{1})\cong L_{\infty}(X,\mu_{2})\cong Z(\mathcal{N})$ . Also note that $\frac{d\tau_{1}}{d\tau_{2}}\leq c$ if and only if $\frac{d\mu_{1}}{d\mu_{2}}\leq c$. Again by the non-negativity of the Lindblad relative entropy, we have
\begin{align*}
D^{\tau_{1}}_{\Lin}(\rho\|\sigma)\lel &\int_{X} D_{\Lin}^{\tau_{1}}(\rho_{x}\|\sigma_{x})d\mu_{1}(x) \\
\leq &c \int_{X} D_{\Lin}^{\tau_{2}}(\rho_{x}\|\sigma_{x})d\mu_{2}(x) \lel cD^{\tau_{2}}_{\Lin}(\rho\|\sigma).
\end{align*}
The second assertion follows from \eqref{inf-lin}.
\end{proof}

\subsubsection{Quantum Fisher information}
The \textit{Fisher information} $I_{A}$ of  $A$ is defined by
\begin{align}\label{fisher-a}
I_{A}^{\tau}(\rho)=\tau(A(\rho)\ln (\rho) ),\quad\forall \rho\in \dom(A^{1/2})\cap L_2(\mathcal{N}) \text{ and } \ln(\rho) \in L_{\infty}(\mathcal{N}).
\end{align}
 Equivalently  $I_{A}(\rho)=\lim_{\epsilon\rightarrow 0^+}\tau(A(\rho)\ln(\rho+\epsilon 1))$. The Fisher information $I_{A}$ is also called the entropy production (\cite{Her}).
For a derivation $\delta$, the Fisher information is defined by
\begin{align}\label{fisher-d}
I_{\delta}^{\tau}(\rho)=\tau\left(\delta(\rho)Q^{\rho}(\delta(\rho))\right), \quad \forall\rho\in \dom(\delta)\subset \mathcal{N}.
\end{align}
Then $I_{\delta}(\rho)=I_{\delta^{*}\bar{\delta}}(\rho).$We use $I_{A}$ or $I_{\delta}$ if the trace is clear from the context.
By Theorem \ref{derivation}, for any $A$ satisfying $\Gamma$-regularity, there exists a closable $*$-preserving derivation $\delta_{A}:\dom(A^{1/2})\to L_{2}(\mathcal{M})$ such that $\Gamma_{A}(x,y)=E_{\mathcal{N}}(\delta_{A}(x)^*\delta_{A}(y))$ where $E_{\mathcal{N}}: \mathcal{M}\to\mathcal{N}$.
Thus $$I_{A}(\rho)=I_{\delta_{A}}(\rho).$$ The choice of $\delta_{A}$ is not necessarily unique, but $I_{A}$ is uniquely determined by (\ref{inducedder}).
%The Fisher information is defined for any self-adjoint operator $A$, but we only discuss the cases where $A$ is the generator of $T_{t}$ throughout the paper.
\begin{prop} \label{fish-po} The Fisher information $I_{A}$ {\rm{(}}$I_{\delta}${\rm{)}} is non-negative and convex.
\end{prop}
\begin{proof}
%By Theorem \ref{derivation}, there exists a larger von Neumann algebra $\mathcal{M}$ containing $\mathcal{N}$ and  a derivation $\delta_{A}:\mathcal{N}\rightarrow \mathcal{M}$ such that $$I_{A}(\rho)=\tau(\delta_{A}(\rho)\delta_{A}(\ln(\rho))) =\tau(\delta_{A}(\rho)Q^{\rho}(\delta_{A}(\rho))).$$

Recall Example \ref{CM} that $$I_{A}(\rho)=\int_{\Bbb{R}_{+}}\tau\left(\delta_{A}(\rho)(\rho+r)^{-1} \delta_{A}(\rho)(\rho+r)^{-1}\right)dr.$$  Let us define the differential form $$w_{r}\lel (\rho+r)^{-1/2}\delta_{A}(\rho)(\rho+r)^{-1/2},$$ then
$$I_{A}(\rho)= \int_{\Bbb{R}_{+}} \tau \left( E_{\mathcal{N}}(w_{r}w_{r}) \right) dr\geq 0.$$
The convexity was proved in \cite{HP} by using Theorem \ref{HPconvex}. Similar argument applies for $I_{\delta}$.
\qd
 % \noindent From the above proof, we note that $I_{A}(\rho)=I_{\delta_{A}}(\rho)$. By Theorem $\ref{derivation}$, the choice of $\delta$ might not be unique, but the Fisher information value is uniquely determined by (\ref{inducedder}).
%Again, this validates our abstract definition of derivation triple in Section \textcolor{red}{3, number may change later}.
\begin{lemma}  \label{fisher-derivative}
Let $E: \mathcal{N}\to \mathcal{N}_{\fix}$ be the conditional expectation onto the fixed-point algebra $\mathcal{N}_{\fix}\subset\mathcal{N}$ of the semigroup $T_{t}=e^{-tA}$, then $E\circ T_{t}=T_t\circ E=E.$ We also have
$$I_{A}(T_{t}(\rho))=-\frac{d}{dt}D_{\mathcal{N}_{\fix}}(T_t(\rho)).$$
\end{lemma}
\noindent See \cite{LJR} for the proof. This result, especially the classical case, goes back \cite{KL51}.
%%% \textcolor{red}{ this result, especially the classical case, dated back to S. Kullback, Information Theory and Statistics. New York: Dover, 1968.}
\begin{lemma} \label{fish-per}Let $\tau_{1}$ and $\tau_{2}$ be two normal faithful traces over a finite von Neumann algebra $\mathcal{N}$ and $\frac{d\tau_{1}}{d\tau_{2}}\geq c$ for  $c>0$. Then for any $\rho\in \mathcal{N}_{+}$,
$$c I_{A}^{\tau_{2}}(\rho)\leq I_{A}^{\tau_{1}}(\rho).$$
It remains true for $I_{\delta}$.
%% compare Fisher information
\end{lemma}
\begin{proof} We use the direct integral argument and notations in Lemma \ref{ent-per}. Again $\frac{d\tau_{1}}{d\tau_{2}}\geq c$ if and only if $\frac{d\mu_{1}}{d\mu_{2}}\geq c$. Let us consider the pointwise differential form $$w_{x,r}\lel (\rho_{x}+r)^{-1/2}\delta_{A}(\rho)_{x}(\rho_{x}+r)^{-1/2}.$$ By the non-negativity of the Fisher information at any point $x\in X$ (Proposition \ref{fish-po}), we have
\begin{align*}
cI_{A}^{\tau_{2}}(\rho)= &c \int_{X}\int_{0}^{\infty} {\tau_{2}}_{x}\left( E_{\mathcal{N}}(w_{x,r}w_{x,r}) \right) d\mu_{2}(x)dr\\
\kl &\int_{X} \int_{0}^{\infty}{\tau_{1}}_{x}\left( E_{\mathcal{N}}(w_{x,r}w_{x,r}) \right) d\mu_{1}(x)dr\lel  I_{A}^{\tau_{1}}(\rho).
\end{align*}
Similar argument applies for $I_\delta$.
\end{proof}

\subsubsection{Log-Sobolev type inequalities}
\begin{definition} The semigroup $T_{t}=e^{-tA}$ or the generator $A$ with the fixed-point algebra $\mathcal{N}_{\fix}$ is said to satisfy:
\begin{itemize}[leftmargin=6mm]
\item[{\rm{(1)}}] the modified log-Sobolev inequality $\lambda$-$\MLSI$ {\rm{(}}with respect to the trace $\tau${\rm{)}} if there exists a constant $\lambda>0$ such that
$$\lambda D_{\mathcal{N}_{\fix}}(\rho)\leq  I_{A}(\rho),\quad \forall \rho\in \dom(\delta)\cap \mathcal{N}_{+};$$
{\rm{(}}or equivalently $D_{\mathcal{N}_{\fix}}(T_{t}(\rho))\leq e^{-\lambda t} D_{\mathcal{N}_{\fix}}(\rho),\: \forall \rho\in \mathcal{N}_{+}.${\rm{)}}
\item[{\rm{(2)}}] the complete log-Sobolev inequality $\lambda$-$\CLSI$ {\rm{(}}with respect to the trace $\tau${\rm{)}} if $A \otimes id_{\mathcal{F}}$ satisfies $\lambda$-$\MLSI$ for any finite von Neumann algebra $\mathcal{F}$.
\end{itemize}
Let $\CLSI(A,\tau)$ be the supremum of $\lambda$ such that $A$ satisfies $\lambda$-$\CLSI$, or denoted by $\CLSI(A)$ if there is no ambiguity.  We also use $\CLSI(T_{t})$ for convenience. %We write $\CLSI(A)>0$  if $A$ satisfies $\CLSI$.
The derivation $\delta$ is said to satisfy $\lambda$-$\MLSI$ {\rm{(}}$\lambda$-$\CLSI${\rm{)}}  if $\delta^{*}\bar{\delta}$ satisfies $\lambda$-$\MLSI$ {\rm{(}}$\lambda$-$\CLSI${\rm{)}}.   Similarly we define $\CLSI(\delta,\tau)$ and $\CLSI(\delta)$ for the derivation $\delta$.
\end{definition}
\noindent An edge of $\CLSI$ over the log-Sobolev inequality  for quantum systems is tensorization stability (see \cite{LJR}).
\begin{prop} \label{tensta}
Let $T^{j}_{t}: \mathcal{N}_{j}\to \mathcal{N}_{j}$ be a family of semigroups with
fixed-point algebras $\mathcal{N}_{\fix,j}\subset \mathcal{N}_{j}$ for $1\leq j\leq k$ . Then the tensor semigroup $T_{t}=\otimes_{j=1}^{k}T_{t}^j$ has the fixed-point algebra $\mathcal{N}_{\fix}=\otimes_{j=1}^{k}\mathcal{N}_{\fix,j}$. Moreover, we have
$$\CLSI(T_{t})\geq \inf_{1\leq j\leq k} \CLSI(T^{j}_{t}).$$
\end{prop}
\noindent Combining Lemma \ref{ent-per} and Lemma \ref{fish-per}, we obtain the following \textit{change of measure principle}.
The following observation is also referred to as \textit{Holley Stroock} argument in the literature.
% This is also referred to as the Holley Stroock perturbation argument in the literature.
\begin{theorem}[change of measure principle]  \label{sutp} Let $\tau_{1}$ and $\tau_{2}$ be normal faithful traces over $\mathcal{N}$ and $c_{2}\leq \frac{d\tau_{1}}{d\tau_{2}}\leq c_{1}$ for some $c_{1},c_{2}>0$. Then $\CLSI(A,\tau_{1})\geq \frac{c_{2}}{c_{1}}\CLSI(A,\tau_{2})$.
\end{theorem}
%\begin{proof}  For any $\rho\in\mathcal{N}_{+}$, we have  \begin{align*} D_{\mathcal{N}_{\fix}}^{\tau_{1}}(\rho)\leq c_{1} D_{\mathcal{N}_{\fix}}^{\tau_{2}}(\rho) \leq \frac{c_{1}}{\CLSI(\delta,\tau_{2})} I_{\delta}^{\tau_{2}}(\rho)\leq \frac{c_{1}}{c_{2}\CLSI(\delta,\tau_{2})} I_{\delta}^{\tau_{1}}(\rho). \end{align*} \end{proof}
\begin{lemma} \label{expo decay} Let $\lambda>0$. If the Fisher information decays exponentially
$$I_{A}(T_{t}(\rho))\leq e^{-t\lambda } I_{A}(\rho),\quad \forall\rho\in\mathcal{N}_{+},$$
then $\CLSI(A)\geq \lambda$.
\end{lemma}
\begin{proof} Let $f(t)=D_{\mathcal{N}_{\fix}}(T_{t}(\rho))$, then
$f'(t)=-I_{A}(T_{t}(\rho))$ by Lemma \ref{fisher-derivative}.
Integrating both sides over $[0,\infty)$ yields the assertion.
\end{proof}

 \subsection{ $\CpSI$ and $\CLSI^+$} Now we give a brief introduction of complete $p$-Sobolev type inequalities and refer to \cite{hao} for details and generalized results. In the sequel, we assume that $p\in(1,2)$. For any $\rho,\sigma\in \mathcal{N}_{+}\cap L_{p}(\mathcal{N})$, the \textit{quantum p-relative entropy} is defined as
\[d^{p}(\rho\|\sigma) \lel  \tau(\rho^{p}-\sigma^{p})-p\tau((\rho-\sigma)\sigma^{p-1}).
\]
It follows from the definition that $d^{p}(\rho\|\sigma)\geq 0$, with the equality if and only if $\rho=\sigma$.
The $p$-relative entropy with respect to the von Neumann subalgebra $\mathcal{K}\subset \mathcal{N}$ is defined by
$$d^{p}_{\mathcal{K}}(\rho)=d^{p}(\rho\|E_{\mathcal{K}}(\rho))=\inf_{\tau(\rho)=\tau(\sigma),\sigma\in \mathcal{K}}d^{p}(\rho\|\sigma).$$
Lemma \ref{iter} remains true for the $p$-relative entropy, see \cite{hao} for the proof.
\begin{lemma} \label{iter-p} Let $E_1$,....,$E_m$ be pairwise commuting conditional expectations on $\N$, then
 \[ d^{p}\left(\rho\|(\small{\prod}_{j=1}^{m} E_j)(\rho) \right)\kl \sum_{j=1}^m d^{p} (\rho\|E_j(\rho)) \pl.\]
\end{lemma}

\noindent In \cite{hao}, we defined the  $p$-Fisher information $I_{A}^{p}$ of the generator $A$ of a semigroup $e^{-tA}$ by
$$I_{A}^{p}(\rho)=p\tau(A(\rho)\rho^{p-1}),  \quad\forall \rho\in \dom(A^{1/2})\cap L_{p-1}(\mathcal{N},\tau)$$
and the $p$-information $I_{\delta}^{p}$ of  $\delta$ by
$$I_{\delta}^{p}(\rho)=p\tau\left(\delta(\rho)Q^{\rho}_{{(x^{p-1})}^{[1]}}(\delta(\rho))\right),\quad \forall \rho\in \dom(\delta)\subset\mathcal{N}.$$
Note that $$I_{A}(\rho)=\lim_{p\rightarrow 1^{+}}\frac{I_{A}^{p}(\rho)}{p-1} \quad \text{and} \quad D_{\Lin}(\rho\|\sigma)=\lim_{p\rightarrow 1^{+}} \frac{d^{p}(\rho\|\sigma)}{p-1}.$$
 See \cite{hao} for the nonnegativity and convexity of the quantum $p$-Fisher information.
%The positivity in the proof above actually follows from  $$I_{\delta}^{\tau}(\rho)=\langle (Q^{\rho})^{1/2}(\delta(\rho))(Q^{\rho})^{1/2}(\delta(\rho))\rangle_{\tau},$$ where $(Q^{\rho})^{1/2}(y)=\int_{0}^{\infty} (\rho+r)^{-1/2}y (\rho+r)^{-1/2}dr$. This still remains true for the double operator integral $Q^{\rho}_{f^{[1]}}$ of the quotient difference of any monotone increasing functions and $$(Q^{\rho}_{f^{[1]}})^{1/2}(y)=\int_{0}^{\infty}\int_{0}^{\infty}\sqrt{\frac{f(s)-f(t)}{s-t}}dE_{\rho}(s)dE_{\rho}(t).$$
\begin{definition} The semigroup $T_{t}=e^{-tA}$ or the generator $A$ with the fixed-point algebra $\mathcal{N}_{\fix}$ is said to satisfy:
\begin{itemize}[leftmargin=6mm]
\item[{\rm{(1)}}] the modified $p$-Sobolev inequality $\lambda$-$\MpSI$ {\rm{(}}with respect to the trace $\tau${\rm{)}} if there exists a constant $\lambda>0$ such that
$$\lambda d_{\mathcal{N}_{\fix}}^{p}(\rho)\leq  I_{A}^{p}(\rho),\quad \forall \rho\in \dom(\delta)\cap \mathcal{N}_{+};$$
{\rm{(}}or equivalently $d_{\mathcal{N}_{\fix}}^{p}(T_{t}(\rho))\leq e^{-\lambda t} d_{\mathcal{N}_{\fix}}^{p}(\rho),\: \forall \rho\in \mathcal{N}_{+}.${\rm{)}}
\item[{\rm{(2)}}] the complete $p$-Sobolev inequality $\lambda$-$\CpSI$ {\rm{(}}with respect to the trace $\tau${\rm{)}} if $A \otimes id_{\mathcal{F}}$ satisfies $\lambda$-$\MpSI$ for any finite von Neumann algebra $\mathcal{F}$.
\item[{\rm{(3)}}] the enhanced complete log-Sobolev inequality $\lambda$-$\CLSI^{+}$ if $A$ satisfies $\lambda$-$\CpSI$ for $p\in(1,2)$.
\end{itemize}
Let $\CpSI(A,\tau)$ {\rm{(}}$\CLSI^{+}(A,\tau)${\rm{)}} be the supremum of $\lambda$ such that $A$ satisfies $\lambda$-$\CpSI${\rm{(}}$\CLSI^{+}${\rm{)}}, or denoted by $\CpSI(A)$ {\rm{(}}$\CLSI^{+}(A)${\rm{)}}  if there is no ambiguity.  %We write $\CLSI(A)>0$  if $A$ satisfies $\CLSI$.
The derivation $\delta$ is said to satisfy $\lambda$-$\MpSI$ {\rm{(}}$\lambda$-$\CpSI${\rm{)}}  if $\delta^{*}\bar{\delta}$ satisfies $\lambda$-$\MpSI$ {\rm{(}}$\lambda$-$\CpSI${\rm{)}}.   Similarly we define $\CpSI(\delta,\tau)$, $\CpSI(\delta)$, $\CLSI^{+}(\delta, \tau)$, and $\CLSI^{+}(\delta)$ for the derivation $\delta$.
\end{definition}
\noindent It is enough to define $\CLSI^{+}$ for $p\in(1, 1+\eps)$ for some $\eps>0$ to retrieve the log-Sobolev inequality. The limiting case $p\to 2^{+}$ reduces to the Poincar\'e inequality.
\begin{theorem}\label{bardet-p} Let $E_{\mathcal{K}}$ be the conditional expectation onto the von Neumann subalgebra $\mathcal{K}\subset\mathcal{N}$, then we have $$\CpSI(I-E_{\mathcal{K}})\geq p.$$
\end{theorem}
\begin{proof} It is obvious that $\mathcal{N}_{\fix}=\mathcal{K}$. By the operator concavity of $x^{p-1}$, we have $$(E_{\mathcal{K}}(\rho))^{p-1}\leq E_{\mathcal{K}}(\rho^{p-1}).$$ Thus
\begin{align*}
d_{\mathcal{N}_{\fix}}^{p}(\rho)=\tau(\rho^{p}-E_{\mathcal{K}}(\rho) (E_{\mathcal{K}}(\rho))^{p-1}  )\leq \tau(\rho^{p}-E_{\mathcal{K}}(\rho) E_{\mathcal{K}}(\rho^{p-1}))=\frac{1}{p}I_{I-E_{\mathcal{K}}}^{p}(\rho).
\end{align*}
\end{proof}
\noindent Here is a list of important properties of $\CpSI$ and refer to \cite{hao} for proofs.
\begin{theorem}  Let $T_{t}=e^{-tA}:\mathcal{N}\to\mathcal{N}$, then
\begin{enumerate}[leftmargin=6mm]
\item[{\rm{(1)}}] $I_{A}^{p}(e^{-tA}(\rho))=-\frac{d}{dt}d_{
\mathcal{N}_{\fix}}^{p}(e^{-tA}(\rho));$
\item [{\rm{(2)}}]   the exponential decay of Fisher information
$I_{A}^{p}(T_{t}(\rho))\leq e^{-t\lambda } I_{A}^{p}(\rho)$
implies $\CpSI(A)\geq \lambda$;
\item[{\rm{(3)}}]  $\CLSI(A)\geq \CLSI^{+}(A)$;
\item[{\rm{(4)}}]  $\CpSI$ and $\CLSI^{+}$  are stable under tensorization;
\item[{\rm{(5)}}]  $\CpSI$ and $\CLSI^{+}$ are stable under change of measure.
\end{enumerate}
Similar results remain true for $\delta$.
\end{theorem}

\section{Derivation triple}
Let $\mathcal{N}$ be a finite von Neumann algebra equipped with a normal faithful tracial state $\tau$, and $\delta$ be a closable $*$-preserving  derivation on $\mathcal{N}$. Suppose there exists a larger finite von Neumann algebra $(\mathcal{M},\tau)$ containing $\mathcal{N}$ and a weakly dense $^*$-subalgebra $\mathcal{A}\subset \mathcal{N}$ such that
\begin{enumerate}[leftmargin=6mm]
% \item there exists a trace preserving inclusion $i: \mathcal{N}\rightarrow \mathcal{M}$;
\item $\mathcal{A}\subset \dom(\delta);$
% \item $\delta^{*}\bar{\delta}: \mathcal{A}\to\mathcal{A};$ (\textcolor{red}{suggested by Li Gao. We need it to define the commutator identity.})
\item $\delta: \mathcal{A}\rightarrow L_{2}(\mathcal{M},\tau)$.
\end{enumerate}
We call such $\triple$ a \textit{derivation triple}. This notion is closely related to, and inspired by Connes' notion of a spectral triple $(\A,H,D)$ given by a representation $\pi:\A\to \B(H)$ and a (usually unbounded) self-adjoint operator $D$ such that
 \[ \delta(a) \lel [D,\pi(a)] \]
is bounded. In classical geometry $D$ is the Dirac operator and $[D,\pi(a)]\in C_0(C\ell(M))$, where $C_0(C\ell(M))$ is $C^*$-bundle of Clifford algebras of the dimension of $M$ over $M$. This Clifford bundle admits a natural faithful trace hence is contained in the von Neumann algebra $CL(M)$ given by the GNS construction, see section 4.1 for more details. 

Thus our notion of derivation triple requires additional conditions on the algebra of differential forms to admit a tracial state. On the other hand we allow for slightly more general derivations, because in many situations it is difficult to identify a good choice of $D$. In order to understand the role of differential forms, we recall Connes' abstract definition
\[  \Omega^{1}(\mathcal{A})=\{\sum_{j}(a_{j}\otimes b_{j}-1\otimes a_{j}b_{j})|a_{j}, b_{j}\otimes \mathcal{A}\}\subset \mathcal{A}\otimes \mathcal{A} \pl. \]
The induced differential representation $\pi_{\delta}: \Omega^{1}(\mathcal{A})\rightarrow \mathcal{M}$ is defined by
$$\pi_{\delta}(a\otimes b-1\otimes ab)=\delta(a)b,$$ and we denote the range $\pi_{\delta}(\Omega^{1}(\mathcal{A}))$ by $\Omega_{\delta}(\mathcal{A})$. Thus $\Omega_{\delta}(\mathcal{A})$ is Hilbert $\mathcal{A}$-bimodule with inner product
$$(\delta(a_{1})b_{1}, \delta(a_{2})b_{2})_{\mathcal{A}}=b_{1}^{*}E_{\mathcal{N}}(\delta(a_{1}^{*})\delta(a_{2}))b_{2},$$
where $E_{\mathcal{N}}: \mathcal{M}\rightarrow \mathcal{N}$ is the conditional expectation and $(\cdot, \cdot)_{\mathcal{A}}$ is the $\mathcal{N}$-valued inner product.  Indeed,
$\Omega_{\delta}(\mathcal{A})$ is also left $\mathcal{A}$-module since $a\delta(b)=\delta(ab)1-\delta(a)b$.

\begin{definition}\label{subtriple} Let $\triple$ and $\ltriple$ be two derivation triples with $\tilde{E}_{\mathcal{N}}:\tilde{\mathcal{N}}\to\mathcal{N}$. Let $\mathcal{N}_{\fix}\subset\mathcal{N}$ and $\tilde{\mathcal{N}}_{\fix}\subset\tilde{\mathcal{N}}$ be the fixed-point algebras of $e^{-t\delta^{*}\bar{\delta}}$ with the corresponding conditional expectations $E$ and $\tilde{E}$. We say $\triple$  is a sub-triple of $\ltriple$, denoted by $\triple\subset\ltriple$, if the first two diagrams are commuting and the last diagram is a commuting square. 
\begin{figure} 
\begin{tikzcd}[row sep=1.0em]
\tilde{\mathcal{N}}  \arrow[r,phantom, "{\subset}", description]  \arrow[d, phantom, "{\cup}", description]&  \tilde{\mathcal{M}}  \arrow[d,phantom, "{\cup}", description] \\
  \mathcal{N} \arrow[r,phantom, "{\subset}", description]& \mathcal{M}
\end{tikzcd}
\quad \quad
\begin{tikzcd}[row sep=1.0em]
   \Omega^{1}(\tilde{\mathcal{A}})   \arrow[d,phantom, "{\cup}",  description]     \arrow[r,"\tilde{\pi}_{\delta}" ]  & L_{2}(\tilde{\mathcal{M}},\tau)  \arrow[d,phantom, "{\cup}", description] \\
    \Omega^{1}(\mathcal{A})   \arrow[r,"\pi_{\delta}" ]  & L_{2}(\mathcal{M},\tau)
\end{tikzcd}
\quad \quad
\begin{tikzcd}[row sep=1.0em]
 \tilde{\mathcal{N}}_{\fix} \arrow[d,phantom, "{\cup}", description] \arrow[r,phantom, "{\subset}", description]&\tilde{\mathcal{N}}  \arrow[d, phantom, "{\cup}", description]\\
\mathcal{N}_{\fix} \arrow[r,phantom, "{\subset}", description]&\mathcal{N}
\end{tikzcd}

\end{figure}
\end{definition}
\begin{theorem} \label{passto} Let $\triple\subset\ltriple$, then $$\CLSI\triple\geq\CLSI\ltriple \quad \text{and} \quad  \CpSI\triple\geq\CpSI\ltriple .$$
\end{theorem}
\begin{proof}
Let $\iota_{\mathcal{N}}: L_{1}(\mathcal{N},\tau)\rightarrow L_{1}(\tilde{\mathcal{N}},\tau)$ be the trace preserving inclusion.
% By the commuting square conditions $\tilde{E}_{\mathcal{N}}\circ \iota_{\mathcal{N}}=id_{\mathcal{N}}$.
 Then we compare $D_{\mathcal{N}_{\fix}}(\rho)$ and $D_{\tilde{\mathcal{N}}_{\fix}}(\rho)$.
\begin{align*}
D_{\mathcal{N}_{\fix}}(\rho)=&D(\tilde{E}_{\mathcal{N}}\circ \iota_{\mathcal{N}} (\rho)\| \tilde{E}_{\mathcal{N}}\circ \iota_{\mathcal{N}}\circ E(\rho))
= D(\tilde{E}_{\mathcal{N}}(\iota_{\mathcal{N}}(\rho))\| \tilde{E}_{\mathcal{N}}(\tilde{E}(\rho)))\\
\leq & D(\iota_{\mathcal{N}}(\rho)\| \tilde{E}(\rho)) \text{ (by data processing inequality)}\\
=& D_{\tilde{\mathcal{N}}_{\fix}}(\iota_{\mathcal{N}}(\rho))
\end{align*}
Thanks to the first condition, we see that $\iota_{\mathcal{N}}(a\delta(f)b)=\iota_{\mathcal{N}}(a)\delta(\iota_{\mathcal{N}}(f))\iota_{\mathcal{N}}(b)$. This implies $I_{\delta}(\rho)=I_{\delta}(\iota_{\mathcal{N}}(\rho))$. The proof for $\CpSI$ is the same.
\end{proof}
% Suppose there exists a strongly coutinous semigroup $\hat{T}_{t}=e^{-tL}: \mathcal{M}\rightarrow \mathcal{M}$ of completely positive trace preserving maps with the generator $L$ such that for any $a,b\in\mathcal{A}$ $$\Upgamma_{L}(a,b)=E_{\mathcal{N}}(\delta(a^{*})\delta(b)).$$
%A linear operator $\Rc_{L}: \Omega_{\delta}(\mathcal{A})\rightarrow \mathcal{M}$ is called the \textit{Bakry-\'Emery (B-E) Ricci operator}  associated to $L$ if it satisfies the following conditions. \begin{enumerate}[nolistsep] \item (bimodule) $\Rc_{L}(a\rho b)=a\Rc_{L}(\rho) b, \text{ for any }a, b\in \mathcal{A}\text{ and } \rho \in \Omega_{\delta}(\mathcal{A})$; \item (commutator identity)  $\delta (A_{\delta} a)-L(\delta(a))=\Rc_{L}(\delta(a)), \text{ for any } x\in\mathcal{A}.$ \end{enumerate}
\noindent A linear operator $\Rc: \Omega_{\delta}(\mathcal{A})\to \mathcal{M}$ is called the \textit{{\rm{(}}geometric{\rm{)}} Ricci operator of $\triple$} provided that
\begin{enumerate}[leftmargin=6mm]
 \item $\Rc$ is bimodule over $\mathcal{A}$ \begin{align} \label{rc-b}\Rc(a\rho b)=a\Rc(\rho) b,\quad \forall a,b\in\mathcal{A},\rho\in\Omega_{\delta}(\mathcal{A});\end{align}
 \item there  exists a strongly continuous semigroup $\hat{T}_{t}=e^{-tL}: \mathcal{M}\rightarrow \mathcal{M}$ of completely positive trace preserving maps such that
$$\Upgamma_{L}(a,b)=E_{\mathcal{N}}(\delta(a^{*})\delta(b)),\quad \forall a,b\in\mathcal{N}.$$
\item  $\delta(\tilde{a})\in \dom(L)$ if there exists $a\in\mathcal{A}$ and $t\geq 0$  such that $\tilde{a}=T_{t}(a)\in\mathcal{A}$ and
\begin{align}\label{rc-c}
\delta (\delta^{*}\bar{\delta} \tilde{a})-L(\delta(\tilde{a}))=\Rc(\delta(\tilde{a})).
\end{align}
\end{enumerate}
The derivation $\delta$ is said to admit a Ricci curvature $\Rc\geq\lambda$ bounded below by a constant $\lambda$,  if $( \Rc(\rho),\rho )_{\mathcal{A}}\geq \lambda E_{\mathcal{N}}(\rho^{*}\rho)$ for any $\rho\in \Omega_{\delta}(\mathcal{A})$. We say the generator  $A$ of $T_{t}=e^{-tA}$ admits $\Rc\geq\lambda$ if there exists a derivation triple $\triple$ such that
$$\Gamma_{A}(a,b)=E_{\mathcal{N}}(\delta(a^{*})\delta(b)), \quad\forall a,b\in\mathcal{A}$$
and $\delta$ admits $\Rc\geq\lambda.$  It shall be noted that the choice of $\delta$ is not unique, thus we may find a larger Ricci lower bound of  $A$ by choosing a better $\delta$.
%$\Rc_{L}$ is said to be bounded below by a constant $\kappa$, i.e., $\Rc_{L}\geq\kappa$,  if $( \Rc_{L}(\rho),\rho )_{\mathcal{A}}\geq \kappa E_{\mathcal{N}}(\rho^{*}\rho)$ for any $\rho\in \Omega_{\delta}(\mathcal{A})$.
\subsection{Abstract Bakry-\'Emery criterion} We establish an operator-valued Bakry-\'Emery criterion relating the Ricci curvature and the log-Sobolev inequality.

% \textit{Bakry-\'Emery theorem} established a concise criterion $\Rc\geq \lambda$  for a probability measure $\mu$ to satisfy the $\lambda$-logarithmic Sobolev inequality. (\cite{B-E}) Furthermore, we show that $\Rc\geq \lambda$  guarantees the $\lambda$-complete logarithmic Sobolev inequalities.
\begin{theorem} \label{main} Let $\triple$ be a derivation triple  with $\Rc\geq \lambda>0$. Then $$\CLSI\triple\geq 2\lambda.$$
\end{theorem}
\begin{proof}
To simplify the notation, we denote $A_{\delta}$, $T_{t}(\rho)$, $\hat{T}_{t}(\sigma)=e^{-tL}\sigma$ by $A$, $\rho_{t}$, $\hat{\sigma}_{t}$, respectively in this proof, where $L$ is given in the definition of Ricci operator $\Rc$.  For any fixed $\rho\in \mathcal{N}_{+}$,  we may consider two functions:
\begin{align*} h(t)&=I_{\delta}(\rho_t)
\end{align*}
and 
\begin{align*}
k(t)&=\left\langle \hat{T}_{t}(\delta(\rho)), Q^{\rho_{t}}\left(\hat{T}_{t}(\delta(\rho))\right)\right\rangle.
\end{align*}
We compute the derivatives of $h$ and $k$ by Example \ref{CM} and obtain:
 \begin{align*}
 h'(t)=&-2\int_{\Bbb{R}_{+}} \tau\left(\delta(A\rho_{t})(r+\rho_{t})^{-1}\delta(\rho_{t})(r+\rho_{t})^{-1}\right) dr\\&-2\int_{\Bbb{R}_{+}}\tau\left(\delta(\rho_{t})(r+\rho_{t})^{-1}\delta(\rho_{t})(r+\rho_{t})^{-1}\delta(\rho_{t})(r+\rho_t)^{-1}\right) dr,
 \end{align*}
 and
 \begin{align*}
k'(t)=&-2\int_{\Bbb{R}_{+}} \tau\left(L\hat{T}_{t}(\delta(\rho))(r+\rho_{t})^{-1} \hat{T}_{t}(\delta(\rho))(r+\rho_{t})^{-1}\right) dr\\
&-2\int_{\Bbb{R}_{+}}\tau\left(\hat{T}_{t}(\delta(\rho))(r+\rho_{t})^{-1}\delta(\rho_{t})(r+\rho_{t})^{-1}\hat{T}_{t}(\delta(\rho))(r+\rho_t)^{-1}\right) dr.
\end{align*}
The key observation is that the second lines of both derivatives coincide at $t=0$. It remains  to compare the first lines: $$h'(0)-k'(0)=-2\int_{\Bbb{R}_{+}}\tau\left( (\delta A-L\delta)(\rho)(r+\rho)^{-1}
\delta(\rho)(r+\rho)^{-1} \right) dr.$$
Thanks to the commutator identity \eqref{rc-c}, the Ricci curvature $\Rc$ finally shows up
$$h'(0)-k'(0)=-2\int_{\Bbb{R}_{+}}\tau\left( \Rc(\delta(\rho))(r+\rho)^{-1}
\delta(\rho)(r+\rho)^{-1} \right) dr.$$
Let us define $\omega_{r}\lel (\rho+r)^{-1/2}\delta(\rho)(\rho+r)^{-1/2}$, then $\omega_{r}=\omega_{r}^{*}\in \Omega_{\delta}(\mathcal{A})$. By  \eqref{rc-b}, we have $$\Rc(\omega_{r})=(\rho+r)^{-1/2}\delta(\rho)(\rho+r)^{-1/2}.$$
We rewrite $h'(0)-k'(0)$ as the trace of $\mathcal{A}$-valued inner product  over $\Omega_{\delta}(\mathcal{A})$,
$$h'(0)-k'(0)=-2\int_{\Bbb{R}_{+}}\tau\left( (\Rc(\omega_{r}), \omega_{r})_{\mathcal{A}}\right)dr.$$
 Since Ricci curvature $\Rc$ is bounded below by $\lambda$,  we deduce that
\begin{align*}
h'(0)-k'(0)
\leq -2\lambda \int_{\Bbb{R}_{+}} \tau( E_{\mathcal{N}}( \omega_{r}^{*}\omega_{r})  )dr=-2\lambda h(0).
\end{align*}
As an application of Theorem \ref{HPconvex}, $k'(0)\leq 0.$ Indeed, by Example\ref{example-hiai}) $$k(t)=\left \langle \hat{T}_{t}(\delta(\rho)), Q_{f^{-1}_{[0]}}^{\rho_{t}}\left(\hat{T}_{t}(\delta(\rho)) \right)  \right \rangle$$ with the $f(x)=\frac{x-1}{\ln(x)}$.  Applying Theorem \ref{HPconvex} with $\beta=\hat{T}_{t}$ yields $$k(t)\leq k(0)$$ for $t\geq 0$.
Together with $k'(0)\leq 0$, we deduce that $$h'(0)\leq -2\lambda h(0).$$
This inequality remains true by replacing the initial state $\rho$ with $\rho_{s}$. Let us define $$h_{s}(t)=I_{\delta}(\rho_{t+s})$$ for fixed $s\geq 0$,  then $$h'_{s}(0)\leq -2\lambda h_{s}(0)=-2\lambda h(s).$$ Note that $h'_{s}(0)=h'(s)$, and consequently for any $s\geq 0$, $$h'(s)\leq -2\lambda h(s).$$ By Gr\"onwall's lemma, this implies the exponential decay of Fisher information $$h(t)\leq e^{-2\lambda t}h(0).$$ Using Lemma \ref{expo decay}, the theorem is established.
\end{proof}

 \noindent In the proof, we actually show that Fisher information of $\triple$ decays exponentially with the decay rate $2\lambda$, and this is a stronger condition than $\CLSI\triple\geq 2\lambda$ (Lemma \ref{expo decay}). See \cite{Led}. Theorem \ref{main} remains true for CpSI and $\CLSI^{+}$, and we refer to \cite{hao} for the proof.
\begin{theorem}
Let $\triple$ be a derivation triple  with $\Rc\geq \lambda>0$. For $p\in(1,2)$, then $$\CpSI\triple\geq 2\lambda.$$
Thus $\CLSI^{+}\triple\geq 2\lambda$.
\end{theorem}
\subsection{Connection to  $\lambda$-convexity}
Otto and Villani (\cite{OV}) show that the Ricci curvature on a Riemannian manifold $M$ is bounded below by $\lambda \in \Bbb{R}$ if and only if the entropy is geodesically $\lambda$-convex in the space of probability measures $P(M)$ endowed with the Kantorovich metric $W_{2}$. Carlen and Maas use this characterization as a starting point and define the lower bound of the Ricci curvature in the noncommutative setting through a \textit{transportation condition}, see \cite{cm20} for details. Indeed, the key ingredient of the characterization may be considered as the noncommutative adaptation of Bakry-\'Emery's $\Gamma_{2}$ condition.  We will indicate here that in finite dimension our geometric definition of  Ricci curvature bounded below implies the complete lower bound of transportation definition.
 %, which also extends to the discrete setting.  We refer to the Ricci curvature defined by this charaterization idea as \textit{transport Ricci operator} and  $\Rc_{L}$ in our setting as the \textit{geometric Ricci operator}.  These two definitions are closely related, and it turns out that the lower bound of geometric Ricci operator implies the lower bound of the transport Ricci operator.
\begin{theorem}[Carlen and Maas] \label{gradient-estimate} A differential structure $(\mathcal{A},\nabla,\sigma)$ has (transportation) Ricci curvature  bounded from below by $\lambda>0$ if and only if the following gradient estimate holds for $\rho\in\mathfrak{B}$, $a\in \mathcal{A}_{0}$ and $t\geq 0$:
\begin{align} \label{cm-gradient}
\|\nabla \mathscr{P}_{t}a\|_{\rho}^{2}\leq e^{-2\lambda t}\| \nabla a\|_{\mathscr{P}_{t}^{\dagger}\rho}^{2}.
\end{align}
\end{theorem}
Let us point out that we have chosen
$\nabla=\delta$, $\mathscr{P}_{t}=T_{t}=e^{-tA}$ and $\mathscr{P}_{t}^{\dagger}=\hat{T}_{t}=e^{-tL}$ in our setting. However, the results remain true for every other choice of $\delta$ as well, as long as $\Gamma_{A}(x,y)=E(\delta(x)^*\delta(y))$ is still satisfied.
 % It is easy to see that $\nabla$, $\mathscr{P}_{t}^{\dagger}$, $\mathscr{P}_{t}$ correspond to $\delta$, $\hat{T}_{t}$, $T_{t}$, respectively, in our setting.
Carlen and Mass did not consider a semigroup acting on the space of differential forms. The $\rho$-inner $\|\cdot\|_{\rho}$ product can be interpreted as
$$\|\sigma\|^{2}_{\rho}=\tau\left(\sigma Q_{f_{[0]}}^{\rho}(\sigma) \right),$$ where $f(x)=\frac{x-1}{\ln(x)}$.
\begin{theorem} Let $A=\delta^{*}\bar{\delta}$ be a (finite dimensional) generator  over $\triple$ with geometry Ricci curvature $\Rc$ bounded below by $\la$. Then $A$ also satisfies
(\ref{cm-gradient}), equivalently a lower bound on the transportation Ricci curvature.
\end{theorem}
%\begin{theorem} Let $\triple$ be a derivation triple, and let $L:\mathcal{M}\rightarrow \mathcal{M}$ be a generator of the completely positive trace preserving semigroup $\hat{T}_{t}=e^{-tL}$ such that for any $a,b\in\mathcal{A}$ $$\Upgamma_{L}(a,b)=E_{\mathcal{N}}(\delta(a^{*})\delta(b)).$$ Then $\Rc_{L}\geq \lambda>0$ implies (\ref{cm-gradient}).\end{theorem}
\begin{proof}
Again, we denote  $T_{t}(\rho)$, $\hat{T}_{t}(\sigma)=e^{-tL}\sigma$ by $\rho_{t}$, $\hat{\sigma}_{t}$, respectively, with $L$ given in the definition of Ricci operator $\Rc$.
We follow the spirit proof of Theorem \ref{gradient-estimate} and define $F:[0,t]\rightarrow \Bbb{R}$ for any fixed $t>0$,
$$F(s)=e^{-2\lambda s}\tau\left(  \delta(T_{t-s}(\sigma)) Q_{f_{[0]}}^{\rho_{s}}( \delta(T_{t-s}(\sigma)) ) \right).$$
If $F'(s)\geq 0$ for any $s\geq 0$, then $F(0)\leq F(t)$ yields the assertion. In the rest of the proof, we want to show that $F'(s)\geq 0$.  Differentiating $F(s)$, we obtain that
\begin{align*}
F'(s)=-&2\lambda F(s)+2e^{-2\lambda s} \tau\left(  \delta(A(\sigma_{t-s})) Q_{f_{[0]}}^{\rho_{s}}(\delta(\sigma_{t-s})) \right)\\
+&e^{-2\lambda s}\tau\left(\delta(\sigma_{t-s})  D({Q_{f_{[0]}}^{\rho_{s}}})(\delta(\sigma_{t-s})) \right),
\end{align*}
where $D({Q_{f_{[0]}}^{\rho_{s}}})$ is given by differentiating ${Q_{f_{[0]}}^{\rho_{s}}}$ in terms of $s$.
 It is convenient to define two more functions $h$ and $k$ for fixed $\rho$ and  $\sigma$:
\begin{align*} h(x)\lel \tau \left(\hat{T}_{x}(\delta(\sigma_{t-s})) Q_{f_{[0]}}^{\rho_{s}} \hat{T}_{x}(\delta(\sigma_{t-s}))  \right)
\end{align*}
and
\begin{align*}
k(x)\lel \tau\left(  \delta(\sigma_{t-s}) Q_{f_{[0]}}^{\rho_{s+x}}(\delta(\sigma_{t-s}))  \right).
\end{align*}
Plugging $f(x)=\frac{x-1}{\ln(x)}$ into Theorem \ref{HPconvex} and Lemma \ref{dc} implies that $h(x)\leq k(x)$. Noting that $h(0)=k(0)$, we infer that ${h'}(0)\leq {k'}(0)$, i.e.,
\begin{equation}\label{hk-2}-2\tau\left( L(\delta(\sigma_{t-s})) Q_{f_{[0]}}^{\rho_{s}}(\delta(\sigma_{t-s}))    \right) \leq \tau\left(  \delta(\sigma_{t-s}) D(Q_{f_{[0]}}^{\rho_{s}})(\delta(\sigma_{t-s}))  \right).
\end{equation}
Applying \eqref{rc-c} implies that
\begin{equation}\label{ricci-2}
2\tau\left( (\delta(A(\sigma_{t-s}))-\Rc(\delta(\sigma_{t-s}))) Q_{f_{[0]}}^{\rho_{s}}(\delta(\sigma_{t-s}))   \right)+\tau\left(  \delta(\sigma_{t-s}) D(Q_{f_{[0]}}^{\rho_{s}})(\delta(\sigma_{t-s}))  \right)  \geq 0
\end{equation}
Since Ricci curvature is bounded below by $\lambda$, we deduce that
\begin{equation}
\label{ricci-3}
\tau\left(\Rc(\delta(\sigma_{t-s}))Q_{f_{[0]}}^{\rho_{s}}(\delta(\sigma_{t-s}))\right) \geq \tau\left( \delta(\sigma_{t-s}) Q_{f_{[0]}}^{\rho_{s}}(\delta(\sigma_{t-s})) \right).
\end{equation}
Putting pieces (\ref{hk-2}), (\ref{ricci-2}) and (\ref{ricci-3}) together, we obtain $F'(s)\geq 0$.
\end{proof}
\begin{rem}{\rm We refer to \cite{Wirth} for a discussion of lower bounds  on the transportation Ricci curvature in infinite dimension.}
\end{rem}
\section{Geometric Applications}
\subsection{Clifford bundle}
Let us recall that the Clifford algebra $C\ell_{n}$ is generated by $n$ self-adjoint unitaries $\{e_{k}\}_{k=1}^{n}$ satisfying $$e_{k}^{*}=-e_{k},\quad e_{k}^{2}=-1, \quad \text{and}\quad e_{k}e_{l}=-e_{l}e_{k}\quad \text{for} \quad k\neq l.$$ Equivalently, we may use the Clifford function from real Hilbert spaces to von Neumann algebras $c: \Hs\to C\ell(\Hs)$ such that $c(h)$ is self-adjoint and
$$c(h)c(k)+c(k)c(h)\lel -2 (h,k),$$
where $(,)$ is the inner product over $\Hs$.
Let  $(M, g)$ be an $n$-dimensional smooth Riemannian manifold without boundary. Let $\mu$ be the probability measure defined by $d\mu=\frac{1}{\Zu} e^{-U}dvol$ over the manifold $M$ with  $\Zu=\int_{M} e^{-U(x)}dvol(x).$ We may consider $C\ell_x\cong C\ell_n$ the Clifford algebra generated $\{c(e_k)\}$, where $\{e_k\}_{k=1}^{n}$ is an orthogonal basis of the cotangent space $\T^*_{x}M$ at a point $x\in M$. Recall that $C\ell_x$ and $\T^*_{x}M$ are also identified as vector spaces, see \cite{chow}, \cite{lm}. We denote by $C_0(C\ell(M))$ ($C_0^{\infty}(C\ell(M))$) the space of continuous (respectively smooth) sections vanishing at infinity. %\textcolor{red}{ Thanks to the Clifford multiplication, $C_0(C\ell(M))$ is a $C^*$-bundle (see \textcolor{red}{\cite{Blanchard}}) over $C(M)$.}
Let $CL(M)$ be the von Neumann algebra generated by the GNS construction with respect to the trace $\tau(a)=\int \tau_x(a(x))d\mu(x)$,  %$\tau(a)=\int \tau_x(a(x))d{\rm vol}(x)$
where $\tau_x$ is the unique trace satisfying $\tau(c_{k_1}\cdots c_{k_m})=0$ for $m\le n$ and mutually different indices $1\le k_1,\dots, k_m\le n$.

We now explain the \textit{derivation triple for the Riemannian manifold}.  Let $\nabla$ be the Levi-Civita connection (covariant derivative) and $\{X_{1}, \dots, X_{n}\}$ be an orthonormal basis of the tangent space $\T M$, then $\nabla_{X_{k}}: C^{\infty}(M)\to C^{\infty}(M)$ defines a family of differential operators. We may combine $\{\nabla_{X_{k}}\}$ and $\{e_{k}\}$ and define
\begin{align}
\delta(f)\lel \sum_{k=1}^{n} e_{k} \nabla_{X_{k}}  (f), \quad \forall f\in C^{\infty}(M). \label{lap-der}
\end{align}
 It is obvious that $\delta:C_{0}^{\infty}(M)\to L_{2}(CL(M),\mu)$ is a $*$-preserving closable derivation.  Thus we obtain the derivation triple
\begin{align}
\triple=(L_{\infty}(M)\ssubset CL(M),\mu, \delta) \label{triple-m}
\end{align}
and identify $\mathcal{A}=C_{0}^{\infty}(M)$.

The \textit{$\mu$-modified Laplace-Beltrami operator}  $\Du:C^{\infty}(M)\rightarrow C^{\infty}(M)$ is defined  by
\begin{align}\label{lbu} \Du f= \sum_{i=1}^{n} X^{*}_{i} X_{i}  f+\nabla U\cdot \nabla f,
\end{align}
 where $\nabla U$ is the gradient of $U$.  It is well-known that $\Du$ is essentially self-adjoint in $L_{2}(M,\mu)$. Moreover, we have $\Du=\delta^{*}\delta$.
The \textit{extended Levi-Civita connection} to Clifford bundle $C\ell(M)$ remains a derivation with respect to Clifford multiplication (see e.g. \cite{lm}), i.e.
$$\nabla_{X}(f\cdot g)=f\cdot( \nabla_{X}g)+ (\nabla_{X}f)\cdot g,\quad \forall X\in C^{\infty}(\T M), f,g\in C^{\infty}(C\ell(M)),$$
where $\cdot$ denotes the Clifford multiplication.
Let $\Lu:L_{2}(CL(M),\mu)\rightarrow L_{2}(CL(M),\mu)$ be the \textit{$\mu$-modified rough (or Bochner) Laplacian}
\begin{align} \label{rlbu} \Lu=\sum_{i=1}^{n}\left(\nabla_{X_{i}}\nabla_{X_{i}}-\nabla_{\nabla_{X_{i}}X_{i}}\right)+\nabla_{\nabla U}.
\end{align}
\begin{lemma} The $\mu$-modified rough Laplacian $\Lu$ is a generator of the completely positive trace preserving semigroup $\hat{T}_{t}=e^{-t\Lu}.$
\end{lemma}
\begin{proof}
We combine the family of differential operators $\nabla_{X_{k}}: C^{\infty}(C\ell(M))\to C^{\infty}(C\ell(M))$ with an additional
(Mayorana)-Clifford operators $\tilde{e}_{k}\in \Mz_{2^{n}}$ such that
$$\tilde{e}_{k}=-\tilde{e}_{k}^{*}, \quad \tilde{e}_{k}^2=-1, \quad \text{and} \quad \tilde{e}_{k}\tilde{e}_{j}=-\tilde{e}_{j}\tilde{e}_{k} \quad \text{for} k\neq l$$
and define
$$\tilde{\delta}(f)=\sum_{k=1}^{n} \tilde{e}_{k}\otimes\nabla_{X_{k}}(f).$$
Note that $\tilde{\delta}$ a $*$-preserving closable derivation. Thus
$\Lu={\tilde{\delta}}^{*}\bar{\tilde{\delta}}$ yields the assertion.
\end{proof}
\noindent Let $\Rcu: C^{\infty}(\T^{*}M)\to C^{\infty}(\T^{*}M)$ be the  Bakry-\'Emery Ricci $(1,1)$-tensor $$\Rcu=\Rc+\nabla_{\nabla U}.$$
\begin{lemma} The derivation triple defined by \eqref{triple-m} admits a Ricci curvature $\Rcu$ with the corresponding strongly continuous semigroup $\hat{T}_{t}=e^{-t\Lu}$.
\end{lemma}
\begin{proof}
The Bakry-\'Emery Ricci $\Rcu$ is bimodule over $\mathcal{A}=C^{\infty}_{0}(M)$ since $\Rcu$ is a $(1,1)$-tensor. By the \textit{Bochner Weitzenb\"ock formula}, we obtain that
\begin{equation} \label{BW}
\delta(\Du f)=\Lu(\delta f)+\Rcu(\delta f),  \forall  f\in C^{\infty}(M).
\end{equation}
We can also identify the restriction of $\Lu$ to $C^{\infty}(M)$ with $\Du$
$$\Lu\: |_{C^{\infty}(M)}=\Du.$$
Thus $\Rcu$ is a Ricci operator of $\triple.$
\end{proof}
\noindent We adopt the geometric convention of Clifford algebra that $e_{k}^{2}=1$ and $e^{*}_{k}=-e_{k}$. The convention of operator algebra is to use $c_{k}^{2}=-1$ and $c_{k}^{*}=-c_{k}$. It is easy to converse between the two versions by using $e_{k}=ic_{k}$.

\subsection{Complete Bakry-\'Emery theory}
We recapture the Bakry-\'Emery criterion for complete log-Sobolev inequality as a corollary of Theorem \ref{main}, and this result motivated our definition of derivation triple.
\begin{theorem} [Complete Bakry-\'Emery theorem]
Let $(M, g, \mu)$ be a smooth Riemannian manifold with the measure $\mu$ defined by $d\mu=\frac{1}{\Zu} e^{-U} dvol$ with $\Zu=\int_{M} e^{-U}dvol$ for $U\in C^{\infty}(M)$. Given that  $\Rcu\geq \kappa>0$, then
$$\CLSI(\Du)\geq \CLSI^{+}(\Du)\geq 2\kappa.$$
\end{theorem}
\noindent Holley and Stroock \cite{hs87} proved that the log-Sobolev inequality is stable under measure perturbation, and this property remains true for the complete log-Sobolev  inequality by Theorem \ref{sutp}, as far as a central change of measure is concerned. Changing from a trace to the state is distinctly more complicated, and will not be considered in this paper. We refer to \cite{Led} for more applications.

\begin{corollary}
Let $\nu$ be the probability measure defined by $d\nu=\frac{1}{\Zv}e^{-V}dvol$ with $V\in C^{\infty}(M)$, where $\Zv$ is the normalization factor.  If $\|U-V\|_{\infty}\leq C$, then
$$e^{2C}\CLSI(\Du)\geq  \CLSI(\Dv)
%\geq \frac{\CLSI(\Du)}{e^{2C}}
\quad \text{and}  \quad e^{2C}\CLSI^{+}(\Du)\geq \CLSI^{+}(\Dv)
 %\geq \frac{\CLSI^{+}(\Du)}{e^{2C}}
 .$$
\end{corollary}

\subsection{Examples and applications}
For illustration purposes, we  discuss some interesting examples of derivation triples and  applications to Lindblad operators.% mainly by the change of measure principle and transference principle.
\begin{example}The Laplace-Beltrami operator of any compact Riemannian manifold with a strictly positive Ricci curvature satisfies $\CLSI$ and $\CLSI^{+}$, such as orthogonal group $O(n)$, special orthogonal group $SO(n)$, and  spheres $S^{n}$ with $n\geq 2$.
\end{example}
\begin{example} \label{Gaussian} Let $d\gamma(x)=(2\pi)^{-n/2} e^{-U(x)}dx$  be the Gaussian measure of $\Bbb{R}^{n}$, where $U(x)=|x|^{2}/2$. Then we have the natural Bakry-\'Emery Ricci $\Rcu=I_{d}$. Thus $$\CLSI(\Du)\geq \CLSI^{+}(\Du)\geq 2.$$
% it implies that $\Du$ satisfies $2$-$\CLSI$.
\end{example}

\begin{example} \label{unit} Let $dx$ be the Lebesgue measure over the $1$-dimensional manifold $(0,1)$ and $\delta$  be the ordinary pointwise derivative.
% on the bounded differentiable functions.
Then we have a derivation triple $\triple$, and  $$\CLSI\triple \geq \CLSI^{+} \triple \geq \frac{4}{5}.$$
%then $\triple$ satisfies $\frac{4}{5}$-$\CLSI$.
\end{example}
\begin{proof} Let us consider the measure $$d\mu(x)=\frac{1}{\sqrt{2\pi}}\sum_{k=-\infty}^{\infty} e^{-(x-k)^{2}/2}dx$$ and the embedding map $$\pi: L_{\infty}(0,1)\rightarrow L_{\infty}(\Bbb{R}), \pi(f)(x)\lel f(x\mod 1).$$ Note that we get a sub-triple of $L_{\infty}(\Bbb{R})$, then $$\CLSI(\mu,\delta)\geq 2.$$
Note that $$\frac{\sqrt{2\pi}}{2+2e^{-1/2}+2e^{-2}+\frac{8}{3}e^{-9/2}}\leq \frac{dx}{d\mu} \leq\frac{\sqrt{2\pi}}{2e^{-1/2}+2e^{-2}+2e^{-9/2}+\frac{48}{125}e^{-25/2}}.$$  For approximation details, see Appendix. Together with Theorem \ref{sutp}, it implies that
$$\CLSI\triple \geq \frac{4}{5}.$$
Similarly $\CLSI^{+}\triple\geq\frac{4}{5}. $
\end{proof}

\begin{example}\label{sphere} Let $\Delta$ be the Laplace-Beltrami operator over the 1-dimensional sphere $S^{1}\subset\Bbb{R}^{2}$ with the Haar measure $\mu$, then $$\CLSI(\Delta)\geq\CLSI^{+}(\Delta)\geq \frac{1}{5\pi^{2}}.$$
\end{example}
\begin{proof}
For $g\in C^{\infty}(S^{1})$, then $\Delta (g)=\frac{d^{2}g}{d\theta^{2}}$. Let $\theta=2\pi x$ for $x\in (0,1)$ and $f(x)=g(2\pi x)$, and we have
$f''(x)=4\pi^{2}\Delta (g)(\theta)$. Let $\delta$ be the ordinary pointwise derivative over $(0,1)$ and $\tilde{E}$ be the conditional expectation mapping $L_{\infty}(0,1)$ onto constant  valued functions.
It is obvious that
$$D(f\| \tilde{E}(f))\lel D(g\| E(g))\quad \text{and} \quad I_{\delta}(f)\lel I_{\Delta}(g).$$
Together with Example \ref{unit}, we obtain
$$D(g\|E(g))\kl5\pi^{2}I_{\Delta}(g).$$
Similarly  $\CLSI^{+}(\Delta)\geq \frac{1}{5\pi^{2}}$.
\end{proof}

\noindent Together with the transferred argument in Section 4.3 of \cite{LJR}, Example \ref{sphere} implies the $\CLSI$ of Lindblad operator with $1$-generator.

\begin{example} \label{lind-1}
Let $L(\rho)=[x,[x,\rho]]=x^{2}\rho+\rho x^{2}-2\rho x\rho$, where $x$ is self-adjoint with discrete spectrum in $\Bbb{Z}$. Then $$\CLSI(L)\geq \CLSI^{+}(L)\geq \frac{1}{5\pi^{2}}.$$
\end{example}
\subsection{Gaussian Example}
Due to initial observation of Meyer, Bakry and \'{E}mery  discovered that the Ornstein-Uhlenbeck semigroup has Ricci curvature $1$. Thanks to central limit type results this applies to all Gaussians functions including tracial Fermionic random variables, see also \cite{CM}. In the context of Clifford algebras and $\ell_{\infty}^2$, this observation, and the connection to canonical derivatives were discovered by \cite{LP}. Junge and Zeng established a \textit{{Gaussians} transference} in \cite{JZ}. Let us illustrate this in the Fermionic case, and refer to \cite{JZ} for the more general setup.  Let $C\ell_{\nz}$ be the Cliford algebra generated by  a sequence of self-adjoint generators $\{c_{k}\}_{k\in\nz}$ satisfying
 \[ c_kc_j\lel -c_jc_k \pl ,\pl c_k^*=c_k\pl ,\pl c_k^2=1  \pl .\]
The number operator $N:C\ell_{\nz}\to C\ell_{\nz}$ is defined by
 \[ N(c_{i_1}\cdots c_{i_k}) \lel k c_{i_1}\cdots c_{i_k} \]
whenever $i_1,...,i_k$ are all different.  There exists a derivation $\delta_{N}$ such that $N=\delta_{N}^{*}\bar{\delta}_{N}$ and the derivation triple $(C\ell_{\nz}, C\ell_{\nz^{2}}, \delta_{N})$ admits a Ricci curvature $\Rc=id$.  Using a bijection between $\nz$ and $\nz\times \nz$, we can assume that $c_{(j,k)}$ are anti-commuting  Clifford generators of $C\ell_{\nz^{2}}$. Let $\tau$ be the trace of $C\ell_{\nz^{2}}$ obtained from GNS construction, then
$\tau(c_{i_{1}}\dots c_{i_{k}} )=0$ if $i_{1},\dots, i_{k}$ are all different.
For any fixed $m\in \nz$, we define
$$ u_m(c_k) \lel m^{-1/2} \sum_{j=1}^m c_{j,k}\ten g_{j} $$
where $\{g_j\}$ are iid Gaussians. Let $\om$ be an ultrafilter on $\nz$ and $M^{\om}=(C\ell_{\nz^2}\ten L_{\infty}(\rz^m,\gamma_m))^{\om}$ be the von Neumann algbraic ultraproduct with a normal faithful trace $\tau_{\om}$, where $d\gamma_m(x)=(2\pi)^{-m/2}e^{-|x|^2/2}dx$. Let $u_{\om}=(u_m(c_k))^{\bullet}$ be the limit object in the ultraprower. The central limit theorem shows that
 \[ \tau_{\om}(u_{\om}(c_{k_1})\cdots u_{\om}(c_{k_d})) \lel \tau(c_{k_1}\cdots c_{k_d})\pl. \]
This means that the map $\pi$ defined by
 \[ \pi(c_{k_1}\cdots c_{k_d})
 \lel u_{\om}(c_{k_1})\cdots u_{\om}(c_{k_d}) \]
extends to a trace preserving $^*$-homomorphism of  $C\ell_\nz$ into $M^{\om}$. Moreover, thanks to \cite{JZ}, we note that
 \[ \pi(T_t(w_{A})) \lel (id\ten T_t^{G})^{\bullet}(\pi(w_{A})), \quad \forall w_A=c_{i_1}\cdots c_{i_k},\]
where $T_t^{G}$ is the the Ornstein Uhlenbeck generator corresponding to the measure $d\gamma_m(x)$ (see \eqref{lbu} and \eqref{rlbu}). Since the latter has Ricci curvature $1$, the same is true for $T_t^{CL}$. We refer to \cite{JZ} for the explicit derivation $\delta(c_{k})\lel c_{k}$ for $\delta: C\ell_\nz\to C\ell_{\nz^{2}}$.

Let us point out a special case.
Let
$X=\left(\begin{smallmatrix}0&1\\ 1& 0\end{smallmatrix}\right), Y=\left(\begin{smallmatrix} 0 &-i\\i &0\end{smallmatrix}\right), \text{ and } Z=\left(\begin{smallmatrix} 1 & 0\\ 0 &-1\end{smallmatrix}\right)$. Then $C\ell_2$ is generated by $c_1=X$ and $c_2=Y$. Then $c_1c_2=-iZ$ implies that
 \[ T_t(\al 1+\beta X+\gamma Y+\zeta Z)
 \lel \al 1+e^{-t}\beta X+e^{-t\beta}Y+e^{-2t}\zeta Z \]
is a semigroup on $C\ell_2=\Mz_2$ which admits constant curvature $1$ and $\CLSI(T_t)=\CLSI^+(T_t)=2$. Here the spectral gap estimate \eqref{introgap} is tight. It will be interesting for us to rewrite the generator of this semigroup differently.
For a self-adjoint element $x$, we define $$ L_x(\rho) \lel x^2\rho+\rho x^2-2x\rho x \pl.$$
Let $a=\frac{X}{2}$ and $b=\frac{Y}{2}$, we find that
 \[ L_{a}(X)=0 \pl , \pl  L_{a}(Y)=Y \pl ,\pl L_{a}(Z) \lel Z \pl\]
  and  \[ L_{b}(X)=X\pl ,\pl L_{b}(Y)=0\pl,\pl L_{b}(Z)\lel Z \pl .\]
Thus $N=L_{a}+L_{b}$ has Ricci-curvarture $1$.  Then $$\CLSI(L_{a}+L_{b})\gl \CLSI^{+}(L_{a}+L_{b})\geq 2.$$  Noting $N$ restricted to $\ell_{\infty}^{2}$ is exactly of the form $I-E$, then  $\CLSI(N) \leq 2\lambda_{2}(N)\leq 2$.
\begin{cor} \label{bi} $\CLSI^+(L_a+L_b)=\CLSI(L_a+L_b)=2$. \end{cor}

At the time of this writing the exact $\CLSI$ constant for $A_n=id-E_{\tau}$, where $E_{\tau}(x)=\frac{\tr(x)}{n}1$, is given by the normalized trace is unknown. According to \cite{Ba}, we have $\CLSI(A_n)\gl 1$.

\begin{cor} $2\gl\CLSI(A_{2})\gl\CLSI^+(A_2)\gl \frac{3}{2}$.
\end{cor}

\begin{proof} The role of $X$, $Y$ and $Z$ can be interchanged in the argument above. By introducing $c=\frac{Z}{2}$, we find that
 \begin{align*}
 6D(\rho\|E_{\tau}(\rho)) &\le I_{L_a+L_b}(\rho)+I_{L_a+L_c}(\rho)+I_{L_b+L_c}(\rho) \\
 &\lel 2I_{L_{a}+L_b+L_c}(\rho) \\
 &\lel 4 I_{I-E}(\rho)\pl .
 \end{align*}
The last equality $L_a+L_b+L_c=2(\text{id}-E_{\tau})$ is easily checked on the Pauli basis. The same argument works for
$\CLSI^+$.
 \qd

\section{Anti-transference}
In this section we consider a finite dimensional compact Lie group $G$ with a normalized Haar measure $\mu$.  Let $\lgg$ be the Lie algebra of $G$ and $\lgg_{\cz}$ be the complexification of $\lgg$.  Let $X=\{X_{1},X_{2},\dots, X_{m}\}\subset \lgg$ be a H\"ormander system (see definition in the introduction), and we consider the sub-Laplacian operator $\Delta_{X}$ given by
$$\Delta_{X}(f)=\sum_{k=1}^{m}X_{k}^{*}X_{k}f,$$
where $X_{k}f(g)=\frac{d}{dt}f(\exp(tX_{k})g)|_{t=0}$ for any $f\in C^{\infty}(G)$ and $g\in G$.
Similar to (\ref{lap-der}), we can find a closable $*$-preserving derivation $\delta$ such that $\Delta_{X}=\delta^{*}\bar{\delta}$. We again identify the derivation triple $$\triple=(L_{\infty}(G,\mu)\ssubset CL(G),\mu,\delta).$$ Let $u: G\rightarrow U(\Hs)$ be a strongly continuous unitary representation  of the Hilbert space $\Hs$. We may transfer $S_{t}\lel e^{-t\Delta_{X}}$ to $\BB(\Hs)$ and obtain a strongly continuous semigroup $T_{t}^{\Hs,u}: \BB(\Hs)\rightarrow \BB(\Hs)$ of self-adjoint, trace preserving, and completely positive maps by using the (co-)representation map
\begin{align} \label{co-rep}
\pi: \BB(\Hs)\rightarrow L_{\infty}(G,\BB(\Hs)), \pi(a)(g)=u(g)^{*}au(g).
\end{align}
Let $E$ and $E_{T}$ be conditional expectations onto the fixed-point algebras of $S_{t}$ and $T_{t}^{\Hs,u}$, respectively. The co-representation $\pi$ allows for the commuting diagram below.
 \[ \begin{array}{cccccc} L_{\infty}(G,\BB(\Hs)) & \stackrel{e^{-t\Delta_X\ten id}}{\rightarrow } & L_{\infty}(G,\BB(\Hs)) & \stackrel{E}{\rightarrow } &\BB(\Hs) \\
 \uparrow_{\pi} & & \uparrow_{\pi} & & \uparrow_{\pi}   \\
 \BB(\Hs) & \stackrel{T_t^{\Hs}}{\rightarrow } & \BB(\Hs)&  \stackrel{E_{T}}{\rightarrow }  & {\BB(\Hs)}_{\fix},
 \end{array}  \] Indeed, for any $\exp(tX_{k})$, we may find the corresponding one-parameter group generated by some self-adjoint operator $x_{k}$   (\cite{KR1}) such that
$$u(\exp(tX_{K}))=\exp(itx_{k}).$$
For a finite dimensional Hilbert space $\Hs$, the self-adjoint matrix $x_{k}$ is bounded. The differential operator $X_{k}$ is transferred to a commutator operator via the co-representation,
 \[ X_k(\pi(a))(g)
 \lel \frac{d}{dt} u(g)^*e^{-itx_k}ae^{itx_k}u(g)|_{t=0}
 \lel -i \pi( [x_k,a])(g) \pl .\]
Note that the generator of $T_{t}^{\Hs}$ is a Lindblad operator. Indeed,
\[ L_{X}^{\Hs,u}(a) \lel \sum_k x_k^2a+ax_k^2-2x_kax_k \pl.\]
We denote the transferred semigroup and the Lindblad operator by $T_{t}$ and $L_{X}$ if there is no ambiguity. To emphasize the underlying Hilbert space (unitary representation), we also use $T_{t}^{\Hs}$ and $L_{X}^{\Hs}$ ($T_{t}^{u}$ and $L_{X}^{u}$, respectively).
 \begin{theorem} [transference principle] Let $S_{t}=e^{-t\Delta_{X}}: L_{\infty}(G,\mu)\rightarrow L_{\infty}(G,\mu)$ be the semigroup generated by the sub-Laplacian $\Delta_{X}$ and $T_{t}$ be the transferred semigroup defined as above, then $$\CLSI(L_{X})\geq \CLSI(\Delta_{X}).$$
\end{theorem}
\noindent The transference principle extends to $\CLSI^{+}$ naturally (\cite{hao}), then $\CLSI^{+}(L_{X})\geq \CLSI^{+}(\Delta_{X})$. The transference principle is developed for any ergodic and right-invariant semigroup $S_{t}$ over a compact Lie group, see \cite{LJR} for details.

\begin{theorem}\label{Rm} Let $T_{t}=e^{-t\L}: \BB(L_{2}(M))\to \BB(L_{2}(M))$ be a semigroup of self-adjoint completely positive trace preserving maps such that
% Let $\L$ on $\BB(L_2(M))$ be the generator of $T_{t}=e^{-t\L}$ such that
 \begin{enumerate}[leftmargin=6.0mm]
   \item[(1)] $\L$  maps $C^{\infty}(M)$ to $C^{\infty}(M)$;
  \item[(2)] $\dom(\L|_{C^{\infty}(M)})$ is dense in $L_{\infty}(M)$;
 \item[(3)]  $\CLSI^+(\L)> 0;$
 \item[(4)] $\L(d^{-\al}ad^{-\beta})=d^{-\al}\L(a)d^{-\beta}$ for any $0$-th order pseudo differential operator $a$.
  \end{enumerate}
Then $\CLSI^{+}(L_{\infty}(M),\L)\geq \CLSI^{+}(\BB(L_{2}(M)), \L)$.
\end{theorem}
 %%% We failed in proving this conjecture by using a similar argument for Theorem \ref{anti}. To apply the tool from noncommutative geometry directly, we shall penetrate into a rather exotic type of $\zeta$-functionals (\textcolor{red}{defined later, label it later}) defined by $\log$ functions. This also develops a little bit motivation to study an exented version of $\zeta$-functionals for noncommutative geometers. (\textcolor{red}{This paragraph is new. I want to point out why we failed in proving the conjecture and also proposing an open problem.})

\subsection{Noncommutative Geometry}
% In this section, we briefly introduce the pseudo-differential operators and the noncommutative residue to formulate the Connes' trace theorem. See \cite{Ruzhansky2010}(chapter 2) and \cite{LSZ} for pseudo-differential operators and noncommutative residue. For readers interested in the origin and development of noncommutative goemetry, we suggest \cite{Con94}.
\subsubsection{Connes' trace theorem}
We use the standard multi-index notation $\alpha=(\alpha_{1},\dots, \alpha_{n})$ with $\alpha_{j}\in\{0,1,2,\dots\}$ and  $|\al|=\sum_{j=1}^n |\al_j|$. Let  us denote the partial derivative by $D^{\al}_{x}f=\frac{\partial^{|\alpha|}f}{\partial_{x_{1}}^{\alpha_{1}} \dots\partial_{x_{n}}^{\alpha_{n}} }$, for $f\in C^{\infty}(\Bbb{R}^{n})$ .  We say that $\sigma\in C^{\infty}(\Bbb{R}^{n}\times \Bbb{R}^{n})$ is a \textit{symbol of order $m$} if
$$|D^{\alpha}_{x}D^{\beta}_{\xi}\sigma(x,\xi)|\leq   C_{\alpha,\beta}((1+|\xi|)^{m-\beta})$$
for any multi-index $\alpha $ and $\beta$ and for any $x,\xi\in \Bbb{R}^{n}$.
Note that $ C_{\alpha,\beta}$ depends on the choice of $\alpha$ and $\beta$ and is independent of $x$ and $\xi$. A \textit{pseudo differential operator $\Psi$ of order m} on $\rz^n$ with symbol $\sigma_{\Psi}$ is defined by
\begin{equation*}
\Psi(f) (x)\lel \int_{\Bbb{R}^{n}} e^{2\pi i x\cdot \xi } \sigma_{\Psi}(x,\xi) (\mathcal{F}f)(\xi) d\xi,
\end{equation*}
where $\mathcal{F}$ is the Fourier transformation and $\sigma_{\Psi}$ is a symbol of order $m$. Let $(M,g)$ be a Riemannian manifold of dimension $n$.  Then $\Psi:C^{\infty}(M)\to C^{\infty}(M)$ is said to be a pseudo-differential operator of order $m$ if this is true for every local chart. (See e.g. \cite{LSZ12}) The most prominent example of pseudo-differential operators  is the Laplacian operator $\Delta=-\left(\frac{\partial^{2}}{\partial_{x_{1}}^{2}}+\dots+\frac{\partial^{2}}{\partial_{x_{n}}^{2}}\right)$ on $\Bbb{R}^{n}$ with the symbol $\sigma_{\Delta}(x,\xi)=4\pi^{2}|\xi|^{2}$. Consequently, the Laplacian power $\Psi_{m}=(1+\Delta)^{m/2}$ is a pseudo-differential operator of order $m$ with the symbol $\sigma_{\Psi_{m}}(x,\xi)=(1+4\pi^{2}|\xi|^{2})^{m/2}$. In addition $\Psi_{m}$ is a \textit{classical} pseudo-differential operator with the asymptotic expansion
$$\si_{\Psi_{m}} (x,\xi)\sim \sum_{j=0}^{\infty} {j\choose\frac{m}{2}} |2\pi \xi|^{m-2j}. $$
Let us denote the Laplacian power of order $n$ by
$$d=(1+\Delta)^{n/2}.$$
 For a classical, compactly supported, pseudo differential operator $\Psi$ of order $-n$ , the \textit{Wodzicki residue} $\Res(\Psi)$ is the integral of the principal symbol $\sigma_{\Psi,-n}$ over the co-sphere bundle $S^{*}M=(\T^{*}M\backslash \{0\})/\Bbb{R}_{+}$
$$\Res(\Psi)=\frac{1}{n}\int_{S^{*}M} \sigma_{\Psi,-n}(\Vs)dvol,$$
where $dvol$ is the volume form of $S^{*}M$. The Wodzicki residue is also well-defined if we replace compactly supported with compactly based. Recall that $\Res([\Psi,\Phi])$ vanishes  for  classical compactly based pseudo differential operators  $\Psi$ of order $k_{1}$  and $\Phi$ of order $k_{2}$ with $k_{1}+k_{2}=-n.$
% As a result of the pseudo-differential calculus, the principle symbol is of order at most $n-1$.
For any $f\in L_{\infty}(G)$, the left multiplication $M_f: L_{2}(G)\to L_{2}(G)$ is a bounded linear operator,
$$M_f(F)(x)=f(x)F(x).$$
Indeed, we obtain the inclusion $L_{\infty}(M)\subset \mathbb{B}(L_2(M))$ by this left regular representation. For any $p\in \Bbb{Z}$, then
$M_{f^{p}}=M_{f}^{p}$. Moreover, there exists a constant $c(n)=\frac{\rm{Vol}(S^{n-1})} {n(2\pi)^{n}}$ such that
\begin{align}\label{cn}
\Res(M_fd^{-1})=c(n) \int_{M} f(x) dvol(x), \forall f\in L_{\infty}(M).
\end{align}

Recall that the \textit{Malzaev ideal} $\mathcal{M}_{1,\infty}$ given by compact operators $T\in \mathbb{B}(L_2(M))$ such that
  \[ \|T\|_{\mathcal{M}_{1,\infty}} \lel \sup_{\nen} \frac{1}{\ln(n+1)}\sum_{j=0}^n u_j(T) \pl <\pl \infty \pl,\]
where $\{u_{j}(T)\}$ is the decreasing sequence of singular values of $T$. The Laplacian power $d$ has continuous extension in $\mathcal{M}_{1,\infty}$, still denoted as $d$.
% In particular $d$ has a continuous extension, still denoted as $d$, in the Malzaev ideal $\mathcal{M}_{1,\infty}$ given by compact operators $T\in \mathbb{B}(L_2(M))$ such that
  %\[ \|T\|_{\mathcal{M}_{1,\infty}} \lel \sup_{\nen} \frac{1}{\ln(n+1)}\sum_{j=0}^n u_j(T) \pl <\pl \infty \pl,\] where $\{u_{j}(T)\}$ is the decreasing sequence of singular values of $T$.
Every dilation invariant extended limit $\om$ on $\ell_{\infty}(\nz)$ defines a weight on $({\mathcal{M}_{1,\infty}})_{+}$
  \[ \Tr_{\om}(T) \lel \om(\{\frac{1}{\ln(n+1)}\sum_{j=0}^n u_j(T)\}_n) \pl.\]
 This weight can be extended, by linearity, to all of $\mathcal{M}_{1,\infty}$ and still remains tracial, and its extension on $\mathcal{M}_{1,\infty}$ is called a \textit{Dixmier trace} on $\mathcal{M}_{1,\infty}$. The Dixmier trace $\Tr_{\omega}$ is non-normal and vanishes on the  Schatten 1-class $S_{1}$. An operator $T\in\mathcal{M}_{1,\infty}$ is said to be \textit{Dixmier measurable} if the value $\Tr_{\omega}(T)$ is independent of the choice of the dilation invariant extended limit $\omega$. Let $T\in(\mathcal{M}_{1,\infty})_{+}$  be positive Dixmier measurable, then
\begin{enumerate}[leftmargin=6.0mm]
\item[(1)] $\lim_{n\to\infty} \frac{1}{\ln(1+n)}\sum_{j=0}^{n}\mu_{j}(T)<\infty$;
\item[(2)] $\Tr_{\omega}(T)=\lim_{t\to\infty} \frac{1}{t}\tau(T^{1+1/t})=\lim_{q\to 1^{+}}(q-1)\tau(T^{q})$ for any dilation invariant extended limit $\omega$.
\end{enumerate}
See \cite{LSZ12} for a proof.

 %Another useful concept is $\zeta$\textit{-functional}, $$\zeta_{\gamma}(T)=\gamma(\frac{1}{t}\tau(T^{1+1/t})), \text{ for } T\in(\mathcal{M}_{1,\infty})_{+},$$ where $\gamma$ is an extended limit. We shall remark that the set of $\zeta$-functionals is strictly smaller than the set of Dixmier traces. (\cite{LSZ}, page 268, Theorem 8.7.1) However for Dixmier measurable positive operators $T\in (\mathcal{M}_{1,\infty})_{+}$, there exist the limits  $$\lim_{n\to\infty} \frac{1}{\ln(1+n)}\sum_{j=0}^{n}\mu_{j}(T)$$  and  $$\lim_{t\to\infty} \frac{1}{t}\tau(T^{1+1/t})=\lim_{q\to 1^{+}}(q-1)\tau(T^{q}).$$ We also have the remarkable equality of the Dixmier trace and the $\zeta$-functional \begin{align}\label{d-zeta \Tr_{\omega}(T)=\lim_{q\to 1^{+}}(q-1)\tau(T^{q}) \end{align} for any $T\in (\mathcal{M}_{1,\infty})_{+}$ and any dilation invariant extended limit $\omega$.

The coincidence between the geometric quantity-Wodzicki residue and the operator algebraic quantity-Dixmier trace was first discovered and proven by Alain Connes,  known as \textit{Connes' trace theorem}. Here we give a simple version of Connes' trace theorem. See \cite{LSZ12} for a complete statement and proof.

\begin{theorem}\label{res} Let $\Psi$ be a classical compactly supported pseudo-differential operator of order $-n$. Then $\Psi$ extends continuously to a Dixmier measurable operator in $\mathcal{M}_{1,\infty}$. Let $\omega$ be any   dilation invariant extended limit, then
$$\Tr_{\omega}(\Psi)=\Res(\Psi).$$
 In particular for any $f\in L_{\infty}(M)$, we have
  \[ \Tr_{\om}(M_fd^{-1}) \lel \Res(M_fd^{-1}) \lel c(n)\int_{M}f(x)dvol(x).\]
\end{theorem}
\noindent Actually it is sufficient to assume that $\Psi$ is compactly based.
%%%%% \textcolor{red}{Note that for compact manifolds the locally defined constants $C_{\al,\beta}$ may be assumed to be globally bounded. This sentence was written by Marius. I dont want to include it since compactness implies boundedness.} \textcolor{red}{Marius used the generator of heat semigroup before. Marius also included the example of $(M,g)$  \cite{Ponge} for more details. }.  Moreover $d=(1+\Delta)^{n/2}$ is well-defined, invertible up to a compact perturbation. Moreover, the singular values behave like $\si_j(d)\simeq j^{-1}$, see [memoirs references therin].

\subsubsection{Applications}  In the following, let $M$ be a compact Riemannian manifold of dimension $n$.
and  $\mathcal{M}$ be a finite von Neumann algebra equipped with a normal faithful trace $\tau_{\mathcal{M}}$. Let $\tau$ denote the normalized trace over $\mathbb{B}(L_2(M))$ and $\tr=\tau\otimes \tau_{\Mm}$. For $a\in L_{p}(M, L_
{p}(\Mm)),$ define $$\tr_{\Mm}(a)=\int_{M} \tau_{\Mm}(a(x))dvol(x).$$
 %The results of this section are based on the interplay between different norms and traces, and here is a list of norms we use in this section.
%\begin{itemize}
%\item[(1)] For $f\in L_{p}(M,\mu)$ (also denoted by $L_{p}(M)$), $\|f\|_{p}^{p}=\int_{M} |f|^{p}d\mu$.
%\item[(2)] For $T\in \mathbb{B}(L_2(M))$, $\|T\|_{p}^{p}=\tau(|T|^{p})$, where $|T|=|T^{*}T|^{1/2}$.
%\item[(3)] For $f\in L_{p}(M\ten\mathcal{M})$, $\|f\|_{p}^{p}=\tr_{\mathcal{M}}(\|f\|^{p})$, where $\tr_{\mathcal{M}}(f)=\int_{M} \tau_{\mathcal{M}}(f)d\mu.$
%\item[(4)] For $T\in L_{p}(\BB(L_{2}(M))\ten \mathcal{M})$,  $\|T\|_{p}^{p}=\tr(|T|^{p})$,  where  $\tr=\tau\otimes \tau_{\mathcal{M}}$.
%\end{itemize}
We obtain the following result as a corollary of Theorem \ref{res}.

\begin{cor}\label{rpscal} Let $1\le p<\infty$ and $f\in C^{\infty}(M)$.  Then
 \[  c(n) \|f\|_p^p \lel
 \Tr_{\om}(|d^{-\frac{1}{2p}}M_fd^{-\frac{1}{2p}}|^p)
  \lel \lim_{q\to 1^{+}} (q-1) \| d^{-\frac{1}{2p}}M_fd^{-\frac{1}{2p}}\|_{pq}^{pq} \pl .
  \]
\end{cor}
\begin{proof}
 We first claim that for any even integer $p\gl 1$
 \begin{equation} \label{integer}
 c(n) \|f\|_p^p \lel
 \Tr_{\om}(|d^{-\frac{1}{2p}}M_fd^{-\frac{1}{2p}}|^p)
  \lel \lim_{q\to 1^{+}} (q-1) \| d^{-\frac{1}{2p}}M_fd^{-\frac{1}{2p}}\|_{pq}^{pq} \pl .
 \end{equation}
Noting that
$M_{|f|}^{p}d^{-1}$ and $|d^{-\frac{1}{2p}}M_f d^{-\frac{1}{2p}} |^{p}$ have the same principal symbol of order $-n$ \cite{RT10}, we infer that $$\Res(M_{|f|}^pd^{-1})=\Res(|d^{-\frac{1}{2p}}M_fd^{-\frac{1}{2p}} |^{p}).$$
Together with \eqref{cn}, we obtain that  $$ c_{n}\|f\|_{p}^{p}\lel \Res(M_{|f|}^{p}d^{-1}) \lel \Res(|d^{-\frac{1}{2p}}M_fd^{-\frac{1}{2p}} |^{p}).$$
Applying Theorem \ref{res} and the assertion follows for any even integer $p$.
%% $$ c_{n}\|f\|_{p}^{p}\lel \Tr_{\om}(|d^{-\frac{1}{2p}}M_fd^{-\frac{1}{2p}}|^p)\lel \lim_{q\rightarrow 1^{+}} (q-1) \||d^{-\frac{1}{2p}}M_fd^{-\frac{1}{2p}}|^p \|_{q}^{q}.$$ Then $$\lim_{q\rightarrow 1^{+}} (q-1) \|\p |d^{-\frac{1}{2p}}M_fd^{-\frac{1}{2p}}|^p \|_{q}^{q}\lel \lim_{q\to 1^{+}}(q-1)\|d^{-\frac{1}{2p}}M_fd^{-\frac{1}{2p}}\|_{pq}^{pq}$$ yields the assertion for any even integer $p$.

Then we prove the upper estimate for all $p$ by interpolation. We may assume that $c(n)\|f\|_{p}^{p}<1$.  Let
$\frac{1}{p}=\frac{1-\theta}{p_0}+\frac{\theta}{1}$ and $p_0$ be an even integer. Let $\al(z)=\frac{1-z}{p_0}+\frac{z}{1}$ and consider the analytic function
 \[ F(z) \lel u {M_{|f|}}^{p\alpha(z)}d^{-\al(z)},\]
where $M_f=uM_{|f|}$ is given by the polar decomposition. Thus $F(\theta)=uM_{|f|}d^{-1/p}=M_fd^{-1/p}$.
Note that for different values of $z=it$ the functions $F(it)$ and $F(0)$ only differ by left and right multiplications by unitaries, and $$|F(it)|=|F(0)|=M_{|f|^{p/p_{0}}}d^{-1/p_{0}}.$$
Applying the limit $q\to 1^{+}$ applies uniformly to $t$ and together with (\ref{integer}), we  infer that
\begin{align*}
&\lim_{q\to 1^{+}}(q-1)\sup_{t}\|F(it)\|_{qp_{0}}^{qp_{0}}\\
=&\lim_{q\to 1^{1+}} (q-1)\| d^{-1/(2p_{0})}{M_{|f|}}^{p/p_{0}} d^{-1/(2p_{0})} \|_{qp_{0}}^{qp_{0}} =c(n)\|f\|_{p}^{p}\leq 1.
\end{align*}
Similarly we obtain that  $\lim_{q\to 1^{+}} (q-1) \sup_{t} \|F(1+it)\|^{q}_{q} \kl 1 \pl .$
By the three line lemma (see \cite{BL}), we deduce that
 \[ \lim_{q\to 1^{+}} (q-1) \|M_fd^{-1/p}\|_{q}^{q} \lel \lim_{q\to 1^{+}} (q-1)\|F(\theta)\|_{qp}^{qp} \kl 1 \pl .\]
Homogeneity implies
 \[ \lim_{q\to 1^{+}} (q-1) \|M_fd^{-1/p}\|_{pq}^{pq} \kl c(n) \|f\|_p^p \pl
  \text{ and } \lim_{q\to 1^{+}} (q-1) \|d^{-1/p}M_f^{*}\|_{pq}^{pq} \kl c(n) \|f\|_p^p \pl . \]
For the lower estimate, we assume that $\|f\|_p^p=1$. For any $\epsilon>0$, there exists $g\in C^{\infty}(M)$ such that  $$\int |g^*f| d{\rm vol}(x)\gl (1-\eps)$$ and $\|g\|_{p'}\le 1$ with $\frac{1}{p'}+\frac{1}{p}=1$.
% From above,  $$\lim_{q\to 1^{+}} (q-1)||M_fd^{-1/p}||_{pq}^{pq}\leq \infty.$$
Note that  $d^{-1/p'}M_{g^*f}d^{-1/p}$ is a pseudo-differential operator of order $-n$, then
$$c(n) \int |g^*f| d{\rm vol}(x) \lel \lim_{q\to 1^{+}} (q-1) \|d^{-1/p'}M_{g^*f}d^{-1/p}\|_q^q .$$
By H\"{o}lder's inequality, then
 \begin{align*} c(n) \int |g^*f| d{\rm vol}(x)&\leq(\limsup_{q\to 1^{+}} (q-1)^{1/p'} \|d^{-1/p'}M_{g^*}\|_{p'q}^q) \limsup_{q\to 1^{+}} (q-1)^{1/p}  \|M_{f}d^{-1/p}\|_{pq}^{p} \\
 &\leq c(n)^{1/p'} \limsup_{q\to 1^{+}} (q-1)^{1/p}  \|M_{f}d^{-1/p}\|_{pq}^{q}.
 \end{align*}
 Together with $\int |g^*f| d{\rm vol}(x)\gl (1-\eps)$, we have
 $$ (1-\eps)\leq c(n)^{1/p'} \limsup_{q\to 1^{+}} (q-1)^{1/p}  \|M_{f}d^{-1/p}\|_{pq}^{q}. $$
%We may consider the pseudo differential operator   $d^{-1/p'}M_{g^*f}d^{-1/p}$.  We deduce that \begin{align*} c(n)(1-\eps) &\le c(n) \int |g^*f| d{\rm vol}(x) \lel \lim_{q\to 1^{+}} (q-1) \|d^{-1/p'}M_{g^*f}d^{-1/p}\|_q^q \\  &\kl \limsup_{q\to 1^{+}} (q-1) (\|d^{-1/p'}M_{g^{*}}\|_{p'q}\|M_{f}d^{-1/p}\|_{pq})^q \\ &\kl (\limsup_{q\to 1^{+}} (q-1)^{1/p'} \|d^{-1/p'}M_{g^*}\|_{p'q}^q) \limsup_{q\to 1^{+}} (q-1)^{1/p}  \|M_{f}d^{-1/p}\|_{pq}^{p} \\ &\kl c(n)^{1/p'} \limsup_{q\to 1^{+}} (q-1)^{1/p}  \|M_{f}d^{-1/p}\|_{pq}^{q} \pl . \end{align*}
Taking the $p$-th power we deduce that
 \begin{align}(1-\eps)^p c(n) \kl \limsup_{q\to 1^{+}} (q-1)  \|M_{f}d^{-1/p}\|_{pq}^{pq} =\lim_{q\to 1^{+}} \|M_{f}d^{-1/p}\|_{pq}^{pq}  \label{pscal}.
\end{align}
Thus the upper and lower estimates yield the equality. \qd
\begin{remark}
We could also use the pseudo-differential calculus developed by Connes (\cite{connes2}) to obtain the result without interpolation. 
\end{remark}
\noindent Let us recall the definition  of vector-valued mixed $(p,q)$-spaces from  see \cite{pvp},\cite{mem}:
 \[ L_{\infty}(\N,L_q(\Mm)) \lel [\N\bar{\ten}\Mm,L_{\infty}(\mathcal{N},L_1(\Mm))]_{1/q} \]
obtained by complex interpolation. We refer to \cite{mem} for  the fact that for elements in $\N\ten \Mm$ the two equivalent expressions for the norm
  \begin{align*}
\|f\|_{L_{\infty}(\N,L_1(\Mm))} &=   \sup_{\|a\|_2,\|b\|_2} \|(a\ten 1)f(b\ten 1)\|_{L_1(\N\ten \Mm)}\\
     &=
   \inf_{f=f_1f_2} \|id\ten \tau(f_1f_1^*)\|_\N^{1/2}
   \|id\ten \tau(f_2^*f_2)\|_\N^{1/2}
   \end{align*}
coincide. Thus by interpolation, we deduce an isometric inclusion
 \[ L_{\infty}(\N_1,L_q(\Mm))\subset L_{\infty}(\N_2,L_q(\Mm)) \]
for every inclusion of von Neumann algebras $\N_1\subset \N_2$. This is in particular true for the inclusion $L_{\infty}(M)\subset \mathbb{B}(L_2(M))$ given by the left regular representation,
\begin{align}\label{inclusion-p}
L_{\infty}(M,L_{p}(\mathcal{M}))\subset L_{\infty}(\BB(L_{2}(M)), L_{p}(\mathcal{M})).
\end{align}
Also recall Pisier's interpolation theorem for vector-valued $L_{p}$ spaces (see \cite{pvp,mem}) that
\begin{align} \label{pisier} L_{p}(L_{\infty}(M)\bar{\ten} \Mm)  \lel L_{p}(M,L_{p}(\Mm)).
\end{align}

\begin{cor}\label{vng} Let $f\in L_{p}(M, L_{p}(\mathcal{M}))$, then
\begin{align*}\lim_{q\to 1^{+}} (q-1) \|(d^{-\frac{1}{2p}}\ten 1)M_{f}(d^{-\frac{1}{2p}}\ten 1)\|_{pq}^{pq} \lel c(n) \|f\|_{L_p(L_{\infty}(M)\bar{\otimes} \mathcal{M})}^{p}.
\end{align*}
\end{cor}
\begin{proof}
%%  Here we use $||f||_{p}^{p}=\int_{M} \tau(|f(x)|^{p}) d{\rm vol}(x)$. \textcolor{red}{shall we use the normalized volume measure?}%% %%% Then
%%% \[ \|f\|_p \lel \lim_{q \to p} \|f\|_q \pl \]
%%%%%% with respect to the  \textcolor{red}{$tr_\Mm=\int \tau_M(a(x)) d{\rm vol}(x)$ which satisfies $tr_\Mm(1)=1$. direct integral. This notation is not consistent with other notations.}
It is sufficient to consider the vector valued function $f\in L_{\infty}(M, \mathcal{M})$. Indeed,  we may extend this result to all of $L_p(L_{\infty}(M)\bar{\ten} \Mm)$ using the Banach space $\prod_\U L_{p}(\mathbb{B}(L_2(M)) \bar{\ten}\Mm)$.

Now let $f\in L_{\infty}(M,\Mm)$, then $f\in L_{\infty}(L_{\infty}(M)\bar{\otimes} \mathcal{M})\subset{L_{p}(L_{\infty}(M)\bar{\otimes} \mathcal{M})}$.  For any $\eps>0$, there exists $p_{0}>p$ such
\begin{align} \label{con-p}
\|f\|_{L_p(L_{\infty}(M)\bar{\otimes} \mathcal{M})}\kl \|f\|_{L_{p_{0}}(L_{\infty}(M)\bar{\otimes} \mathcal{M})}\kl (1+\eps)\|f\|_{L_p(L_{\infty}(M)\bar{\otimes} \mathcal{M})}.
\end{align}
%Thanks to Pisier's interpolation theorem for vector-valued $L_p$ spaces (see \cite{pvp,mem}) we have
% $$  L_{p_{0}}(L_{\infty}(M)\ten \Mm)  \lel L_{p_{0}}(M,L_{p_{0}}(\Mm)).$$
%%%%%% $$  L_{p_{0}}(L_{\infty}(M)\ten \Mm)  \lel L_{p_{0}}(M,L_{p}(\Mm)) \subset L_{p_{0}}(\BB(L_{2}(M)),L_{p_{0}}(\mathcal{M}))= L_{p_{0}}(\BB(L_{2}(M))\ten\mathcal{M}).$$
%%%% \noindent (\cite{Pisier}(Non-commutative vector valued Lp-spaces and completely p-summing map, page 2, on the bottom.))
Thus there exist $f_1,f_2\in L_{2p_{0}}(M)$ and $F\in L_{\infty}(M,L_{p_{0}}(\Mm))$ such that
 $f= (f_1\ten 1)F(f_2\ten 1)$,
$\max\{\|f_1\|_{2p_{0}}\|f_2\|_{2p_{0}}\}\le \|f\|_{L_{p_{0}}(L_{\infty}(M)\bar{\otimes} \mathcal{M})}^{1/2}$, and $\|F\|_{p_{0}}\le 1$, where $$\|F\|_{p_{0}}=\int_{M} \tau_{\Mm}(  |F(x)|^{p_{0}} )^{1/{p_{0}}}dvol(x) .$$  Indeed,  the functions $$f_1(x)=f_2(x)= \tau_{\mathcal{M}}(|f(x)|^{p_{0}})^{\frac{1}{2p_{0}}}$$ will do  the job.
The inclusion result \eqref{inclusion-p} implies that
$M_{F}\in L_{\infty}(\BB(L_{2}(M)), L_{p_{0}}(\mathcal{M}))$.
 Since $p<p_{0}$, we apply Corollary \ref{rpscal} to $f_{1}f_{1}^{*}$ and $f_{2}f_{2}^{*}$ and continuity of $L_{p}$ spaces, then
 \begin{align*}
 \lim_{q\to 1^{+}}  \left((q-1) \|d^{-\frac{1}{2p}}M_{f_1f_1^*}d^{-\frac{1}{2p}}\|_{qp}^{qp}\right)^{1/2}  \lim_{q\to 1^{+}} \left((q-1) \|d^{-\frac{1}{2p}}M_{f_2^*f_2}d^{-\frac{1}{2p}}\|_{qp}^{qp}\right)^{1/2} \\
=  c(n) \|f_{1}f_{1}^{*}\|_{p}^{p/2} \|f_{2}f_{2}^{*}\|_{p}^{p/2} \leq c(n) \left( \|f_{1}f_{1}^{*}\|_{p_{0}}^{p} \|f_{2}f_{2}^{*}\|_{p_{0}}^{p} \right)^{1/2}
\leq  (1+\eps)^{p}c(n)||f||_{p}^{p}.
 \end{align*}
% Thus $$(d^{-\frac{1}{2p}}\ten 1)M_{f}(d^{-\frac{1}{2p}}\ten 1)\lel  (d^{-\frac{1}{2p}}M_{f_1}\otimes 1)M_{F}(M_{f_2}d^{-\frac{1}{2p}}\ten 1).$$ For any $p$, the inclusion result \eqref{inclusion-p} gives that $M_{F}\in L_{\infty}(\BB(L_{2}(M)), L_{p_{0}}(\mathcal{M}))$.   Note that $d^{-\frac{1}{2p}}M_{f_1}\ten 1$, $M_{f_2}d^{-\frac{1}{2p}}\ten 1\in L_{2pq}(\BB(L_{2}(M))\otimes \mathcal{M})$ for small $q$ since $p<p_{0}$.

By \cite{pisier93} we find that
\begin{align*}
&\lim_{q\to 1^{+}} (q-1) \|(d^{-\frac{1}{2p}}\ten 1)M_{f}(d^{-\frac{1}{2p}}\ten 1)\|_{pq}^{pq}\\
 \leq &\limsup_{q\to 1^{+}} \: (q-1) \|d^{-\frac{1}{2p}}M_{f_1}\|_{2qp}^{pq}
\|M_{F}\|_{L_{\infty}(\BB(L_{2}(M)), L_{pq}(\mathcal{M}))}^{pq} \|M_{f_2}d^{-\frac{1}{2p}}\|_{2pq}^{pq}\\
 \leq &\limsup_{q\to 1^{+}}  \left((q-1) \|d^{-\frac{1}{2p}}M_{f_1f_1^*}d^{-\frac{1}{2p}}\|_{qp}^{qp}\right)^{1/2}
\limsup_{q\to 1^{+}} \left((q-1) \|d^{-\frac{1}{2p}}M_{f_2^*f_2}d^{-\frac{1}{2p}}\|_{qp}^{qp}\right)^{1/2}.
\end{align*}
%%% check lp inclusion of vna
% Apply Corollary \ref{rpscal} and continuity of $L_{p}$ space \eqref{con-p}to $f_{1}f_{1}^{*}$ and $f_{2}f_{2}^{*}$, then  \begin{align*}&\limsup_{q\to 1^{+}}  \left((q-1) \|d^{-\frac{1}{2p}}M_{f_1f_1^*}d^{-\frac{1}{2p}}\|_{qp}^{qp}\right)^{1/2}  \limsup_{q\to 1^{+}} \left((q-1) \|d^{-\frac{1}{2p}}M_{f_2^*f_2}d^{-\frac{1}{2p}}\|_{qp}^{qp}\right)^{1/2} \\=  &c(n) \|f_{1}f_{1}^{*}\|_{p}^{p/2} \|f_{2}f_{2}^{*}\|_{p}^{p/2} \leq c(n) \left( \|f_{1}f_{1}^{*}\|_{p_{0}}^{p} \|f_{2}f_{2}^{*}\|_{p_{0}}^{p} \right)^{1/2} \leq  (1+\eps)^{p}c(n)||f||_{p}^{p}. \end{align*}
Thus $$\lim_{q\to 1^{+}} (q-1) \|(d^{-\frac{1}{2p}}\ten 1)M_{f}(d^{-\frac{1}{2p}}\ten 1)\|_{pq}^{pq}\leq  (1+\eps)^{p}c(n)||f||_{p}^{p}.$$ Sending $\eps$ to $0$ yields the upper bound.  The same interpolation argument as in \ref{pscal}  also shows the lower bound by duality.
\qd

%%% \begin{remark*} We need to work on $p_{0}$ because we need some flexibility to work with $pq$-norm for $q\geq 1$. Note that $pq$-norm is also changing.
%% \end{remark*}
\begin{lemma}\label{lim}
Let  $a$,$b\in L_{p}(M, L_{p}(\mathcal{M}))$ for $p\in(1,2)$. Let $a$ be positive and $b$ be self-adjoint.
Define $A=(d^{-\frac{1}{2p}}\otimes 1)M_a (d^{-\frac{1}{2p}}\otimes 1)$ and  $B=(d^{-\frac{1}{2p}}\otimes 1)M_b(d^{-\frac{1}{2p}}\otimes 1)$. If there exists $C>0$ such that $-C a\leq b\leq C a$, then
\begin{align*}
\lim_{q\to 1^{+}} (q-1)\tr \left( BA^{pq-1} \right) \leq c(n)\tr_{\mathcal{M}}(ba^{p-1}).
\end{align*}
\end{lemma}

\begin{proof}
 Let $t\geq 0$ and $tC\leq 1$, then $tb+a\leq 1$. Applying Corollary \ref{vng} to $a+tb$ and $a$ implies
 \begin{align} \label{a+tb}
 \lim_{q\to 1^{+}} (q-1) \|A+tB\|_{pq}^{pq} = c(n) \|a+tb\|_{L_{p}(L_{\infty}(M)\bar{\otimes} \Mm)}^{p},
 \end{align}
 and
 \begin{align}
  \lim_{q\to 1^{+}} (q-1) \|A\|_{pq}^{pq} = c(n) \|a\|_{L_{p}(L_{\infty}(M)\bar{\otimes} \Mm)}^{p}. \label{a}
  \end{align}
Noting that $\|\pl\|_{pq}^{pq}$ is convex for $q\geq 1$ small enough, we obtain
\begin{align*}
 pq \tr\left(BA^{pq-1}\right) \kl  \frac{\|A+tB\|_{pq}^{pq}-\|A\|_{pq}^{pq}}{t} \pl.
\end{align*}
Together with \eqref{a+tb} and \eqref{a}, we have
$$\lim_{q\to 1^{+}} (q-1) \tr(BA^{pq-1})\leq  \frac{c(n)\left(\|a+tb\|_{L_{p}(L_{\infty}(M)\bar{\otimes} \Mm)}^{p}-\|a\|_{L_{p}(L_{\infty}(M)\bar{\otimes} \Mm)}^{p} \right)}{tp}.$$
Using the differentiation formula for the $p$-norm, we observe that
 \begin{align*}
 &\|a+tb\|_{L_{p}(L_{\infty}(M)\bar{\otimes} \Mm)}^{p}-\|a\|_{L_{p}(L_{\infty}(M)\bar{\otimes} \Mm)}^{p}\\
  =&  t\int_0^1 p\tr_{\mathcal{M}}(b(a+stb)^{p-1}) ds  \\
  =& tp \tr_\Mm(ba^{p-1}) + tp\int_0^1 \tr_\Mm(b((a+stb)^{p-1}-a^{p-1})) ds \pl .
  \end{align*}
 Thus it suffices to show that
\begin{align}
\int_0^1 \tr_\Mm(b((a+stb)^{p-1}-a^{p-1})) ds=\mathcal{O}(t) \label{error-p}
\end{align}
 as $t\to 0$, then sending $t\to 0$ implies the assertion. Indeed, we decompose $b=b^{+}-b^{-}$ for positive $b^{+}$ and $b^{-}$.
 % By the operator concavity of  $x\mapsto x^{p-1}$, we have  $$|\tr_\Mm\left(b^+((a+stb)^{p-1}-a^{p-1})\right)|\leq (p-1)st  \tr_\Mm(b^+a^{p-1}).$$
 Using the the monotonicity of $x\mapsto x^{p-1}$ and $b\leq C a$ we deduce
 \begin{align*}
  \tr_\Mm(b^+a^{p-1})\le   \tr_\Mm(b^+(a+stb)^{p-1})\kl (1+st C )^{p-1} \tr_\Mm(b^+a^{p-1}).
  \end{align*}
The same argument applies for $b^-$ and hence
\begin{align*}
|\tr_\Mm\left(b((a+stb)^{p-1}-a^{p-1})\right)|\kl &\left( (1+stC)^{p-1}-1\right) \tr_{\mathcal{M}}\left( |b|a^{p-1}\right)\\ \kl&  (p-1)st C  \tr_\Mm(|b|a^{p-1}).
\end{align*}
% Since $x^{p-1}$ is also operator concave, then $|\tr_\Mm\left(b^+((a+stb)^{p-1}-a^{p-1})\right)|\leq (p-1)st C  \tr_\Mm(b^+a^{p-1}).$ The same argument applies for $b^-$ and hence \[|\tr_\Mm\left(b((a+stb)^{p-1}-a^{p-1})\right)|\kl (p-1)st C \tr_{\mathcal{M}}\left( |b|a^{p-1}\right).\]
Integrating the inequality above yields \eqref{error-p}.
\qd
\noindent Now we are ready to prove the main theorem.
\begin{proof} Let $1<p<2$ and $p<q<2$. Let $a:M\to \Mm$ be a smooth positive function and $a_{\epsilon}=a+\eps 1$ for $\epsilon >0$. Then $b=\L(a_{\epsilon})=\L(a)$  since $\L$ is self-adjoint.  Let $C=\eps^{-1}\|\L(a)\|_{L_{\infty}(L_{\infty}(M)\bar{\otimes} \Mm)}$, then $-Ca_{\epsilon}\leq b\leq Ca_{\epsilon}$. Let $A_{\epsilon}=(d^{-\frac{1}{2p}}\otimes 1)M_{a_{\epsilon}} (d^{-\frac{1}{2p}}\otimes 1)$ and  $B=(d^{-\frac{1}{2p}}\otimes 1)M_b(d^{-\frac{1}{2p}}\otimes 1)$. It follows from Lemma \ref{lim} that
 \begin{align} \label{Rm-1}
\lim_{q\to 1^{+}} (q-1) \tr (BA_{\epsilon}^{pq-1})
 \kl  c(n)\tr_\Mm( ba_{\epsilon}^{p-1}).
 \end{align}
 Noting $(\L\otimes id_{\mathcal{M}})(A_{\epsilon})=B$ since $\L(d^{-\al}xd^{-\beta})=d^{-\al}\L(x)d^{-\beta}$.
%By $\CLSI^{+}(\BB(L_{2}(M)), \L)=\inf_{p\in (1,2)} \CpSI(\BB(L_2(M), \L)$,
we have
 \begin{align*}
 &\|A_{\epsilon}\|_{pq}^{pq}-\|(d^{-\frac{1}{2p}}\otimes 1)M_{E(a_{\epsilon})}(d^{-\frac{1}{2p}}\otimes 1)\|_{pq}^{pq}\\ \leq &\frac{pq}{\CLSI^{+}(\BB(L_{2}(M)), \L)  } \tr_{\Mm}\left(BA_{\epsilon}^{pq-1}\right).
\end{align*}
Together with \eqref{Rm-1}, we obtain
\begin{align*}&\lim_{q\to 1^{+}}(q-1)\left( \|A_{\epsilon}\|_{pq}^{pq}-\|(d^{-\frac{1}{2p}}\otimes 1)M_{E(a_{\epsilon})}(d^{-\frac{1}{2p}}\otimes 1)\|_{pq}^{pq} \right) \\ \leq &\frac{pc(n)}{\CLSI^{+}(\BB(L_{2}(M)), \L) }\tr_\Mm( ba_{\epsilon}^{p-1}).\end{align*}
%Note $(\L\otimes id_{\mathcal{M}})(A_{\epsilon})=B$ since $\L(d^{-\al}xd^{-\beta})=d^{-\al}\L(x)d^{-\beta}$.
% Thus it follows from Lemma \ref{lim} that \begin{align*} \lim_{q\to 1^{+}}  \frac{(q-1)pq}{\CLSI^{+}(\BB(L_{2}(M)), \L)}    \pl  \tr (BA_{\epsilon}^{pq-1})  \kl  \frac{p\: c(n)}{\CLSI^{+}(\BB(L_{2}(M)), \L)}\tr_\Mm(\L(a_{\epsilon})a_{\epsilon}^{p-1})  \pl . \end{align*}
Applying Corollary \ref{vng}  to $(d^{-\frac{1}{2p}}\otimes 1)M_{a_{\epsilon}}(d^{-\frac{1}{2p}}\otimes 1)$  and $(d^{-\frac{1}{2p}}\otimes 1)M_{E(a_{\epsilon})}(d^{-\frac{1}{2p}}\otimes 1)$ implies
 \begin{align}\label{p-eps}
&\|a_{\epsilon}\|_{L_p(L_{\infty}(M)\bar{\otimes} \Mm)}^p-\|E(a_{\epsilon})\|_{L_p(L_{\infty}(M)\bar{\otimes} \Mm)}^p\\\leq &\frac{p}{   \CLSI^{+}(\BB(L_{2}(M)), \L)  }   tr_\Mm(\L(a)a_{\epsilon}^{p-1}) \pl .
   \end{align}
The left hand side is continuous in $\eps$. By functional calculus and the dominated convergence theorem for the sequence of functions $g_k(x)=(x+\frac{1}{k})^{p-1}$,  we deduce that
 \[ \lim_{k\to \infty} \tr_{\mathcal{M}}(b(a+\frac{1}{k})^{p-1})
 \lel \lim_k \int_{\rz} g_k(x) d\mu_b(x)
 \lel \tr_{\mathcal{M}}(ba^{p-1}) \pl ,\]
where we use the spectral measure $\int f(x)d\mu_b(x)=(b_1, f(x)b_2)$ given by a decomposition $b=b_1b_2^*$ with $b_{1},b_{2}\in L_2(\Mm)$. ($L_{2}(\Mm)$. By sending $\eps$ to $0$, then
 \[ \lim_{\eps\to 0} \tr_\Mm(\L(a)(a+\eps)^{p-1}) \lel \tr_\Mm(\L(a)a^{p-1})\pl .\]
Thus \eqref{p-eps} does indeed imply
 \[ \CLSI^{+}(\BB(L_2(M),\L) \kl \CpSI(L_{\infty}(M),\L) \pl .\]
Taking the infimum over $p>1$, then yields  the assertion.\qd
\subsection{Collective Lindbladian} Let $\mathcal{X}=\{X_{1},\dots, X_{d}\}$ be a set of self-adjoint matrices in $\Mz_{n}$, hence hence $i\mathcal{X}\subset \mathfrak{su}(n)$ is a subset of the Lie algebra of $SU(n)$.  Let us first consider a general Lindbladian
 \[ L_{\mathcal{X}} \lel \sum_{k=1}^d [X_k,[X_k,\pl]]\pl .\]
 This also allows us to define the right invariant differential operators $X_k(f)(g)=\frac{d}{dt}f(e^{itX_k}g)$ with sub-Laplacian
 \[ \Delta_{\mathcal{X}} \lel -\sum_k X_k^2 \pl .\]
In some cases this differential operator is only ergodic for a Lie-subgroup $G\subset SU(n)$, and then we add the relevant group $\Delta_{\mathcal{X},G}$ in the notation. Let us fix the notation $\Hs=\ell_2^n$. For a arbitrary  representation $u:G\to U(\Hs)$, we may then define new Lie-derivatives  $X_k^{\Hs}=\frac{d}{idt}u(e^{itX_k})|_{t=0}$ and obtain the  individual Lindbladian
 \[ L_{X_k}^{\Hs}(\rho) \lel [X_k,[X_k,\rho]] \]
and the transferred Lindbladian
  \[ L_{\mathcal{X}}^{\Hs} \lel \sum_k L_{X_k}^{\Hs} \pl .\]
We deduce from the Peter-Weyl theorem \cite{Dieck} that all finite dimensional irreducible representations $u $ are contained in $\Ks={\Hs}^{\ten_{m_1}}\ten \bar{\Hs}^{\ten_{m_2}}$ for some $m_{1}$ and $m_{2}$. Let us determine the corresponding Lindbladian by differentiation.
Then we consider $u_m(g)=u(g)^{\ten_m}$ and deduce that
 \[  X_k^{(m)} \lel \frac{d}{idt} u(e^{itX_k})^{\ten_m}|_{t=0}
 \lel \sum_{j=1}^m \pi_j(X_k) \pl ,\]
 where $\pi_{j}(a)=1^{\ten(j-1)}\ten a\ten 1^{\ten (n-j)}$.
%is given by the $^*$ homomorphism $\pi_j$ which sends $a$ to $1\ten \cdots\underbrace{a}_{\mbox{\scriptsize{$j$-the position}}}\ten 1\ten 1\cdots$ to the $j$-th register. 
The adjoint representation is given by $\bar{u}(g)=u(g^{-1})^{\trans}\lel \overline{u(g)}$. This means  that
 \begin{align*}
  X_k^{\bar{\Hs}^{\ten_m}}&=  \frac{d}{idt}\overline{u(e^{-itX_k})}^{\ten_m}|_{t=0}
  \lel \frac{1}{i} \sum_{j=1}^m \overline{\pi_j(-iX_k)} \\
 &= \sum_{j=1}^m \pi_j(\overline{X_k})
 \lel \sum_{j=1}^m \pi_j((X_k^{\trans})^*)
 \lel \sum_{j=1}^m \pi_j(X_k^{\trans}) \pl .
 \end{align*}
Let us therefore define
 \[ \bar{L}_{\mathcal{X}}^{\Hs}\lel \sum_k [X_k^{\trans},[X_k^{\trans},\cdot ]] \pl .\]
%For the representations deemed essential by the Peter-Weil theorem we find collective Lie-derivatives\[ X_k^{(m_1,m_2)} \lel X_k^{(m_1)}\oplus (X_k^{\trans})^{(m_2)} \lel \sum_{j=1}^{m_1}\pi_j(X_k)+\sum_{j=m_1+1}^{m_1+m_2}\pi_j(X_k^{\trans}) \pl .\]This implies   \begin{align*}L_{\mathcal{X}}^{\Hs^{\ten_{m_1}}\ten \bar{\Hs}^{\ten_{m_2}}}  \lel \sum_{k=1}^d \textcolor{red}{L_{X_k}^{(m_1,m_2)} } (\textcolor{blue}{L_{X_{k}^{(m_{1},m_{2})}}^{\Hs^{\ten_{m_1}}\ten \bar{\Hs}^{\ten_{m_2}}}})\pl .  \end{align*}(\textcolor{blue}{ or $L_{\mathcal{X}}^{\Ks}=\sum_{k=1}^{d} L_{X_{k}^{(m_{1},m_{2})}}^{\Ks}.$ })
%\begin{rem}{\rm  $\CLSI(\textcolor{red}{\bar{L}_{\mathcal{X}}}(\textcolor{blue}{L_{\mathcal{X}}^{\bar{\Hs}}}))=\CLSI(\textcolor{red}{L_{\mathcal{X}}}(\textcolor{blue}{L_{\mathcal{X}}^{\Hs}}))$. The same holds for $\CLSI^+$.}\end{rem}

Let us now introduce the diagonal representation $\hat{u}(g)=\kla \begin{array}{cc} u(g)& 0 \\ 0&\bar{u}(g)\end{array}\mer$ on $\Hs\oplus \bar{\Hs}$ and $\hat{X}_k={\rm diag}(X_k,X_k^{\trans})$. 
The combined collective Lindbladians are given by
$$\hat{L}^{m}_{X_{k}}=\sum_{j=1}^{m} [\pi_{j}(\hat{X}_{k}),[\pi_{j}(\hat{X}_{k}),\cdot]]$$
The corresponding generator for the system is denoted by $\hat{L}_{\mathcal{X}}^m=\sum_{k=1}^{d}\hat{L}^{m}_{X_{k}}$.
\begin{rem}{\rm  $\CLSI(L_{\mathcal{X}}^{\bar{\Hs}})=\CLSI(L_{\mathcal{X}}^{\Hs})$. The same holds for $\CLSI^+$.}\end{rem}
%\begin{lemma} $\CLSI(\bar{L}_{\mathcal{X}})=\CLSI(L_{\mathcal{X}})$. The same holds for $\CLSI^+$.
%\end{lemma}
%\begin{proof} Let $\mathcal{M}$ be a finite von Neumann algebra. Then $\mathcal{M}$ equipped with multiplication $a.b=ba$ is also finite, and hence
% \[ \CLSI(L_{\mathcal{X}}) D(\rho||E_{fix}\ten id_{\mathcal{M}^{op}}(\rho)) \kl \sum_k tr(L_{X_k}(\rho)\ln\rho) \pl .\]
%Note that fixpoint algebra of $L_{\mathcal{X}}$ is closed under taking the transposed and hence coincides with the fixpoint algebra of $\bar{L}_{\mathcal{X}}$. We use the notation $\rho^{op}$ for the same density in $(\Mz_n\ten \mathcal{M}^{op})^{op}=\Mz_n^{op}\ten \mathcal{M}$. Then we deduce that
% \begin{align*}
%   \CLSI(L_{\mathcal{X}}) D(\rho^{op}||E_{fix}\ten id_{\mathcal{M}^{op}}(\rho^{op}))
% &\kl    \sum_k tr(L_{X_k}(\rho)\ln\rho)\\
% &= \sum_k tr(\ln \rho. (X_k^2.\rho+\rho.X_k^2-2X_k.\rho.X_k)) \\
% &\lel \sum_k tr((X_k^{op})^2\rho^{op}+\rho^{op}(X_k^{op})^2-2X_k^{op}\rho^{op}X_k^{op}) \\
% &\lel I_{\bar{L}}(\rho) \pl .
% \end{align*}
%Here we use that $\rho\to \rho^{op}$ is compatible with functional calculus of a single variable. \qd

\begin{lemma}\label{ct} 
$\CLSI^+(\Delta_{\mathcal{X}})\lel \inf_{m}
  \CLSI^+(\hat{L}_{\mathcal{X}}^{m}) \pl .$
\end{lemma}
\begin{proof} In view of Theorem \ref{Rm} it suffices to control $L_\mathcal{X}^{L_2(G)}$.
By Peter-Weyl theorem
   \[ L_2(G) \lel \oplus_{\pi} \left(\Hs_{\pi}\ten \bar{\Hs}_{\pi}\right) \]
is given by the sum of equivalence classes of irreducible representations. Here we use left regular representation $\lambda: G\to U(L_{2}(G))$ given by $\lambda_{g}(f)(h)=f(g^{-1}h)$, which corresponds to $(u_{\pi}(g)\otimes 1)_{\pi}$. Since $\la(g)$ commutes with the spectral projections onto $\oplus_{\pi\in F}\left( \Hs_{\pi}\ten \bar{\Hs}_{\pi}\right)$ for any finite set $F$,  it suffices to consider
$$\Hs_{F}=\oplus_{\pi\in F}\Hs_{\pi}$$
via the completeness. Again thanks to the Peter-Weyl theorem we can find $m_{1}(\pi)$ and $m_{2}(\pi)$ such that 
$$\Hs_{\pi}\subset \Hs^{\otimes m_{1}(\pi)} \otimes \bar{\Hs}^{\ten m_{2}(\pi)} .$$
For any $F$, there exists a large integer $m(F)$ such that 
$$\Hs_{F}\subset (\Hs\oplus \bar{\Hs})^{\ten m(F)}.$$
Indeed $m(F)=\max_{\pi\in F}\{m_{1}(\pi)\}+\max_{\pi\in F}\{m_{2}(\pi)\}.$ Using the distributive law, we observe that 
$$(\Hs\oplus \bar{\Hs})^{m(F)}\lel \oplus_{A\subset \{1,\dots, m(F)\} }\left( \Hs^{A}\otimes \bar{\Hs}^{A^{c}}\right),$$
where  $\Hs^{A}$ stands for tensor in $\Hs$ at the position given by the set $A$. 
This shows that 
$$\CLSI^{+}(L_{\mathcal{X}}^{\Hs_{F}})\geq \CLSI^{+}(\hat{L}^{m(F)}_{\mathcal{X}}) .$$
By taking the infimum over $m$ and the infimum over $F$, we obtain a lower bound for $L_{\mathcal{X}}^{L_2(G)}$.
  \qd

Let us point out that for a tensor product $\Hs^{\ten_m}$, the induced Lindbladian in general does not coincide with the tensor product  Lindbladian
 \[ L_{X}^{m}(\rho)
  \lel [X^{m},[X^{m},\rho]]
  \lel \sum_{j,k} [\pi_j(X),[\pi_k(X),\rho]]
  \neq
  \sum_j [\pi_j(X),[\pi_j(X),\rho]]  \pl. \]
%In some special cases, however, we have such an equality. We will abuse notation and write $\pi_j(L_X)=[\pi_j(X),[\pi_j(X),\pl]]$.
\section{Connected Graphs}
%  \textcolor{dblue}{See more details in the Appendix. We can show that $\delta$ is a derivation by both giving concrete multiplication definitions and embedding the graph into the diagonal matrices. The second part was suggested by Marius, which would be consistent with Section 7. It has not been included in the Appendix.}
In this section, we study $\CLSI$ constants and stability properties of connected graphs.
Let ${\sf{G}}=(\mathscr{V},\mathscr{E},\mu,w)$  be a connected graph with $\mathscr{V}=\left(v_{1},v_{2},\dots. v_{n}\right)$, where $\mu: \mathscr{V}\to (0,1)$ is a probability measure and $w$ is a symmetric weight function over the edges. Let $$\mathscr{V}(D)=(v_{1}, \dots, v_{1}, v_{2}, \dots, v_{n}, \dots, v_{n} )$$ be an ordered set with  $\rm{degree}(v_{i})$ copies of $v_{i}$'s. We define the derivation $$\delta: L_{\infty}(\mathscr{V})\to L_{\infty}(\mathscr{V}(D)), f\mapsto (\delta(f)(v_{1}), \dots, \delta(f)(v_{n})),$$ 
 $$\text{where} \quad \delta(f)(v_{r})=(\sqrt{w_{s_{1}r}}(f(v_{s_1})-f(v_{r})), \dots,\sqrt{w_{s_{k}r}}(f(v_{s_k})-f(v_{r})))_{s_{1}<\dots<s_{k} ;(v_{s_j}, v_{r})\in\mathscr{E}\text{ and } 1\leq j\leq k}.$$
% where $\delta(f)(x)=(\sqrt{w_{yx}}(f(y)-f(x)))_{y}$. \textcolor{red}{ordered as well.}
We define the left and right representations  $\pi_{1,2}:L_{\infty}(\mathscr{V})\to L_{\infty}(\mathscr{V}(D))$ by   $$\pi_{1,2}: , \pi_{1,2}(f)\mapsto \left( \pi_{1,2}(f)(v_{1}), \dots, \pi_{1,2}(f)(v_{n})\right),$$ 
where 
\begin{align*}  \quad \pi_{1}(f)(v_{r})&=\left( f(v_{s_1}), \dots, f(v_{s_k})  \right)_{s_{1}<\dots<s_{k} ;(v_{s_j}, v_{r})\in\mathscr{E}\text{ and } 1\leq j\leq k}
\end{align*}
and
\begin{align*} \quad \pi_{2}(f)(v_{r})&=\left( f(v_{r}), \dots, f(v_{r}) \right)_{\text{length}=\text{degree}(v_{r})}.
\end{align*}
% \textcolor{red}{Why are the definitions so cumbersome? Because we need to emphasize the order of the elements. Otherwise, we cannot define the $\cdot$ multiplication (see below) properly. If we define $\delta(f)$ for some edge $e$ instead of a vertex, we are going to face the same difficulty. \\This is what we had in the appendix before plus heavier notations.}\\
Thus $\delta$ satisfies the Leibniz rule 
\begin{align*}
\delta(fg)=\pi_{1}(f)\cdot \delta(g)+\delta(f) \cdot \pi_{2}(g),
\end{align*}  
where $\cdot$ is entry-wise multiplication.
%The derivation $\delta(f)(x)$ of $f\in L_{\infty}(\mathscr{E},\mu)$ at $x\in \mathscr{V}$ is a vector given by \begin{align} \delta_{w}(f)(x)=\left(\sqrt{ w_{yx}}(f(y)-f(x)) \right)_{y}. \end{align} For more about $\delta$, see \textcolor{red}{Appendix A}. (\textcolor{red}{Marius, would you please read Appendix A? I do understand that you do not like the concrete examples, but maybe you can \textbf{ONLY} read texts in red.-H})
The \textit{Fisher information} $I_{\delta,w}$ of $f$ is defined by 
$$I_{\delta, w}^{\mu}(f)=\sum_{x\in \mathscr{V}} I_{x}(f)\mu(x),$$ where $I_{x}(f)$ is the  \textit{pointwise Fisher information} at $x$ defined by 
\begin{align} \label{pt-i}
I_{x}(f) %=\tau\left(  \delta_{w}(f)(x) \delta_{w}(\ln(f))(x) \right)
=\sum_{y}\tau\left(w_{yx}(f(y)-f(x))(\ln(f(y))-\ln(f(x))\right).
\end{align}
We may further decompose $I_{x}(f)$ as $I_{x}(f)=\sum I_{y,x}(f)$ with the \textit{edge Fisher information} $I_{y,x}$ defined by 
\begin{align} \label{eg-i}
I_{y,x}(f)=w_{yx}\tau\left( (f(y)-f(x))(\ln(f(y))-\ln(f(x)))  \right).
\end{align}

% \textcolor{dblue}{see appendix for the normalization.}

We use $\delta(f)$ and  $I_{\delta}(f)$ if the weight and the measure are clear from the context. The connected graph ${\sf{G}}$ constitutes a concrete example of derivation triple. Indeed, let $\mathcal{N}$ be $L_{\infty}(\mathscr{V},\mu)$ and $\mathcal{M}$ be bounded sections of the discrete Clifford bundle. We are particularly interested in regular-weighted \rm{(}$w_{xy}=1, \forall (x,y)\in \mathscr{E}$\rm{)} graphs with a uniform distribution over the vertices, denoted by ${\sf{G}}=(\mathscr{V},\mathscr{E})$. Let ${\sf{G}}=(\mathscr{V}^{\circ}, \mathscr{E}^{\circ},\mu, w)$ denote a cyclic graph $(\mathscr{V}^{\circ}, \mathscr{E}^{\circ})$ with a probability measure $\mu$ and a weight function $w$, and ${\sf{G}}^{\circ}=(\mathscr{V}^{\circ}, \mathscr{E}^{\circ})$ denote a regular-weighted cyclic graph $(\mathscr{V}^{\circ}, \mathscr{E}^{\circ})$ with a uniform distribution. 
\begin{lemma}\label{hook-pre} Let $g(t)=(1-t)\rho+t\sigma$ for $t\in[0,1]$, then $$\int_{0}^{1}\tau \left((\rho-\sigma) Q^{g(t)} (\rho-\sigma) \right)dt=\tau \big( (\rho-\sigma) (\ln(\rho)-\ln(\sigma)) \big).$$ 
\end{lemma}
\begin{proof} Noting $g'(t)=\sigma-\rho$, we obtain that 
\begin{align*}
&\int_{0}^{1}\tau \big((\rho-\sigma) Q^{g(t)} (\rho-\sigma) \big )dt\\
\lel& \int_{0}^{1} \tau \big(   g'(t) Q^{g(t)} (g'(t)) \big) dt\lel \int_{0}^{1} \tau \big( g'(t) \left(\ln(g(t))\right)'   \big)dt
\end{align*}
By integration by parts, then 
\begin{align*}
\int_{0}^{1} \tau \big( g'(t) \left(\ln(g(t))\right)'   \big)dt
\lel& -\int_{0}^{1} \tau\big(g''(t)\ln(g(t))\big)dt+ \tau \big( g'(t)\ln(g(t))\big) |_{0}^{1}\\
\lel&\tau\big( (\rho-\sigma) (\ln(\rho)-\ln(\sigma)) \big).
\end{align*}
\end{proof}
\noindent We can give concrete estimates of $\CLSI({\sf{G}}^{\circ})$ and $\CLSI^{+}({\sf{G}}^\circ)$.
\begin{lemma} \label{hook2} Let ${\sf{G}}^{\circ}=(\mathscr{V}^{\circ},\mathscr{E}^{\circ})$ be a regular-weighted cyclic graph with a uniform  distribution over the vertices and $\mathscr{V}^{\circ}=\{1,2,\dots,n\}$, then we have $$\CLSI({\sf{G}}^{\circ})\geq\CLSI^{+}({\sf{G}}^{\circ})\geq \frac{16}{45n^2}.$$
\end{lemma}
\begin{proof} 
Let $I_{1}=(0,\frac{1}{n})$ and $I_{k}=[\frac{k-1}{n}, \frac{k}{n})$ for $k>1$. We further divide $I_{k}=I_{k}^{1}\cup I_{k}^{2}\cup I_{k}^{3}$ into three intervals of equal length $\frac{1}{3n}$. Let $f: \{1,2,\dots, n\}\to \mathcal{M}$ be a function with values in a finite von Neumann algebra $\mathcal{M}$ such that $c\leq f(k) \leq c^{-1}$ for some $c>0$. 
Let $$f(n+1)=f(1) \quad \text{and} \quad f(0)=f(n),$$ then we may define a function
$$
  F(t) =
  \begin{cases}
\frac{3nt}{2} (f(k)-f(k-1))+\frac{3k-2}{2}f(k-1)-\frac{3k-4}{2}f(k),& \text{for } t\in I_{k}^{1};\\
                                   f(k), & \text{for }  t\in I_{k}^{2}; \\
\frac{3nt}{2}(f(k+1)-f(k))+\frac{3k+1}{2}f(k)-\frac{3k-1}{2}f(k+1),& \text{for } t\in  I_{k}^{3}.
  \end{cases}
$$
We  conclude that $c\leq F(t)\leq c^{-1}$ and $E_{\fix}(f)=E_{\fix}(F)$. (Note $F$ is not differentiable, then we consider the convolution $F*g_{m}$ for the dilation $g_{m}(x)=mg(mx) $ with support in $[-\frac{1}{m},\frac{1}{m}]$.)  Let $\xi=E_{\fix}(f)$ and $I=\cup_{k}I_{k}^{2}$, then
\begin{align*}
D_{\Lin}(f\|\xi)
= 3\int_{I}[\tau \left( F(t) \ln F(t)\right)-\tau(F(t)\ln \xi)+\tau(\xi)-\tau(F(t))] dt.
\end{align*}
By the nonnegativity of the Lindblad relative entropy we obtain that
\begin{align*}
D_{\Lin}(f\|\xi)
\leq % 3\int_{0}^{1}[\tau \left( F(t) \ln F(t)\right)-\tau(F(t)\ln \xi)+\tau(\xi)-\tau(F(t))] dt\\
 3D_{\Lin}(F\|\xi).
\end{align*}
Now we discuss the Fisher information term 
$$I_{\tilde{\delta}}(F)=\int_{0}^{1} \tau\left( F'(t) Q^{F(t)}F'(t) \right)dt,$$
where $\tilde{\delta}$ is the ordinary derivative.
Indeed, we may consider $I_{\tilde{\delta}}(F*g_{m})$ which is well-defined and vanishes on $I^{2}_{k}$. Assume that  the double operator integral is uniformly bounded, and so is $(F*g_{m})'(t)$. This implies that 
$$\lim_{m}I_{\tilde{\delta}}(F*g_{m})=\int_{I'} \tau\left( F'(t) Q^{F(t)}F'(t) \right)dt,$$
where $I'=(0,1)\backslash I.$ 
% Let us consider $I_{1}^{3}\cup I_{2}^{1}=[\frac{2}{3n},\frac{4}{3n})$, and define $f(1)=\rho$, $f(2)=\sigma$.   We want to map $I_{1}^{3}\cup I_{2}^{1}$ to $[0,1]$ by $t(s)=\frac{2}{3n}+\frac{2s}{3n}$, thus
Without loss of generality, let us consider $I_{1}^{3}\cup I_{2}^{1}=[\frac{2}{3n},\frac{4}{3n})$. Let $f(1)=\rho$ and $f(2)=\sigma$ and define $$a(s)=(1-s)\rho+s\sigma.$$ Then $F(t)=a(s)$ with the substitution $s=\frac{3n}{2}(t-\frac{2}{3n})$ for $t\in[\frac{2}{3n},\frac{4}{3n})$.  Thus
\begin{align*}\int_{\frac{2}{3n}}^{\frac{4}{3n}} \tau\left( F'(t) Q^{F(t)}F'(t) \right)dt \lel \frac{3n}{2}\int_{0}^{1} \tau \left((\rho-\sigma) Q^{a(s)}(\rho-\sigma) \right) ds.
\end{align*}
Applying Lemma \ref{hook-pre}, we have
\begin{align*}\int_{0}^{1} \tau \left((\rho-\sigma) Q^{a(s)}(\rho-\sigma) \right) ds \lel  I_{2,1}(f).
\end{align*}
%By substitution and Lemma \ref{hook-pre}, we have   \begin{align*}\int_{\frac{2}{3n}}^{\frac{4}{3n}} \tau\left( F'(t) Q^{F(t)}F'(t) \right)dt & \lel \int_{\frac{2}{3n}}^{\frac{4}{3n}} \frac{9n^{2}}{4}\tau \left( (\rho-\sigma)Q^{F(t)} (\rho-\sigma) \right) dt\\&\lel \frac{3n}{2}\int_{0}^{1} \tau \left((\rho-\sigma) Q^{a(s)}(\rho-\sigma) \right) ds\\&\lel \frac{3n}{2} \tau\big((\rho-\sigma) (\ln(\rho)-\ln(\sigma)) \big)\lel \frac{3n}{2} I_{2,1}(f)\end{align*}
Recall that $$I_{\delta}(f)=\frac{1}{n} \sum_{j=1}^{n}I_{j}(f)=\frac{1}{n}\sum_{j=1}^{n} \left(I_{j+1,j}(f)+I_{j-1,j}(f)\right).$$
Summing over all these intervals, we obtain $$\int_{I'} \tau\left( F'(t) Q^{F(t)}F'(t) \right)dt\lel \frac{3n^{2}}{4} I_{\delta}(f).$$ 
Together with Example \ref{unit}, we conclude that 
$$D(f\| \xi) \leq 3D(F\| \xi) \leq \frac{15}{4} I_{\tilde{\delta}}(F)\leq \frac{45n^{2}}{16} I_{\delta}(f).$$
The same argument applies  for $p$-entropy and $p$-Fisher information.
% \textcolor{red}{Check the normalization of Fisher information. It might to be up to a constant of 2. If we renormalize the Fisher information here, we need to change the constant in the following two results as well.}
\end{proof}
\begin{definition} \label{def-cover}
$\tilde{{\sf{G}}}=(\tilde{\mathscr{V}},\tilde{\mathscr{E}}, \tilde{\mu},\tilde{w})$ is said to be a cover of ${\sf{G}}=(\mathscr{V},\mathscr{E},\mu, w)$ via $\phi$ (or ${\sf{G}}$ is covered by $\tilde{{\sf{G}}}$) if there exists a surjective map $\phi: \tilde{\mathscr{V}}\rightarrow \mathscr{V}$ satisfying the following conditions:

\begin{enumerate}[leftmargin=6mm]
\item edge preserving, i.e., $(\phi(x), \phi(y))\in \mathscr{E}$ if $(x,y)\in \tilde{\mathscr{E}}$;
\item measure preserving, i.e., $\mu(x)=\sum_{\phi(y)=x}\tilde{\mu}(y)$;
\item weight preserving, i.e., $\frac{w_{\phi(x)\phi(y)}}{m_{\phi}(\phi(x),\phi(y))}=\tilde{w}_{xy}$ if $(x,y)\in \tilde{\mathscr{E}}$, where $m_{\phi}(\phi(x),\phi(y))$ is the number of the preimages of $(\phi(x),\phi(y))$ under the function $\phi\times \phi$.
% \textcolor{red}{Marius suggested $\frac{w_{\phi(x)\phi(y)}}{\sqrt{m_{\phi}(\phi(x),\phi(y))}}=\tilde{w}_{xy}$, but I believe there would be no square root. Or maybe it was defined by Nick?}
\end{enumerate} 
\end{definition}
\noindent We denote the number of preimages $\phi^{-1}(x)$ by $m_{\phi}(x)$and $\max_{\phi}=\max_{x} m_{\phi}(x)$. Define the embedding map $$\pi: L_{\infty}(\mathscr{V},\mu)\rightarrow L_{\infty}(\tilde{\mathscr{V}},\tilde{\mu}) 
,f\mapsto f\circ \phi.$$ Applying Theorem \ref{passto}, we obtain that:
\begin{lemma} \label{cover} Let $\tilde{{\sf{G}}}$ be a cover of ${\sf{G}}$ via $\phi$,  then $$\CLSI({\sf{G}})\geq \CLSI (\tilde{{\sf{G}}})  \quad\text{and} \quad \CLSI^{+}({\sf{G}})\geq \CLSI^{+} (\tilde{{\sf{G}}}).$$ 
\end{lemma}
\noindent By Theorem \ref{sutp} $\CLSI$ and $\CLSI^{+}$ of connected graphs are  stable under change of measure and change of weight.  
\begin{corollary}\label{change-both} For connected graphs ${\sf{G}}_{1}=(\mathscr{V},\mathscr{E},\mu_{1},w_{1})$ and ${\sf{G}}_{2}=(\mathscr{V},\mathscr{E},\mu_{2},w_{2})$ with $\frac{w_{2}}{w_{1}}\leq b$ and $c_{2}\leq\frac{d\mu_{2}}{d\mu_{1}}\leq c_{1}$, we have 
\[
\CLSI({{\sf{G}}}_{1})\geq \frac{c_{2}} {c_{1}b} \CLSI({\sf{G}}_{2}) \quad
\text{and} \quad  \CLSI^{+}({\sf{G}}_{1})\geq \frac{c_{2}} {c_{1}b} \CLSI^{+}({\sf{G}}_{2}).
\]
\end{corollary}
\begin{proof} Noting ${\sf{G}}_{1}$ and ${\sf{G}}_{2}$ have the same fixed-point algebra $\mathcal{N}_{\fix}$, we only have to compare $I_{\delta, w_{1}}^{\mu_{1}}$ and $I_{\delta,w_{2}}^{\mu_{2}}$. By the definition of Fisher information, we have
\begin{align*}
I^{\mu_{1}}_{\delta,w_{2}} (f)\leq b I^{\mu_{1}}_{\delta, w_{1}}(f). 
\end{align*} 
Together with Theorem \ref{sutp}, thus
$$D^{\mu_{1}, w_{1}}_{\mathcal{N}_{\fix}}(f)=D^{\mu_{1}, w_{2}}_{\mathcal{N}_{\fix}}(f)\leq \frac{c_{1}}{c_{2}\lambda}I_{\delta, w_{2}}^{\mu_{1}}(f)\leq \frac{c_{1}b}{c_{2}\lambda} I_{\delta, w_{1}}^{\mu_{1}}(f).$$
The same argument applies for $\CLSI^{+}$.
\end{proof}
In graph theory, a tree $\sf{T}=(\mathscr{V}_{\sf{T}}, \mathscr{E}_{\sf{T}})$ is a connected and undirected (symmetric weighted) graph with no cycles. 
A rooted tree, where the root is singled out, comes with a hierarchical data structure. A tree traversal is to traverse (visit) each node (vertex) in the data structure. Recall that the any vertex of degree one can be chosen as a root, and the vertices directly connected to the root are called the children of the root.  A tree $\sf{T}_{s}=(\mathscr{V}_{s}, \mathscr{E}_{s})$ is said to be a spanning tree of a graph $(\mathscr{V},\mathscr{E})$ if  $\mathscr{V}_{s}=\mathscr{V}$ and $\mathscr{E}_{s}\subset \mathscr{E}$. It is well-known that every connected graph has a spanning tree, and we may find the minimum spanning tree within time $\mathcal{O}(|\mathscr{E}|\log(|\mathscr{V}|))$ by using Kruskal's algorithm. \cite{Kruskal} \\
\begin{lemma}\label{change-g} Any tree ${\sf{T}}=(\mathscr{V}_{\sf{T}},\mathscr{E}_{\sf{T}}, \mu_{\sf{T}}, w_{\sf{T}})$  is covered by a cyclic graph ${\sf{G}}=(\mathscr{V}^{\circ},\mathscr{E}^{\circ},\mu,w)$ with $|\mathscr{V}^{\circ}|= 2|\mathscr{E}_{\sf{T}}|$.  Moreover, there exist $\mu_{\sf{T}}'$ and $w_{\sf{T}}'$ such that 
${\sf{T}}'=(\mathscr{V}_{\sf{T}}, \mathscr{E}_{\sf{T}}, \mu_{\sf{T}}', w_{\sf{T}}')$ is covered by a cyclic graph ${\sf{G}}^{\circ}=(\mathscr{V}^{\circ},\mathscr{E}^{\circ})$ with $|\mathscr{V}^{\circ}|=2|\mathscr{E}_{\sf{T}}|$.
\end{lemma}

\begin{proof} 
By the preorder traversal, we develop the following recursive algorithm. We start with an empty graph $(\mathscr{V}^{\circ},\mathscr{E}^{\circ})=(\emptyset, \emptyset)$ and define $\phi:  \mathscr{V}^{\circ}\to\mathscr{V}_{\sf{T}}$  in the algorithm.
\begin{enumerate}
\item[step 1:] Select a root $v_{1}$, and label it as vertex $v'_{j}$ for $j=1$. (We say $v_{j}$ has been visited.) Define $\phi(v'_{1}):=v_{1}$, and update the vertex set $\mathscr{V}^{\circ}:=\mathscr{V}^{\circ}\cup \{v'_{1}\}$.
\item[step 2:]  If $v$ has unvisited children, select an unvisited child $v_{c}$ and label it as vertex $v'_{j+1}$, i.e.,  define $\phi(v'_{j+1}):=v_{c}.$ If $v$ has no unvisited child, then go back to the parent $v_{p}$ of $v$ and label the parent again using $v'_{j+1}$, i.e., define $\phi(v'_{j+1})=v_{p}$. We also record the edge $(v'_{j},v'_{j+1})$. Update the vertex set $\mathscr{V}^{\circ}:=\mathscr{V}^{\circ}\cup \{v'_{j+1}\}$ and the edge set $\mathscr{E}^{\circ}:=\mathscr{E}^{\circ}\cup \{ (v'_{j},v'_{j+1})\}$.   Assign the value $(j+1)$ to $j$.  
\item[step 3:] Repeat step 2 until the root $v_{1}$ is visited twice.
\end{enumerate}
Every edge of $\mathscr{E}_{\sf{T}}$ is traversed twice, then $|\mathscr{V}^{\circ}|=|\mathscr{E}^{\circ}|=2|\mathscr{E}_{\sf{T}}|$.
Note $\phi: \mathscr{V}^{\circ}\rightarrow \mathscr{V}_{\sf{T}}$ is defined in thr  algorithm is surjective and edge preserving. Let
$$\mu(x)=\frac{1}{m_{\phi}(\phi(x))} \mu_{\sf{T}}(\phi(x)) \quad \text{and} \quad w_{xy}=\frac{1}{m_{\phi}(\phi(x),\phi(y))} {w_{\sf{T}}}_{\phi(x)\phi(y)}, \quad \forall x,y\in \mathscr{V}^{\circ},$$
then  ${\sf{G}}=(\mathscr{V}^{\circ}, \mathscr{E}^{\circ}, \mu, w)$ is a cover of $\sf{T}$. We can also define the measure and weight of $(\mathscr{V}_{\sf{T}}, \mathscr{E}_{\sf{T}})$ by
$$ \mu'_{\sf{T}}(x)=\sum_{\phi(y)=x}\frac{1}{2|\mathscr{E}_{\sf{T}}|} \quad  \text{and} \quad {w'_{\sf{T}}}_{xy}={m_{\phi}(x,y)}, $$
then ${\sf{T}}'=(\mathscr{V}_{\sf{T}}, \mathscr{E}_{\sf{T}}, \mu_{\sf{T}}', w_{\sf{T}}')$ is covered by  ${\sf{G}}^{\circ}=(\mathscr{V}^{\circ},\mathscr{E}^{\circ})$. 
\end{proof}

\begin{theorem} Let ${\sf{G}}=(\mathscr{V},\mathscr{E},\mu, w)$ be a connected graph, then 
$$\CLSI({\sf{G}})\geq \CLSI^{+}({\sf{G}})>0.$$
\end{theorem}
\begin{proof} Let ${\sf{T}}_{s}=(\mathscr{V}_{s}, \mathscr{E}_{s}, \mu, w^{s})$ be a spanning tree of ${\sf{G}}$. Then $\sf{T}_{s}$ and ${\sf{G}}$ have the same fixed-pointed algebra since $\mathscr{V}_{s}=\mathscr{V}$.  Noting $w^{s}_{x,y}=w_{x,y}$ for $(x,y)\in \mathscr{E}_{s}$, we obtain that
$$I_{\delta, w^{s}}^{\mu}(f)\leq I_{\delta, w}^{\mu}(f), \quad \forall f\in L_{\infty}(\mathscr{V},\mu).$$
Thus we have $\CLSI({\sf{G}})\geq \CLSI(\sf{T}_{s}).$ By Lemma \ref{change-g}, there exists  ${\sf{T}}_{s}'=(\mathscr{V}_{s}, \mathscr{E}_{s}, \mu', w')$  covered by ${\sf{G}}^{\circ}=(\mathscr{V}^{\circ},\mathscr{E}^{\circ})$. By Lemma \ref{cover} and Lemma \ref{hook2}, we have $$\CLSI(\sf{T}_{s}') \geq \CLSI({\sf{G}}^{\circ})\geq \frac{16}{45|\mathscr{V}^{\circ}|^{2}}=\frac{4}{45|\mathscr{E}_{s}|^{2}}.$$  Applying Corollary \ref{change-both} and Lemma \ref{hook2}, we have 
$$\CLSI({\sf{G}})\geq \CLSI({\sf{T}}_{s})\geq \frac{4}{45|\mathscr{E}_{s}|^{2} \|\frac{d\mu}{d\mu'}\|_{\infty} \|\frac{d\mu'}{d\mu}\|_{\infty} \|\frac{w'}{w^{s}}\|_{\infty}}\pl.$$ 
The same argument applies for $\CLSI^{+}$.
\end{proof}
\begin{corollary}\label{cyclic-clis}
Let ${\sf{G}}=(\mathscr{V},\mathscr{E})$ be a connected graph with the maximum degree $d$ and $l$ be the number of the edges of the minimum spanning tree of the graph. Then we have 
\begin{align*} \CLSI({\sf{G}})\geq  \CLSI^{+}({\sf{G}})\geq \frac{2}{45 l^{2}d}\pl.
\end{align*}
\end{corollary}
%\begin{proof}  $|\mathscr{E}_{s}|=l$.\\ $\mu(x)=\frac{1}{|\mathscr{V}|}$. $\mu'(x)=\frac{degree(x)}{|2\mathscr{E}|}$.\\ $w^{s}=1$. $w'=2$. \end{proof}

% \begin{proof} The weight ${w_{T}}'_{xy}\leq 2$ since the recursive algorithm visits each edge at most twice. ${w_{T}}'_{xy}=1$, if  $(x,y)$ is the edge connected the root and its direct child. Otherwise ${w_{T}}'_{xy}=\sqrt{2}$. Since $m_{\phi}(\phi(x))=\text{degree}(x)$, then $\mu_{T}^{'}(x)=\frac{\text{degree}(x)}{2n}$.\end{proof}
%%%%% this is details of the proof above $$\CLSI(G)\geq \CLSI(T_{s})\geq \frac{\|\frac{\mu'}{\mu}\|_{\infty, V_{s}}\|\frac{\mu}{\mu'}\|_{\infty, V_{s}}}{\|\frac{w'}{w}\|_{\infty, E_{s}}} \CLSI(T_{s}') \geq \frac{\|\frac{\mu'}{\mu}\|_{\infty, V_{s}}\|\frac{\mu}{\mu'}\|_{\infty, V_{s}}}{\|\frac{w'}{w}\|_{\infty, E_{s}}} \CLSI(G^{\circ}) \geq \frac{\|\frac{\mu'}{\mu}\|_{\infty, V_{s}}\|\frac{\mu}{\mu'}\|_{\infty, V_{s}}}{\|\frac{w'}{w}\|_{\infty, E_{s}}} \frac{c}{2|E_{T}|^{2}}, $$
\section{From graphs to graph H\"ormander systems}
We extend $\CLSI$ from a commutative subsystem $\ell_{\infty}^{n}\subset\Mz_{n}$ to the full noncommutative system $\Mz_{n}$. We will perform this task for the so-called graph H\"ormander system.
Let ${\sf{G}}=(\mathscr{V},\mathscr{E})$ be a connected graph with $\mathscr{V}=\{v_{1}, v_{2},\dots, v_{n}\}$ ($n\geq 2$). Here we adopt the convention in  Section 6 that ${\sf{G}}$ is regular-weighted and equipped with  the uniform distribution  over the vertices $\mathscr{V}$. Then   $L_\infty(\mathscr{V})=\ell_{\infty}^n$ is a subalgebra of $\Mz_{n}$ via the $^*$ homomorphism $\pi$ defined by
$\pi(f)=\left(\begin{smallmatrix}
 f(1)&0  &\cdots & 0\\
 0& f(2)& \cdots &0\\
         \vdots&  \vdots& \vdots&                      \vdots \\
                           0&     0&\cdots & f(n)
\end{smallmatrix}\right)$. 
We will work with the Lie group $SO_n$ with the normalized Haar measure $\mu$ and Lie algebra $\mathfrak{so}_n=\{a| a^{\trans}=-a\}$ of real anti-symmetric matrices.  Since ${\sf{G}}$ is undirected, then.  $e=(v_{r},v_{s})=(v_{s},v_{r})$. Throughout this section, we represent $e=(v_{r},v_{s})$ using $r<s$. We define
 \[ X_{e}=|r\ran\lan s|-|s\ran\lan r|, \quad \forall e=(v_r,v_s)\in\mathscr{E}.\]
Then $X_{\mathscr{E}}\lel \{X_{e}, e\in \mathscr{E}\}\subset \mathfrak{so}_{n}$. 
We define the edge Laplacian $\Delta_{e} : C^{\infty}(SO_{n}) \to C^{\infty}(SO_{n})$ by 
$$ \Delta_{e}(f)= X_{e}^{*}X_{e}f,$$
where $(X_{e}f)(a)=\frac{d}{dt}f(\exp(tX_{e})a)|_{t=0} $ for any $a\in SO_{n}$.  The sub-Laplacian $\Delta_{\mathscr{E}}$ is the sum of edge Laplacians
$$\Delta_{\mathscr{E}}= \sum_{e\in \mathscr{E}}\Delta_{e}.$$
\begin{lemma} The following conditions are equivalent.
\begin{enumerate}[leftmargin=6mm]
\item[(1)] $\Delta_{\mathscr{E}}$ is ergodic, i.e. $\Delta_{\mathscr{E}}(f)=0$ iff $f=\la 1$ for some $\la$;
 \item[(2)]  $(\mathscr{V},\mathscr{E})$ is a connected graph;
\item[(3)]  $X_\mathscr{E}'\lel \cz$, i.e. the commutant of $X_{\mathscr{E}}$ is trivial.
\end{enumerate}
\end{lemma}
\noindent If any condition of above is satisfied, we call $X_{\mathscr{E}}$  a \textit{graph H\"ormander system} of the graph ${\sf{G}}$.  We define the Lindblad operator
$$L_{e}(\rho)=x_{e}^{2}\rho+\rho x_{e}^{2}-2x_{e}\rho x_{e},\forall \rho\in \Mz_{n},$$
where $x_{e}=iX_{e}$. Then  there exists a derivation $\delta_{e}(a)=-i[x_e,a]=[X_{e}, a]$ such that $L_{e}=\delta_{e}^{*}\bar{\delta_{e}}$. We define the Lindblad operator $L_{\mathscr{E}}$ associated to $\mathscr{E}$ by
  $$L_{\mathscr{E}}=\sum_{e=(r,s)\in\mathscr{E},r<s}L_{e}.$$
Noting $L_{\mathscr{E}}$ restricted the diagonal is the graph Laplace operator, we have the following result. (See Appendix A for an example.)
%Define the Lindblad operator using $x_{e}=iX_{e}$, $L_{e}(a)=x_{e}^{2}a+ax_{e}^{2}-2x_{e}ax_{e}$ and $L_{\mathscr{E}}=\frac{1}{2}\sum_{e}L_{e}.$ Then we see that $L_{e}$ is the self-adjoint Lindblad operator with derivation $\delta_{e}(a)=-i[x_e,a]=[X_{e}, a]$. It is easy to verify the following property.
\begin{prop} 
\label{lind-lap} 
Let ${\sf{G}}=(\mathscr{V},\mathscr{E})$ be a connected graph, then 
\begin{align*}\CLSI(\ell_{\infty}^{n}, L_{\mathscr{E}})=\CLSI({\sf{G}})  \text{ and }
\CLSI(\ell_{\infty}^{n}, L_{\mathscr{E}})=\CLSI({\sf{G}})  
\end{align*}
\end{prop}
\noindent Let $\rm{K}_{rs}$ be an $n$-by-$n$ matrix with $1$ on $(r,s)$ and $0$'s otherwise. For $e=(r,s)\in\mathscr{E}$, we define $\A_e$ to be the the subalgebra  of $\Mz_{n}$ generated by $\{1, \rm{K}_{rr}, \rm{K}_{ss}, \rm{K}_{rs}, \rm{K}_{sr}, \rm{K}_{kj}\}$, where  $k\neq r,s, j\neq r,s \text{ and } k\neq j $. Indeed $\mathcal{A}_{e}$ is generated by block diagonal matrices up to permutations. Let $E_{e}: \Mz_{n}\to \A_{e}$ be the conditional expectation onto the sub-algebra $\mathcal{A}_{e}$ and $E_{\infty}$ be the conditional expectation onto the diagonal matrices. Throughout this section, let $\mathcal{M}$ be a finite von Neumann algebra equipped with a normal faithful trace $\tau_{\mathcal{M}}$.
\begin{lemma}\label{edge} For $\rho\in  \Mz_n\ten \MM$,  we have
 \[ \CLSI(\Mz_{n}, L_{e}) D(\rho\|E_{e}(\rho)) \le I_{L_{e}}(\rho) \pl 
 \text{ and } \CpSI(\Mz_{n}, L_{e})  d^{p}(\rho\|E_{e}(\rho)) \le I_{L_{e}}^{p}(\rho) \pl .\]
\end{lemma}

\begin{proof}  Without loss of generality we work with $e=(v_{1} ,v_{2})$. The fixed-point algebra of $L_{e}$ is given by the commutant $\mathcal{N}_e=\{x_e\}'$.
%%The spectrum of $x_e$ consists of $0,1,-1$.
 Note $\mathcal{N}_e$ is a subalgebra of $\mathcal{A}_e$. Indeed, let $\rho=(\rho)_{ij}\in\mathcal{N}_{e}$, then $\rho x_{e}=x_{e}\rho$. Thus
$$\left(\begin{smallmatrix}  -\rho_{12} & \rho_{11} & 0 &\cdots & 0\\
 -\rho_{22} & \rho_{21} & 0 &\cdots & 0\\
 -\rho_{32} & \rho_{31} & 0 &\cdots & 0\\
          \vdots & \vdots & \vdots  &  \vdots &\vdots \\
-\rho_{n2} &\rho_{n1} & 0 &\cdots & 0
\end{smallmatrix}\right)= \left(\begin{smallmatrix}
 \rho_{21 }& \rho_{22} &  \rho_{23} &\cdots & \rho_{2n}\\
        -\rho_{11 }& -\rho_{12} &  -\rho_{13} & \cdots & - \rho_{1n}   \\
0  & 0 & 0 &\cdots&0 \\
\vdots &  \vdots& \vdots&  \vdots& \vdots\\
0  & 0 &  0 &\cdots&0   
\end{smallmatrix}\right).$$
%%%  Indeed, let $a\in N_e$ and $f=e_{rr}+e_{ss}$. Using the different signs in the spectrum, we deduce that  $a=faf+(1-f)a(1-f)$, where $(1-f)a(1-f)$ is an arbitrary matrix and $faf$ commutes with $x_e$. 
By \eqref{fix-d}, we deduce that 
 \[ D(\rho\|E_{e}(\rho) \kl D(\rho\|E_{N_e}(\rho)) \kl \frac{1}{\CLSI(\Mz_{n}, L_{e})} I_{L_{e}}(\rho), \forall \rho\in \Mz_{n}\ten \Mm\pl.\]
For the $p$-version, it follows from the facts that $d^{p}(\rho\|\sigma)$ is non-negative and the noncommutative martingale equality  (See \cite{hao})
\qd
\begin{lemma} \label{edge-expectation} The conditional expectations $\{E_{e}\}_{e\in\mathscr{E}}$ commute pairwise. Moreover $\prod_{e\in\mathscr{E}} E_{e}=E_{\infty}$.
\end{lemma}
\begin{proof} 
Noting $E_{e}$ is a Schur multiplier,  we infer that all conditional expectations commute. It is obvious that  $\prod_{e\in\mathscr{E}} E_{e}(a)=a$ for any diagonal matrix $a\in\Mz_{n}$. Now assume that $\prod_{e\in\mathscr{E}}E_{e} (\rho)=\rho$ for some $\rho\in \Mz_{n}$. Then $E_{e}(\rho)=\rho$ for any $e=(v_{r},v_{s})\in\mathscr{E}$.
 By the proof of Lemma \ref{edge}, we have $\rho_{jr}=\rho_{rj}=\rho_{js}=\rho_{sj}=0$ for any $j\neq r,s$. Thus $\rho$ is diagonal, which yields the first assertion. Since ${\sf{G}}$ is connected, we conclude that $\rho_{ij}=0$ for $i\neq j$.
\qd
% Let $a=(a_{kl})$ be a $n\times n$ matrix such that $E_e(a)=a$ for all $e\in E$. For $e=(r,s)$ we see  $\E_e(a)=a$ implies  that $a_{kl}=0$ whenever $k\in \{r,s\}$ and $l\notin\{r,s\}$. Let us use the notation $\{e\}=\{r,s\}$.
% {\bf Case 1:} For every $k\neq l$ there exists two different edges $e_1\in E$ and $e_2\in E$ such that $\{e_1\}\cap \{e_2\}=\{k\}$ or $\{e_1\}\cap \{e_2\}=\{l\}$. Since $E$ is symmetric, we may then assume that $e_1=(k,s_1)$, $e_2=(k,s_2)$ and either $l\neq s_1$ or $l\neq s_2$. If $l\neq s_1$, we deduce form $E_e(a)=a$ that $a_{kl}=0$. The argument for $l$ being in the intersection is similar.
% {\bf Case 2:} There exists a pair $(k,l)$ such that $e=(k,l)$ and $e=(l,k)$  is the only edge which contains which contains the vertex $k$ or the vertex $l$. Then $\{k,l\}$ is an irreducible block of the graph and hence $E$ is irreducible, unless $n=2$. This contradicts our assumption.
% We deduce that $\prod_{e\in E} E_e=a$ implies that $a$ is a diagonal matrix. Since diagonal; matrices satisfy $E_e(a)=a$ for all edges of $V=\{1,...,n\}$, we deduce that indeed $\prod_{e\in E} E_e=E_{\infty}$. Finally, for $\rho \in \Mz_n\ten \MM$, we may apply the same argument and observe that $E_{\infty}\ten id_\MM=\prod_{e\in E} E_e\ten id_{\MM}$.
\begin{lemma}\label{diagon}  For $\rho\in \Mz_n\ten \MM$, we have 
 \[ D(\rho\|E_{\infty}(\rho))\kl 5\pi^{2} I_{L_{\mathscr{E}}}(\rho) \pl \]
 and 
  \[ D^{p}(\rho\|E_{\infty}(\rho))\kl 5\pi^{2} I^{p}_{L_{\mathscr{E}}}(\rho) \pl .\]
\end{lemma}

\begin{proof} 
% We observe that projections $E_{e}$ are all Schur multipliers. Recall that a Schur multiplier $M_{\si}$ is of the form $ M_{\si}(a) \lel (\si_{rs}a_{rs}) \pl .$
%% Note that the conditional expectations $E_{e}$ defined in Lemma \ref{edge} for different edges are commuting and the product $\prod_{e\in E} E_e $ is again a conditional expectation. 
By Lemma \ref{edge-expectation}, we have $D(\rho\|E_{\infty}(\rho)) = D(\rho\|\prod_{e\in \mathscr{E}} E_e(\rho))$. Together with Lemma \ref{iter} and Lemma \ref{edge}, we deduce that
 $$ \inf_{e\in\mathscr{E}}\{\CLSI(\Mz_{n}, L_{e})\} D(\rho\|\prod_{e\in \mathscr{E}} E_e(\rho)) \kl  \inf_{e\in\mathscr{E}}\{\CLSI(\Mz_{n}, L_{e})\}\sum_{e\in \mathscr{E}} D(\rho\|E_{e}(\rho))\kl I_{L_{\mathscr{E}}}(\rho).$$
 We obtain $5\pi^{2}$ by Example \ref{lind-1}. 
 The same argument applies for the $p$-version.
\qd
\begin{lemma}\label{diagon2} Let  $\rho\in \Mz_n\ten \MM$, then
\begin{align*} I_{L_{\mathscr{E}}}(E_{\infty}(\rho))\kl I_{L_{\mathscr{E}}}(\rho)  \quad \text{      and      } \quad I_{L_{\mathscr{E}}}^{p}(E_{\infty}(\rho))\kl I_{L_{\mathscr{E}}}^{p}(\rho).
\end{align*}
\end{lemma}

 \begin{proof} It suffices to prove that
  \begin{align} I_{L_{e}}(E_{\infty}(\rho))\kl I_{L_{e}}(\rho)\label{edge-convex} \end{align}
 for any edge $e\in\mathscr{E}$.   For any fixed $1\leq i\leq n-1$, define $$E_{i}^{n}=\frac{1}{2}(U_{i}^{*}\rho U_{i} +\rho),$$ where $U_{i}$ is a diagonal matrix with $-1$ for the $i$-th entry and $1$ for other diagonal entries.  Then $E_{\infty}=\prod_{i=1}^{n-1} E_{i}^{n}$ and $[E_{i}^{n}, E_{j}^{n}]=0$ since $E_{i}$ is a Schur multiplier.  By the convexity of the Fisher information in Theorem \ref{HPconvex}, we have
  \begin{align} \label{edge-2} I_{L_{2}}(E_{i}^{n}(\rho))\leq \frac{1}{2} I_{L_{e}}(\rho)+\frac{1}{2} I_{L_{e}} (U_{i}^{*}\rho U_{i}). \end{align}
 For any fixed $U_{i}$, we see that $$[a, U_{I}^{*}\rho U_{i}]=U_{i}^{*}U_{i}[a, U_{i}^{*}\rho U_{i}]U_{i}^{*}U_{i}=U_{i}^{*}[U_{i}aU_{i}^{*},\rho]U_{i}.$$ Together with the unitary invariance of the trace, then
  \begin{align*} 
I_{L_{e}}(U_{i}^*\rho U_{i}) \lel & \tau\left([X_e, U_{i}^{*}\rho U_{i}] Q^{U_{i}^*\rho U_{i}}([X_{e},U_{i}^{*}\rho U_{i}])\right)\\ \lel & \tau\left( U_{i}^{*}[U_{i}X_{e}U_{i}^{*},\rho]U_{i} Q^{U_{i}^{*}\rho U_{i}} U_{i}^{*}[U_{i}X_{e}U_{i}^{*},\rho]U_{i}  \right)\\ \lel& \tau\left( [U_{i}X_eU_{i}^*,\rho]Q^{\rho}([U_{i}X_eU_{i}^*,\rho])\right) \pl . \end{align*}
 Note that for any edge that $U_{i}X_{e}U_{i}^{*}=X_{e}$ or $-X_{e}$, thus  $I_{L_{e}}(U_{i}^*\rho U_{i}) \lel =I_{L_{e}}(\rho).$ Together with \eqref{edge-2}, we have $I_{L_{e}}(E^{n}_{i}(\rho))\leq I_{L_{e}}(\rho).$  Thus repeating this $n-1$ times yields \eqref{edge-convex} $$I_{L_{e}}(E_{\infty}(\rho))\leq I_{L_{e}} ((\prod_{i=2}^{n-1}E_{i})(\rho))\leq \cdots \leq I_{L_{e}}(\rho).$$ 
The same proof applies for the $p$-version.
\qd

\begin{theorem}\label{grpp}  Let ${\sf{G}}=(\mathscr{V},\mathscr{E})$ be a connected graph. Then
 \[ \frac{\CLSI({\sf{G}})}{1+5\pi^{2}\CLSI({\sf{G}})}\kl
  \CLSI(\Mz_{n},L_{\mathscr{E}})
 \kl \CLSI({\sf{G}})  \pl.
  \]
and 
 \[ \frac{\CpSI({\sf{G}})}{1+5\pi^{2}\CpSI({\sf{G}})} \kl
  \CpSI(\Mz_{n},L_{\mathscr{E}})
 \kl \CpSI({\sf{G}})  \pl.
  \]
\end{theorem}

\begin{proof} Note  $(\ell_{\infty}^{n}, L_{\mathscr{E}})$ is a subsystem of $(\Mz_{n}, L_{\mathscr{E}})$. By Proposition \ref{lind-lap},  we deduce the second part of the inequality . Let $E_{\fix}$ be the conditional expectation onto the fixed-point algebras of $L_{\mathscr{E}}$. By Lemma \ref{diagon2}, we have 
\begin{align} \label{grpp-1}
D(E_{\infty}(\rho)\| E_{\fix}(\rho))\leq  \frac{1}{\CLSI({\sf{G}})} I_{L_\mathscr{E}}(E_{\infty}(\rho)) \leq  \frac{1}{\CLSI({\sf{G}})} I_{L_\mathscr{E}}(\rho).
\end{align}
Together with \eqref{grpp-1} and Lemma \ref{diagon}, we obtain the first part of the inequality
 \begin{align*}
  D(\rho\|E_{\fix}(\rho))&=D(\rho\|E_{\infty}(\rho))+D(E_{\infty}(\rho)\|E_{\fix}(\rho)) \\
  &\leq \left(5\pi^{2}+\frac{1}{\CLSI({\sf{G}})} \right)I_{L_{\mathscr{E}}}(\rho).
  \end{align*}
The argument also applies for $\CpSI$.  \qd

\section*{Appendix}
 Here we give more details about Lemma \ref{unit}. Let us recall the tail approximation of Gaussian distribution,  $$\left(\frac{1}{x}-\frac{1}{x^{3}} \right)e^{-x^{2}/2}\leq  \int_{x}^{\infty} e^{-t^{2}/2}dt \leq \frac{1}{x}e^{-x^{2}/2}.$$
Let $g(x)=\frac{1}{\sqrt{2\pi}}\sum_{k=-\infty}^{\infty} e^{-(x-k)^{2}/2}$, thus $\frac{dx}{d\mu}=\frac{1}{g(x)}$. It suffices to show that 
$$ 2e^{-1/2}+2e^{-2}+2e^{-9/2}+\frac{48}{125}e^{-25/2}\leq \sqrt{2\pi}g(x)\leq 2+2e^{-1/2}+2e^{-2} +\frac{8}{3}e^{-9/2}.$$
\begin{align*}
\sqrt{2\pi}g(x) =& e^{-x^{2}/2}+\sum_{k=1}^{\infty} e^{-(x-k)^{2}/2}+\sum_{k=-\infty}^{-1} e^{-(x-k)^{2}/2}\\
= & e^{-x^{2}/2}+\sum_{k=1}^{\infty} e^{-(k-1)^{2}/2}+\sum_{k=1}^{\infty} e^{-(x+k)^{2}/2}\\
\sqrt{2\pi}g(x) \leq &1+\sum_{k=1}^{\infty} e^{-(k-1)^{2}/2}+\sum_{k=1}^{\infty} e^{-k^{2}/2}\\
=& 1+e^{-0/2}+2\sum_{k=1}^{\infty} e^{-k^{2}/2}\\
=& 2+2e^{-1/2}+2e^{-2}+2e^{-9/2}+2\sum_{k=4}^{\infty} e^{-k^{2}/2}\\
\leq & 2+2e^{-1/2}+2e^{-2}+2e^{-9/2}+2\int_{3}^{\infty} e^{-x^{2}/2} dx\\
\leq & 2+2e^{-1/2}+2e^{-2}+\frac{8}{3}e^{-9/2}\\
\sqrt{2\pi}g(x)\geq& e^{-1/2}+\sum_{k=1}^{\infty} e^{-k^{2}/2}+\sum_{k=1}^{\infty} e^{-(k+1)^{2}/2}\\
= & 2e^{-1/2}+2\sum_{k=2}^{\infty} e^{-k^{2}/2}\\
=& 2e^{-1/2}+2e^{-2}+2e^{-9/2}+2\sum_{k=4}^{\infty} e^{-k^{2}/2}\\
\geq & 2^{-1/2}+2e^{-2}+2e^{-9/2}+2\int_{5}^{\infty} e^{-x^{2}/2}dx\\
\geq & 2^{-1/2}+2e^{-2}+ 2e^{-9/2}+2\left( \frac{1}{5}-\frac{1}{5^{3}} \right) e^{-25/2}\\
=& 2e^{-1/2}+2e^{-2}+2e^{-9/2}+\frac{48}{125} e^{-25/2}
\end{align*}
\bibliographystyle{alpha}
%\bibliography{biblin22}

\begin{thebibliography}{BBLW12}

\bibitem[AS80]{AS}
I.~Atsushi and \^O. Schôichi.
\newblock Derivations on algebras of unbounded operators.
\newblock {\em Transactions of the American Mathematical Society}, 261(2):567,
  1980.

\bibitem[BBLW12]{Ba}
Howard Barnum, Jonathan Barrett, Matthew Leifer, and Alexander Wilce.
\newblock Teleportation in general probabilistic theories.
\newblock In {\em Proceedings of Symposia in Applied Mathematics}, volume~71,
  pages 25--48, 2012.

\bibitem[BCR20]{bardet2020}
Ivan Bardet, Angela Capel, and Cambyse Rouzé.
\newblock Approximate tensorization of the relative entropy for noncommuting
  conditional expectations, 2020.

\bibitem[BE86]{BE}
D.~Bakry and Michel \'{E}mery.
\newblock Propaganda for {$\Gamma_2$}.
\newblock In {\em From local times to global geometry, control and physics
  ({C}oventry, 1984/85)}, volume 150 of {\em Pitman Res. Notes Math. Ser.},
  pages 39--46. Longman Sci. Tech., Harlow, 1986.

\bibitem[BL76]{BL}
J.~Bergh and J.~L{\"o}fstr{\"o}m.
\newblock {\em Interpolation spaces.}
\newblock Springer-Verlag, Berlin, 1976.
\newblock Grundlehren der Mathematischen Wissenschaften, No. 223.

\bibitem[Bou02]{bous}
Olivier Bousquet.
\newblock A {B}ennett concentration inequality and its application to suprema
  of empirical processes.
\newblock {\em C. R. Math. Acad. Sci. Paris}, 334(6):495--500, 2002.

\bibitem[BR76]{BR}
Ola Bratteli and Derek~W. Robinson.
\newblock Unbounded derivations of von neumann algebras.
\newblock 25(2):139, 1976.

\bibitem[BR18]{BaRo}
Ivan Bardet and Cambyse Rouz\'{e}.
\newblock Hypercontractivity and logarithmic sobolev inequality for
  non-primitive quantum markov semigroups and estimation of decoherence rates.
\newblock {\em arXiv:1803.05379}, 2018.

\bibitem[BS03]{bs3}
Mikhail~Sh Birman and Michael Solomyak.
\newblock Double operator integrals in a hilbert space.
\newblock {\em Integral equations and operator theory}, 47(2):131--168, 2003.

\bibitem[BT06]{Bobtet}
Sergey~G. Bobkov and Prasad Tetali.
\newblock Modified logarithmic {S}obolev inequalities in discrete settings.
\newblock {\em J. Theoret. Probab.}, 19(2):289--336, 2006.

\bibitem[BtD95]{Dieck}
Theodor Br\"{o}cker and Tammo tom Dieck.
\newblock {\em Representations of compact {L}ie groups}, volume~98 of {\em
  Graduate Texts in Mathematics}.
\newblock Springer-Verlag, New York, 1995.
\newblock Translated from the German manuscript, Corrected reprint of the 1985
  translation.

\bibitem[Cho39]{chow}
Wei-Liang Chow.
\newblock \"uber die multiplizit\"at der schnittpunkte von hyperfl\"achen.
\newblock {\em Mathematische Annalen}, 116(1):598, 1939.

\bibitem[CM14]{CM2}
Eric~A Carlen and Jan Maas.
\newblock An analog of the 2-wasserstein metric in non-ommutative probability
  under which the fermionic fokker--planck equation is gradient flow for the
  entropy.
\newblock {\em Communications in mathematical physics}, 331(3):887--926, 2014.

\bibitem[CM17]{CM}
Eric~A Carlen and Jan Maas.
\newblock Gradient flow and entropy inequalities for quantum markov semigroups
  with detailed balance.
\newblock {\em Journal of Functional Analysis}, 273(5):1810--1869, 2017.

\bibitem[CM20]{cm20}
Eric~A Carlen and Jan Maas.
\newblock Non-commutative calculus, optimal transport and functional
  inequalities in dissipative quantum systems.
\newblock {\em Journal of Statistical Physics}, 178(2):319--378, 2020.

\bibitem[CS84]{connes2}
A.~Connes and G.~Skandalis.
\newblock The longitudinal index theorem for foliations.
\newblock {\em Publ. Res. Inst. Math. Sci.}, 20(6):1139--1183, 1984.

\bibitem[DK51]{krein1}
Yu.\~L. Daleckii and S.~G. Krein.
\newblock Formulas of differentiation according to a parameter of functions of
  hermitian operators.
\newblock 76:13, 1951.

\bibitem[dPS04]{bs1}
B.~de~Pagter and F.~A. Sukochev.
\newblock Differentiation of operator functions in non-commutative
  $l_p$-spaces.
\newblock {\em Journal of Functional Analysis}, 212(1):28, 2004.

\bibitem[dPS07]{bs2}
Ben de~Pagter and Fyodor Sukochev.
\newblock Commutator estimates and $r$-flows in non-commutative operator
  spaces.
\newblock {\em Proceedings of the Edinburgh Mathematical Society. Series II},
  50(2):293, 2007.

\bibitem[DR20]{Datta_2020}
Nilanjana Datta and Cambyse Rouzé.
\newblock Relating relative entropy, optimal transport and fisher information:
  A quantum hwi inequality.
\newblock {\em Annales Henri Poincaré}, Feb 2020.

\bibitem[DSC96]{DSF}
P.~Diaconis and L.~Saloff-Coste.
\newblock Logarithmic {S}obolev inequalities for finite {M}arkov chains.
\newblock {\em Ann. Appl. Probab.}, 6(3):695--750, 1996.

\bibitem[ELP08]{LP}
L.~Ben Efraim and F.~Lust-Piquard.
\newblock Poincar\'{e} type inequalities on the discrete cube and in the {CAR}
  algebra.
\newblock {\em Probab. Theory Related Fields}, 141(3-4):569--602, 2008.

\bibitem[GJL18]{LJR}
Li~Gao, Marius Junge, and Nicolas LaRacuente.
\newblock Fisher information and logarithmic sobolev inequality for matrix
  valued functions.
\newblock {\em arXiv preprint arXiv:1807.08838}, 2018.

\bibitem[Gro75]{Gross}
Leonard Gross.
\newblock Hypercontractivity and logarithmic {S}obolev inequalities for the
  {C}lifford {D}irichlet form.
\newblock {\em Duke Math. J.}, 42(3):383--396, 1975.

\bibitem[Gro96]{gromov}
Mikhael Gromov.
\newblock Carnot-{C}arath\'{e}odory spaces seen from within.
\newblock In {\em Sub-{R}iemannian geometry}, volume 144 of {\em Progr. Math.},
  pages 79--323. Birkh\"{a}user, Basel, 1996.

\bibitem[Gro06]{overview}
Leonard Gross.
\newblock Hypercontractivity, logarithmic {S}obolev inequalities, and
  applications: a survey of surveys.
\newblock In {\em Diffusion, quantum theory, and radically elementary
  mathematics}, volume~47 of {\em Math. Notes}, pages 45--73. Princeton Univ.
  Press, Princeton, NJ, 2006.

\bibitem[GW79]{wu79}
R.~E. Greene and H.~Wu.
\newblock {\em The Hessian comparison theorem}, pages 19--42.
\newblock Springer Berlin Heidelberg, Berlin, Heidelberg, 1979.

\bibitem[HP12]{HP}
Fumio Hiai and D.~Petz.
\newblock From quasi-entropy to various quantum information quantities.
\newblock {\em Publications of the Research Institute for Mathematical
  Sciences}, 48(3):525--542, 2012.

\bibitem[HS87]{hs87}
R.~Holley and D.~Stroock.
\newblock Logarithmic sobolev inequalities and stochastic ising models.
\newblock {\em Journal of Statistical Physics}, 46(5-6):1159--1194, 1987.

\bibitem[Irv53]{Kap}
K.~Irving.
\newblock Modules over operator algebras.
\newblock {\em American Journal of Mathematics}, 75(4):839, 1953.

\bibitem[{J-L}90]{Sau}
{J-L. Sauvageot}.
\newblock {Quantum dirichlet forms, differential calculus and semigroups}.
\newblock {\em Lecture Notes in Mathematics, Berlin Springer Verlag}, 1442:334,
  1990.

\bibitem[JBK56]{Kruskal}
Jr. Joseph B.~Kruskal.
\newblock On the shortest spanning subtree of a graph and the traveling
  salesman problem.
\newblock {\em Proceedings of the American Mathematical Society}, 7(1):48,
  1956.

\bibitem[JPPP13]{mem}
Marius Junge, Carlos Palazuelos, Javier Parcet, and Mathilde Perrin.
\newblock Hypercontractivity in group von neumann algebras.
\newblock {\em arXiv preprint arXiv:1304.5789}, 2013.

\bibitem[JRS14]{JRS18}
M~Junge, E~Ricard, and D~Shlyakhtenko.
\newblock Noncommutative diffusion semigroups and free probability.
\newblock {\em Preprint}, 3, 2014.

\bibitem[JS05]{JSher}
Marius Junge and David Sherman.
\newblock Noncommutative {$L^p$} modules.
\newblock {\em J. Operator Theory}, 53(1):3--34, 2005.

\bibitem[JZ15]{JZ}
Marius Junge and Qiang Zeng.
\newblock Noncommutative martingale deviation and {P}oincar\'e type
  inequalities with applications.
\newblock {\em Probab. Theory Related Fields}, 161(3-4):449--507, 2015.

\bibitem[KL51]{KL51}
Solomon Kullback and Richard~A Leibler.
\newblock On information and sufficiency.
\newblock {\em The annals of mathematical statistics}, 22(1):79--86, 1951.

\bibitem[KR97]{KR1}
Richard~V. Kadison and John~R. Ringrose.
\newblock {\em Fundamentals of the theory of operator algebras. {V}ol. {I}},
  volume~15 of {\em Graduate Studies in Mathematics}.
\newblock American Mathematical Society, Providence, RI, 1997.
\newblock Elementary theory, Reprint of the 1983 original.

\bibitem[kre56]{krein2}
Integration and differentiation of functions of hermitian operators and
  applications to the theory of perturbations.
\newblock 1956(1):81, 1956.

\bibitem[Lan95]{Lance}
E.~Christopher Lance.
\newblock {\em Hilbert C*-modules : a toolkit for operator algebraists}.
\newblock 1995.

\bibitem[Led11]{Led}
Michel Ledoux.
\newblock From concentration to isoperimetry: semigroup proofs.
\newblock In {\em Concentration, functional inequalities and isoperimetry},
  volume 545 of {\em Contemp. Math.}, pages 155--166. Amer. Math. Soc.,
  Providence, RI, 2011.

\bibitem[Led19]{Ld2}
Michel Ledoux.
\newblock Four talagrand inequalities under the same umbrella.
\newblock {\em arXiv preprint arXiv:1909.00363}, 2019.

\bibitem[Li20]{hao}
Haojian Li.
\newblock Generalized fisher information and complete sobolev type
  inequalities.
\newblock arXiv, 2020.

\bibitem[Lin73]{Lindblad}
G.~Lindblad.
\newblock Entropy, information and quantum measurements.
\newblock {\em Communications in Mathematical Physics}, 33(4):305--322, 1973.

\bibitem[LM16]{lm}
H.~Blaine Lawson and Marie-Louise Michelsohn.
\newblock {\em Spin Geometry (PMS-38), Volume 38}.
\newblock Princeton University Press, Princeton, 2016.

\bibitem[LR02]{Lie73}
Elliott~H Lieb and Mary~Beth Ruskai.
\newblock Some operator inequalities of the schwarz type.
\newblock In {\em Inequalities}, pages 135--139. Springer, 2002.

\bibitem[LSZ13]{LSZ12}
Steven Lord, Fedor Sukochev, and Dmitriy Zanin.
\newblock {\em Singular traces.}
\newblock De Gruyter Studies in Mathematics, 46. De Gruyter, Berlin, 2013.

\bibitem[LY93]{Ya5}
Sheng~Lin Lu and Horng-Tzer Yau.
\newblock Spectral gap and logarithmic {S}obolev inequality for {K}awasaki and
  {G}lauber dynamics.
\newblock {\em Comm. Math. Phys.}, 156(2):399--433, 1993.

\bibitem[Mey83a]{Me1}
P.-A. Meyer.
\newblock Quelques r\'{e}sultats analytiques sur le semi-groupe
  d'{O}rnstein-{U}hlenbeck en dimension infinie.
\newblock In {\em Theory and application of random fields ({B}angalore, 1982)},
  volume~49 of {\em Lect. Notes Control Inf. Sci.}, pages 201--214. Springer,
  Berlin, 1983.

\bibitem[Mey83b]{Me2}
P.-A. Meyer.
\newblock Quelques r\'esultats analytiques sur le semi-groupe
  d'{O}rnstein-{U}hlenbeck en dimension infinie.
\newblock In {\em Theory and application of random fields ({B}angalore, 1982)},
  volume~49 of {\em Lecture Notes in Control and Inform. Sci.}, pages 201--214.
  Springer, Berlin, 1983.

\bibitem[NC10]{Nielsen}
Michael~A. Nielsen and Isaac~L. Chuang.
\newblock {\em Quantum Computation and Quantum Information: 10th Anniversary
  Edition}.
\newblock Cambridge University Press, 2010.

\bibitem[OL93]{OL}
{\em Lie groups and {L}ie algebras. {I}}, volume~20 of {\em Encyclopaedia of
  Mathematical Sciences}.
\newblock Springer-Verlag, Berlin, 1993.
\newblock Foundations of Lie theory. Lie transformation groups, A translation
  of {\it Current problems in mathematics. Fundamental directions. Vol. 20}
  (Russian), Akad.\ Nauk SSSR, Vsesoyuz.\ Inst.\ Nauchn.\ i Tekhn.\ Inform.,
  Moscow, 1988 , Translation by A. Kozlowski, Translation edited by A. L.
  Onishchik.

\bibitem[OV00]{OV}
F.~Otto and C.~Villani.
\newblock Generalization of an inequality by {T}alagrand and links with the
  logarithmic {S}obolev inequality.
\newblock {\em J. Funct. Anal.}, 173(2):361--400, 2000.

\bibitem[Pas73]{Paschke}
William~L. Paschke.
\newblock Inner product modules over {$B^{\ast} $}-algebras.
\newblock {\em Trans. Amer. Math. Soc.}, 182:443--468, 1973.

\bibitem[Pet85]{petz85}
D.~Petz.
\newblock Quasi-entropies for states of a von neumann algebra.
\newblock {\em Publications of the Research Institute for Mathematical
  Sciences}, 21(4):787--800, 1985.

\bibitem[Pet09]{Jesse}
Jesse Peterson.
\newblock A 1-cohomology characterization of property ({T}) in von {N}eumann
  algebras.
\newblock {\em Pacific J. Math.}, 243(1):181--199, 2009.

\bibitem[Pis98a]{pvp}
G.~Pisier.
\newblock Non-commutative vector valued {$L\sb p$}-spaces and completely
  {$p$}-summing maps.
\newblock {\em Ast\'erisque}, (247):vi+131, 1998.

\bibitem[Pis98b]{pisier93}
Gilles Pisier.
\newblock Non-commutative vector valued $l_p$ p-spaces and completely
  $p$-summing maps.
\newblock {\em Asterisque-Societe Mathematique de France}, 247, 1998.

\bibitem[PS10]{PS10}
D.~Potapov and F.~Sukochev.
\newblock Double operator integrals and submajorization.
\newblock {\em Mathematical Modelling of Natural Phenomena}, 5(4):317 -- 339,
  2010.

\bibitem[RS76]{SR}
Linda~Preiss Rothschild and E.~M. Stein.
\newblock Hypoelliptic differential operators and nilpotent groups.
\newblock {\em Acta Math.}, 137(3-4):247--320, 1976.

\bibitem[RT10]{RT10}
Michael Ruzhansky and Ville Turunen.
\newblock {\em Pseudo-differential operators and symmetries.}
\newblock Pseudo-Differential Operators. Theory and Applications, 2. Birkhauser
  Verlag, Basel, 2010.

\bibitem[Sau90]{SA}
Jean-Luc Sauvageot.
\newblock Quantum {D}irichlet forms, differential calculus and semigroups.
\newblock In {\em Quantum probability and applications, {V} ({H}eidelberg,
  1988)}, volume 1442 of {\em Lecture Notes in Math.}, pages 334--346.
  Springer, Berlin, 1990.

\bibitem[SC94]{SFC}
L.~Saloff-Coste.
\newblock Precise estimates on the rate at which certain diffusions tend to
  equilibrium.
\newblock {\em Math. Z.}, 217(4):641--677, 1994.

\bibitem[Spo08]{Her}
H.~Spohn.
\newblock On the boltzmann equation for weakly nonlinear wave equations.
\newblock In {\em Boltzmann's legacy}, ESI Lect. Math. Phys., 2008.

\bibitem[Wil13]{Wilde}
Mark~M. Wilde.
\newblock {\em Quantum Information Theory}.
\newblock Cambridge University Press, 2013.

\bibitem[Wir18]{Wirth}
Melchior Wirth.
\newblock A noncommutative transport metric and symmetric quantum markov
  semigroups as gradient flows of the entropy.
\newblock {\em arXiv:1808.05419}, 2018.

\bibitem[Wit74]{Wi}
Gerd Wittstock.
\newblock Ordered normed tensor products.
\newblock In {\em Foundations of quantum mechanics and ordered linear spaces
  ({A}dvanced {S}tudy {I}nst., {M}arburg, 1973)}, pages 67--84. Lecture Notes
  in Phys., Vol. 29. Springer, Berlin, 1974.

\bibitem[Yau96]{Ya4}
Horng-Tzer Yau.
\newblock Logarithmic {S}obolev inequality for lattice gases with mixing
  conditions.
\newblock {\em Comm. Math. Phys.}, 181(2):367--408, 1996.

\bibitem[Yau97]{Ya3}
Horng-Tzer Yau.
\newblock Logarithmic {S}obolev inequality for generalized simple exclusion
  processes.
\newblock {\em Probab. Theory Related Fields}, 109(4):507--538, 1997.

\end{thebibliography}

\end{document}